\documentclass[aps,pra,11pt,onecolumn,nofootinbib,tightenlines]{revtex4-2}

\usepackage{amsfonts, amsthm}
\usepackage{amssymb}
\usepackage{amsmath}  
\usepackage{mathtools}
\usepackage{graphicx, comment}
\usepackage{enumitem}

\usepackage{mathrsfs}
\usepackage[colorlinks=true, allcolors=blue]{hyperref}
\usepackage{cleveref}
\usepackage[dvipsnames]{xcolor}
\usepackage{fullpage}
\usepackage{braket, physics}
\usepackage{tikz}

\definecolor{myblue}{rgb}{0.078,0.176,0.706}
\definecolor{myred}{rgb}{0.922,0.039,0.039}
\definecolor{myyellow}{rgb}{1.000,0.843,0.000}

\newtheorem{theorem}{Theorem}

\newtheorem{lemma}[theorem]{Lemma}

\newtheorem{prop}[theorem]{Proposition}

\newtheorem{corollary}[theorem]{Corollary}
\newtheorem{assumption}{Assumption}

\theoremstyle{plain}
\newtheorem{remark}[theorem]{Remark}

\theoremstyle{definition} 
\newtheorem{example}[theorem]{Example}

\newcommand{\rmc}{{\rm{c}}}
\newcommand{\rms}{{\rm{s}}}
\newcommand{\rmd}{{\rm{d}}}
\newcommand{\rmsc}{{\rm{sc}}}

\newcommand{\cH}{\mathcal{H}}

\newcommand{\cP}{\mathcal{P}}

\newcommand{\cF}{\mathcal{F}}

\newcommand{\cM}{\mathcal{M}}

\newcommand{\diag}{\text{diag}}

\newcommand{\N}{\mathbb{N}}

\newcommand{\n}{(n)}

\DeclareMathOperator{\sinc}{sinc}

\DeclareMathOperator*{\argmin}{argmin}
\DeclareMathOperator*{\argmax}{argmax}
\DeclareMathOperator{\supp}{supp}
\def \Trm {\textup{Tr}} 
\def \d {\mathrm{d}} 
\def\cH{\mathcal{H}}
\def\ep{\varepsilon}
\def\bN{\mathbb{N}}
\def\bR{\mathbb{R}}

\def\ds{\mbox{ }\mbox{ }}

\def\sand{\widetilde D}
\def\sandq{\widetilde Q}

\newcommand{\reg}[1]{\prescript{\textnormal{reg}}{}{#1}} 
\newcommand{\Dreg}{\reg{\mathfrak{D}}}

\begin{document}

\title{\Large Additivity of quantum relative entropies\\ as a single-copy criterion} 

\author{Salman Beigi}
\affiliation{School of Mathematics, Institute for Research in Fundamental Sciences (IPM), P.O. Box 19395-5746, Tehran, Iran}

\author{Roberto Rubboli}
\affiliation{Centre for Quantum Technologies, National University of Singapore, Singapore 117543, Singapore}

\author{Marco Tomamichel}
\affiliation{Department of Electrical and Computer Engineering, College of Design and Engineering, National University of Singapore 117583, Singapore}
\affiliation{Centre for Quantum Technologies, National University of Singapore, Singapore 117543, Singapore}

\begin{abstract}
The fundamental goal of information theory is to characterize complex operational tasks using efficiently computable information quantities, Shannon's capacity formula being the prime example of this. However, many tasks in quantum information can only be characterized by regularized entropic measures that are often not known to be computable and for which efficient approximations are scarce. It is thus of fundamental importance to understand when regularization is not needed, opening the door to an efficiently computable characterization based on additive quantities. Here, we demonstrate that for a large class of problems, the question of whether regularization is needed or not can be determined at the single-copy level. Specifically, we demonstrate that regularization of the Umegaki relative entropy, along with related quantities such as the Petz and sandwiched relative entropies, is not needed if and only if a single-copy optimizer satisfies a certain property. These problems include hypothesis testing with arbitrarily varying hypotheses as well as quantum resource theories used to derive fundamental bounds for entanglement and magic state distillation. We derive the Stein, Chernoff, and Hoeffding exponents for these problems and establish necessary and sufficient conditions for their additivity, while also presenting partial results for the strong converse exponent.

\end{abstract}

\maketitle


\section{Introduction}
Additivity problems play a central role in quantum information theory, arising in a variety of contexts. 
As an example, the Holevo information of quantum channels had long been believed to be additive~\cite{holevo06additivity}, but counterexamples were ultimately found~\cite{hastings09additivity}. Nonetheless, it becomes additive for certain channels, including entanglement-breaking channels~\cite{shor2002additivity}, which implies that the classical capacity for such channels is efficiently computable.
Additivity problems are also ubiquitous in quantum resource theories~\cite{chitambar19resource}, for example, when characterizing how much entanglement can be distilled from a state~\cite{horodecki2000limits}. Many operational quantities are characterized by entropic quantities that appear in regularized form, defined in the asymptotic limit of many copies. However, such regularized measures are generally intractable as the underlying optimizations become increasingly complex with the number of copies. As a result, explicit evaluation of these quantities is typically possible only in cases where additivity holds, making regularization unnecessary. However, determining whether additivity holds is generally challenging.

A central example, and the main focus of this work, is the minimization of the Umegaki relative entropy over a convex and compact subset of positive operators. In particular, we are interested in the scenario where multiple copies of a state are considered. To this end, let us introduce a family of subsets of positive operators $\{\mathcal{F}_n\}_{n \in \mathbb{N}}$, where $\mathcal F_n$ consists of operators acting on the $n$-fold tensor product of a Hilbert space $\mathcal{H}$. Moreover, let $\rho_1, \dots, \rho_k$ be quantum states on $\mathcal{H}$, $m_1, \dots, m_k \in \mathbb{N}$ and define $M=\sum_{j=1}^k m_j$. We consider the optimization problem
\begin{equation} \mathfrak{D}\left(\bigotimes_{j=1}^k \rho_j^{\otimes m_j}\right) = \min_{\sigma^{(M)} \in \cF_{M}}D \left(\bigotimes_{j=1}^k \rho_j^{\otimes m_j}\middle\|\sigma^{(M)}\right) \,, \end{equation}
where $D(\rho \| \sigma) = \Tr[\rho(\log \rho - \log \sigma)]$ is the Umegaki relative entropy. A main question is to determine under which conditions additivity holds; that is, for any $m_1,\dots, m_k \in \mathbb{N}$,
\begin{equation}
\label{eq:strong-weak additivity}
\mathfrak{D}\left(\bigotimes_{j=1}^k \rho_j^{\otimes m_j}\right) =  \sum_{j=1}^k m_j\min_{\sigma^{(1)} \in \cF_{1}} D \left(\rho_j\middle\|\sigma^{(1)}\right) = \sum_{j=1}^k m_j\mathfrak{D}(\rho_j) \,.
\end{equation}
For a general family of sets $\{\cF_n\}_{n \in \mathbb{N}}$, the minimized relative entropy is generally non-additive. A prominent example is the relative entropy of entanglement, which fails to be additive~\cite{vollbrecht2001entanglement}. A similar phenomenon occurs in the resource theory of magic, where the set of stabilizer states also leads to non-additivity~\cite{veitch2014resource}. Nonetheless, additivity holds for some specific resource theories, such as the resource theory of coherence~\cite{winter2016operational} or for some specific family of states~\cite{zhu2010additivity,audenaert2002asymptotic,rubboli2024new}.

When $\rho_j=\rho$ for all $j=1,\dots,k$, the problem in~\eqref{eq:strong-weak additivity} reduces to weak additivity, which concerns the case of multiple identical copies of a single state.  An important and closely related quantity is the regularized relative entropy:
\begin{equation}
\label{additivity intro}
\Dreg(\rho) = \lim_{n \rightarrow \infty} \frac{1}{n} \mathfrak{D}\big(\rho^{\otimes n}\big),
\end{equation}
which plays a central role in several settings, including providing upper bounds on resource distillation rates~\cite{horodecki2013quantumness} and determining the Stein exponent in generalized quantum hypothesis testing~\cite{Brandao2,hayashi2024generalized,lami2025solution}. The regularized expression involves an optimization problem over the space $\cF_n$ together with a limit \(n \rightarrow +\infty\) which makes its computation generally challenging. In particular, the additivity problem in~\eqref{eq:strong-weak additivity} is directly relevant to the computation of the regularized quantity, as weak additivity implies that the regularized expression coincides with its single-copy form, making the problem tractable in many instances. 

A natural setting where additivity questions become relevant is generalized hypothesis testing, where the error exponents can be expressed in terms of regularized relative entropies. Consider a device that outputs either $n$ copies of a state $\rho$ (null hypothesis), or a state $\sigma^{\n}$ drawn from a closed and compact subset $\mathcal{F}_n$ of positive operators (alternative hypothesis). The goal is to distinguish between these two scenarios
\begin{equation}\label{eq:problem introduction}
\begin{aligned}
& \text{Null hypothesis: The state is }  \rho^{\otimes n} \\
    &\text{Alternative hypothesis: The state is some } \sigma^{(n)} \text{ belonging to } \cF_n.
\end{aligned}
\end{equation} 
The task is to perform a binary measurement $\{T_n,I-T_n\}$, where $0\leq T_n\leq I$ is a test operator, in order to distinguish between the two hypotheses. If the measurement outcome corresponds to $T_n$, we conclude that the null hypothesis is true; if the outcome corresponds to $I - T_n$, we decide that the alternative hypothesis is true. 
The type I and type II errors correspond, respectively, to the probability of incorrectly concluding that the state is some $\sigma^{\n} \in \cF_n$ while the true state is $\rho^{\otimes n}$, and the worst-case probability of incorrectly concluding that the state is $\rho^{\otimes n}$ when it is in fact a state $\sigma^{\n} \in \cF_n$.
We are interested in the asymptotic behavior of these errors as $n$ tends to infinity.

One of the most studied and arguably relevant scenarios in the literature is the asymmetric setting, also known as Stein's exponent, where the objective is to minimize the type II error under the constraint that the type I error vanishes in the asymptotic limit (see Section~\ref{sec: generalized hypothesis} for further details). This problem has been extensively studied across a wide range of settings.
In the classical case, the generalized hypothesis testing with arbitrarily varying sources was investigated in~\cite{FangweiShiyi1996hypothesis}, with a broader characterization of the Stein exponent later established in~\cite{Brandao+2020adversarial}.
In the quantum case, the authors of~\cite{Berta+2021composite} investigated the Stein exponent for the case where the alternative hypothesis is given by the convex hull of i.i.d.\ states. The authors prove that the Stein exponent equals the regularized relative entropy and show that, in general, such regularization cannot be avoided. 
In~\cite{hayashi2024generalized,lami2025solution}, the authors study a general framework in which the alternative hypothesis is represented by an arbitrary convex set of quantum states, subject to certain mild additional conditions. Finally, in~\cite[Theorem 33]{fang2024generalized}, the authors analyze the case where the hypothesis consists of a subset of positive operators satisfying an additional constraint on the corresponding polar sets, which we describe in more detail below.
These works collectively establish that, under various assumptions on the family of subsets $\{\cF_n\}_{n \in \mathbb{N}}$, the Stein exponent is characterized by the regularized relative entropy
\begin{align}\label{eq:Stein-composite results}
\rms(\rho\| \{\cF_n\}_{n\in \N}) = \Dreg(\rho)\,.
\end{align}
Therefore, the additivity property in~\eqref{eq:strong-weak additivity} is directly linked to the ability to express the Stein exponent in a single-letter formula, i.e., one involving only a single copy of the state $\rho$.

\bigskip

In this work, we investigate the additivity problems introduced above within a particular class of theories. Specifically, we focus on cases where the family of sets $\{\cF_n\}_{n \in \mathbb{N}}$
satisfies (in addition to other mild assumptions) a condition on its polar sets: the polar sets are closed under tensor products, meaning that 
\begin{equation}
    \cF_m^\circ \otimes \cF_n^\circ \subseteq \cF_{m+n}^\circ \qquad \forall m, n \in \mathbb{N}\,,
\end{equation}  
where the polar set of a convex set $\cF$ of hermitian operators is defined as
\begin{equation}
    \cF^{\circ} = \{X\geq 0:\Tr[XY]\leq 1, \forall Y \in \cF \}.
\end{equation}
This condition was introduced in~\cite{fang2024generalized} and earlier in~\cite{boreland2022information} in the context of additivity problems and, although somewhat restrictive, encompasses several important cases. These include hypothesis testing with arbitrarily varying sources, as well as sets used to approximate relevant measures in quantum resource theories, such as the Rains set, which approximates separable states, and the non-positive mana set used to approximate stabilizer states in odd dimensions. Within these theories, we derive several results concerning additivity and discuss their applications in the context of generalized hypothesis testing.

\begin{itemize}
    \item Our first main result (Theorem~\ref{th: strong additivity Umegaki}) is to establish necessary and sufficient conditions under which the additivity relation in equation~\eqref{eq:strong-weak additivity} holds. In particular, in Theorem~\ref{th: strong additivity Umegaki}, we prove that $\mathfrak{D}\big(\bigotimes_{j=1}^k \rho_j^{\otimes m_j}\big) = \sum_{j=1}^k m_j\mathfrak{D}(\rho_j)$ for all $m_j\in \mathbb{N}$, $j=1,\dots,k$ if and only if
    \begin{align}\label{eq:comment-8}
        \sup_{t \in \mathbb{R}, \tau \in \mathcal{F}_1}\textup{Tr}\Big[\tau \sigma_{0,j}^{-\frac{1}{2}(1+it)}\rho_j \sigma_{0,j}^{-\frac{1}{2}(1-it)}\Big] = 1, \qquad \forall j=1,\dots, k\,.
    \end{align}
Here, $\sigma_{0,j} \in \argmin_{\sigma \in \cF_1} D(\rho_j\|\sigma)$ is an optimizer for $\rho_j$. We further extend the result to states defined on Hilbert spaces of different dimensions, provided this is compatible with the structure of the theory under consideration.
\end{itemize}

Crucially, this property can be verified at the single-copy level, without the need to consider arbitrarily many copies, as it depends solely on whether the single-copy optimizers of $\rho_j$ satisfy a specific condition. We complement this discussion with several examples. In particular, within the resource theories of entanglement and magic, we show that certain symmetric states—such as Werner and isotropic states, which play a prominent role in entanglement manipulation tasks such as entanglement distillation—satisfy the required condition. As a result, their regularized relative entropy coincides with the single-copy expression. More interestingly, in the context of hypothesis testing with arbitrarily varying sources, we construct qubit examples for which additivity holds up to a certain number of copies (that can be arbitrarily large) but breaks down beyond that threshold. To the best of our knowledge, this constitutes the first explicit and elementary example exhibiting such behavior.

\begin{itemize}
    \item Our second main result is to apply the additivity conditions to error exponents. In particular, we establish a single-copy necessary and sufficient condition for when regularization is required for various error exponents. For the Stein exponent, we obtain that the regularization can be removed, i.e. $s(\rho\|\{\cF_n\}_{n \in \mathbb{N}}) =D(\rho\|\sigma_0)$, 
if and only if
\begin{equation}\label{eq:comment-9}
\sup_{t \in \mathbb{R}, \tau \in \mathcal{F}_1}\textup{Tr}\Big[\tau \sigma_0^{-\frac{1}{2}(1+it)}\rho \sigma_0^{-\frac{1}{2}(1-it)}\Big] = 1.
\end{equation}
Here, $\sigma_0 \in \argmin_{\sigma \in \mathcal{F}_1} D(\rho \| \sigma)$ is an optimizer for $\rho$. We derive expressions involving regularized quantities for the Chernoff and Hoeffding exponents and establish necessary and sufficient conditions under which the regularization can be removed. These conditions are also formulated at the single-copy level and must be verified at a saddle point. In addition, we present partial results concerning the strong converse exponent.
\end{itemize}

To establish the above results, we extend our additivity results to general $\alpha$-$z$-R\'enyi relative entropies, which include both the Petz and sandwiched relative entropies. Beyond the case of simple null and alternative hypotheses, the Chernoff, Hoeffding, and strong converse exponents in quantum hypothesis testing are known only in certain special cases, which are discussed in~\cite{Mosonyi+22errorexponent, HayashiTomamichel2016correlation}.

\begin{itemize}
    \item Our final main result (Proposition~\ref{prop: Additivity regularized}) establishes the strong additivity of the regularized relative entropy, i.e., we prove that 
  \begin{equation}
\label{eq:additivity regularized}
    \Dreg\left(\bigotimes_{j=1}^k \rho_j\right) = \sum_{j=1}^k \Dreg(\rho_j)\,.
\end{equation}
We further extend the result to states defined on Hilbert spaces of different dimensions, provided this is compatible with the structure of the theory under consideration.
\end{itemize}

Since both the Rains set and the set of states with non-positive mana satisfy the required assumptions, our results introduce two additional additive monotones that are based on the Umegaki relative entropy for the resource theories of entanglement and magic. In particular, the Rains set gives rise to a new additive—though unfaithful—entanglement monotone, expanding the family of known additive measures such as squashed entanglement~\cite{christandl2004squashed}, logarithmic negativity~\cite{vidal2002computable, plenio2005logarithmic}, and conditional entanglement of mutual information~\cite{yang2007conditional,yang2008additive}.
For the Rains set, we show that while additivity for the minimized Umegaki relative entropy does not generally hold, additivity holds when regularized. This regularized monotone is computable for small systems whenever regularization is unnecessary—a condition that can be checked using Theorem~\ref{th: strong additivity Umegaki}.
Indeed, the underlying sets admit an SDP representation, and the relative entropy itself can be approximated via semidefinite programming~\cite{fawzi2019semidefinite}. Finally, we note that the additivity of the regularized relative entropy of entanglement remains a long-standing open question since our setting does not include the set of separable states.

\section{Preliminaries}

\subsection{Quantum relative entropies}
In this section, we introduce quantum relative entropies, which will be the central focus of our study on additivity throughout this manuscript.

The Umegaki relative entropy plays a central role in the quantum hypothesis testing problem. Given quantum states $\rho$ and $\sigma$, it is defined as
\begin{equation}
D(\rho \| \sigma)  = \begin{cases}
 \text{Tr}[\rho (\log{\rho} - \log{\sigma})] & \text{if} \; \text{supp}(\rho) \subseteq \text{supp}(\sigma),\\
 +\infty & \text{else}.
\end{cases}
\end{equation}
Two other key quantities in our analysis are the Petz and sandwiched R\'enyi divergences, both of which play a central role in characterizing error exponents in quantum hypothesis testing. The \emph{Petz R\'enyi divergence}~\cite{petz1986quasi,Tomamichel} is defined as
\begin{equation}
\label{definition Petz}
D_{\alpha}(\rho \| \sigma) :=
\begin{cases}
\frac{1}{\alpha-1} \log{\Tr[\rho^\alpha \sigma^{1-\alpha}]}  & \text{if}\; (\alpha<1 \wedge \rho \not \perp \sigma) \vee \rho \ll \sigma, \\
 +\infty & \text{else}.
\end{cases}
\end{equation}
Also, the \textit{sandwiched quantum R\'enyi divergence}~\cite{muller2013quantum, wilde2014strong,Tomamichel} is defined as
\begin{equation}
\label{definition sandwiched}
\widetilde{D}_{\alpha}(\rho \| \sigma) :=
\begin{cases}
\frac{1}{\alpha-1}\log{\Tr\big[\big(\sigma^{\frac{1-\alpha}{2\alpha}}\rho \sigma^{\frac{1-\alpha}{2\alpha}}\big)^{\alpha}\big]}   & \text{if}\; (\alpha<1 \wedge \rho \not \perp \sigma) \vee \rho \ll \sigma, \\
 +\infty & \text{else}.
\end{cases}
\end{equation}

The Petz and sandwiched R\'enyi divergences are special cases within the broader family of $\alpha$-$z$-R\'enyi relative entropies. Although we do not provide an operational interpretation for these quantities, this parametrization provides a convenient and unifying framework for formulating our results in a general way, encompassing many divergences of interest.

The $\alpha$-$z$-R\'enyi relative entropies
are defined for \(\sigma\) and \(\rho\) as~\cite{audenaert2015alpha,zhang2020wigner}
\begin{equation}
\label{definition alpha-z}
D_{\alpha,z}(\rho \| \sigma):=
\begin{cases}
\frac{1}{\alpha-1}\log{\text{Tr}\big[\big(\rho^\frac{\alpha}{2z}\sigma^\frac{1-\alpha}{z}\rho^\frac{\alpha}{2z}\big)^z\big]} & \text{if}\; (\alpha<1 \wedge \rho \not \perp \sigma) \vee \rho \ll \sigma, \\
 +\infty & \text{else}.
\end{cases}
\end{equation}
In the following, we denote $Q_{\alpha,z}(\rho \| \sigma):=\exp \big( (\alpha-1)D_{\alpha,z}(\rho \| \sigma) \big)$.
When $z=1$, the $\alpha$-$z$-R\'enyi relative entropy reduces to the Petz R\'enyi relative entropy, while for $z=\alpha$, it reduces to the sandwiched R\'enyi relative entropy.
Finally, it reduces to Umegaki's relative entropy in the limit of $\alpha\to 1$ (when $z \neq 0$). 

In the following, we consider quantum relative entropies minimized over a convex and compact subset of positive semidefinite operators. Denoting this set by 
$\cF$, we use calligraphic letters to represent the corresponding minimized quantities:
\begin{align}
&\mathfrak{D}(\rho) = \min_{\sigma \in \cF}D(\rho\|\sigma) \\
&\mathfrak{D}_{\alpha,z}(\rho) = \min_{\sigma \in \cF}D_{\alpha,z}(\rho\|\sigma) \,.
\end{align}
Finally, the regularized Umegaki relative entropy is defined as
\begin{equation}
    \Dreg(\rho) = \lim_{n \rightarrow \infty} \frac{1}{n}\mathfrak{D}\big(\rho^{\otimes n}\big) \,.
\end{equation}

\subsection{Polar set and support function}
\label{sec:polar set}

As mentioned in the introduction, our main objective is to study the additivity properties of quantum relative entropies optimized over subsets of positive operators. In this section, we formally introduce the key assumptions imposed on these sets, particularly the condition on their polar sets, which will play a central role throughout the manuscript.

First, we introduce the primary notation for sets. Given two sets of operators $\cF, \cF'$, their tensor product is defined as
\begin{equation}
    \cF\otimes \cF'=\{X\otimes X':\, X\in \cF, X'\in \cF'\}.
\end{equation}
The polar set of a convex set $\cF$ of hermitian operators is defined as\footnote{Here, we indeed consider the intersection of the polar set with the set of positive semidefinite operators and our notation slightly deviates from that of~\cite{fang2024generalized}, which denotes this set by $\cF^\circ_+$.}
\begin{equation}
\cF^{\circ} = \{X\geq 0:\Tr[XY]\leq 1, \forall Y \in \cF \} \,.
\end{equation}
Next, we consider a family of subsets of positive operators, denoted by $\{\cF_n\}_{n \in \mathbb{N}}$, where $\cF_n$ consists of operators acting on $\mathcal H^{\otimes n}$ where $\mathcal H$ is a finite-dimensional Hilbert space. The first assumption we impose on this family is as follows:
\begin{assumption}\label{assumption:polar}
The family of sets $\{\cF_n\}_{n\in \N}$ satisfies the following properties:
\begin{enumerate}[label=(1.\Alph*), itemsep=1pt, topsep=1pt]
\medskip
    \item Each $\cF_n$ is convex and compact;
    \item \label{A.1.B} $\cF_m \otimes \cF_n \subseteq \cF_{m+n}$,  for all $m, n \in \mathbb{N};$
    \item \label{A.4-polar} $\cF_m^\circ \otimes \cF_n^\circ \subseteq \cF_{m+n}^\circ$,  for all  $m, n \in \mathbb{N}$. 
\end{enumerate}
\end{assumption}

We now introduce additional assumptions that will be essential in various parts of our analysis. The first asserts that, for the given state $\rho$, the set $\cF_1$ contains at least one state whose support includes that of $\rho$. This condition ensures, for instance, the continuity of the minimized divergences over the set $\cF_1$ (see Appendix~\ref{app:continuity}), and it also guarantees the strong convexity of the reliability functions associated with error exponents with respect to certain parameters (see Section~\ref{sec:Chernoff}).
\begin{assumption}\label{assumption:support}
There exists $\sigma\in \cF_1$ such that the support of $\sigma$ includes that of $\rho$.
\end{assumption}

Finally, we sometimes impose the additional assumption that each set $\cF_n$ is invariant under permutations of the subsystems. In other words, if $\sigma^{\n} \in \cF_n$, then any state $\sigma'^{\n}$ obtained by permuting the subsystems in $\mathcal{H}^{\otimes n}$ also belongs to $\cF_n$.

\begin{assumption}\label{assumption:perm-inv}
 Each $\cF_n$ is closed under permutations.
\end{assumption}

In the study of strong additivity, we also consider pairs of states defined on possibly different Hilbert spaces. This situation arises frequently in quantum information theory, with a
prominent example being entanglement theory. Consider bipartite states
$\rho_1, \rho_2$ defined on (possibly different) finite-dimensional Hilbert spaces $\mathcal H_{AB}, \mathcal H_{A'B'}$, respectively. The subsystems $A$ and $A'$ are assumed to be held by the same
party, and likewise for $B$ and $B'$.
Under this assumption, the tensor
product state $\rho_1 \otimes \rho_2$ is naturally regarded as a
bipartite state shared between the composite parties $AA'$ and $BB'$. 
When families of states indexed by different Hilbert spaces are considered, we make the underlying Hilbert spaces explicit in the
notation. We use the notation $\mathcal{F}(\mathcal{H})$ to denote subsets of positive semidefinite operators on the Hilbert space $\mathcal{H}$.
We introduce the following additional assumption.
\begin{assumption}
\label{ass:tensor-compatibility}
The family of sets
$\bigl\{\,\mathcal{F}\bigl(\bigotimes_{j=1}^k \mathcal{H}_j^{\otimes m_j}\bigr)\,\bigr\}_{(m_1,\dots,m_k)\in\mathbb{N}^k}
$ satisfies the following properties:
\begin{enumerate}[label=(4.\Alph*), itemsep=1pt, topsep=1pt]
\medskip
 \item Each $\mathcal{F}\bigl(\bigotimes_{j=1}^k \mathcal{H}_j^{\otimes m_j}\bigr)$ is convex and compact;
\item \label{4B-prime} $\bigotimes_{j=1}^k\mathcal{F}\bigl( \mathcal{H}_j^{\otimes m_j}\bigr) \subseteq \mathcal{F}\bigl(\bigotimes_{j=1}^k \mathcal{H}_j^{\otimes m_j}\bigr)$ for all $(m_1,\dots,m_k)\in\mathbb{N}^k$;
\item \label{4C-prime}
$\bigotimes_{j=1}^k\mathcal{F}\bigl( \mathcal{H}_j^{\otimes m_j}\bigr)^\circ \subseteq \mathcal{F}\bigl(\bigotimes_{j=1}^k \mathcal{H}_j^{\otimes m_j}\bigr)^\circ$ for all $(m_1,\dots,m_k)\in\mathbb{N}^k$;
\item There exists $\sigma\in \cF(\mathcal{H}_j)$ such that the support of $\sigma$ includes that of $\rho_j$ for all $j=1,\dots,k$.
\item Each $\mathcal{F}\bigl(\bigotimes_{j=1}^k \mathcal{H}_j^{\otimes m_j}\bigr)$ is invariant under permutations of the $m_j$ tensor factors corresponding to each $\mathcal{H}_j$.
\end{enumerate}
\end{assumption}

While most of the above assumptions are standard in the literature on generalized quantum hypothesis testing, the assumptions concerning the polar sets in Assumptions~\ref{assumption:polar} and~\ref{ass:tensor-compatibility} are less common. They first appeared in the context of composite hypothesis testing in~\cite[Assumption~25]{fang2024generalized}. Equivalently, these conditions can be expressed as the submultiplicativity of the so-called support functions, which will serve as one of our main tools throughout the manuscript for establishing additivity results. The support function associated with a convex set $\cF$, evaluated on a positive semidefinite operator $X$, is defined by
\begin{equation}
h_{\cF}(X) = \sup_{\sigma \in \cF}\Tr(\sigma X) \,.
\end{equation}
As shown in~\cite[Lemma 8]{fang2024generalized}, the condition on the polar set~\ref{A.4-polar} is related to the submultiplicativity of the corresponding support functions.

\begin{lemma}
\label{lemma:multiplicativity of support function}
Let $\mathcal{H}_{j}$, $j=1,\dots,k$ be finite-dimensional Hilbert spaces. Let $\cF(\mathcal{H}_j)$ and $\cF(\bigotimes_{j=1}^k\mathcal{H}_j)$ be sets of positive semidefinite operators on the corresponding Hilbert spaces. Their polar sets are closed under tensor product if and only if their support functions are submultiplicative, i.e.,
\begin{align}
\label{eq:submultiplivativity support function}
\bigotimes_{j=1}^k\mathcal{F}(\mathcal{H}_j)^\circ
\subseteq \mathcal{F}\left(\bigotimes_{j=1}^k\mathcal{H}_j\right)^\circ
\quad \Longleftrightarrow \quad
h_{\mathcal{F}\left(\bigotimes_{j=1}^k\mathcal{H}_j\right)}\left(\bigotimes_{j=1}^k X_{j}\right) \le \prod_{j=1}^k h_{\cF(\mathcal{H}_j)}(X_{j}), \quad \forall X_{j}\geq 0.
\end{align}
If, moreover, $\bigotimes_{j=1}^k\mathcal{F}(\mathcal{H}_j)
\subseteq \mathcal{F}(\bigotimes_{j=1}^k\mathcal{H}_j)$, then
\begin{align}
\label{eq:multiplivativity support function}
\bigotimes_{j=1}^k\mathcal{F}(\mathcal{H}_j)^\circ
\subseteq \mathcal{F}\left(\bigotimes_{j=1}^k\mathcal{H}_j\right)^\circ
\quad \Longleftrightarrow \quad
h_{\mathcal{F}(\bigotimes_{j=1}^k\mathcal{H}_j)}\left(\bigotimes_{j=1}^k X_{j}\right) = \prod_{j=1}^k h_{\cF(\mathcal{H}_j)}(X_{j}), \quad \forall X_{j}\geq 0.
\end{align}
\end{lemma}
\begin{proof}
The first statement in~\eqref{eq:submultiplivativity support function} is a straightforward generalization of the result in~\cite[Lemma~8]{fang2024generalized}. For completeness, we provide a proof based on the dual program of the support function~\cite[Lemma~5]{fang2024generalized} stating that $h_{\cF}(\omega) = \inf\{t > 0: \omega \in t \cF^\circ\}$.

($\Rightarrow$ in~\eqref{eq:submultiplivativity support function}) 
Let $t_j$ for $j=1,\dots,k$ be feasible solutions of the dual programs for $h_{\cF(\mathcal{H}_j)}(X_j)$. 
Then, we have $X_j \in t_j \cF(\mathcal{H}_j)^\circ$. 
This implies 
\begin{align}
    \bigotimes_{j=1}^k X_j \in  \bigotimes_{j=1}^k \left(t_j\cF(\mathcal{H}_j)^\circ\right)  \subseteq \left (\prod \limits_{j=1}^k t_j \right ) \mathcal{F}\left(\bigotimes_{j=1}^k\mathcal{H}_j\right)^\circ,
\end{align}
and therefore, $\prod_{j=1}^k t_j$ is also a feasible solution for the dual program for $h_{\mathcal{F}(\bigotimes_{j=1}^k\mathcal{H}_j)}(\bigotimes_{j=1}^k X_{j})$. Hence, $h_{\mathcal{F}(\bigotimes_{j=1}^k\mathcal{H}_j)}(\bigotimes_{j=1}^k X_{j}) \leq \prod_{j=1}^k t_j$. As this holds for any feasible solutions $t_1, \dots, t_k$, we conclude that 
\begin{align}
    h_{\mathcal{F}(\bigotimes_{j=1}^k\mathcal{H}_j)}\left(\bigotimes_{j=1}^k X_{j}\right) \leq  \prod_{j=1}^k h_{\cF(\mathcal{H}_j)}(X_{j}) \,.
\end{align}

($\Leftarrow$ in~\eqref{eq:submultiplivativity support function}) For any $X_j \in \cF(\mathcal{H}_j)^\circ$, we have $0 \leq h_{\cF(\mathcal{H}_j)}(X_j) \leq 1$. Then, the submultiplicativity assumption implies that $h_{\mathcal{F}(\bigotimes_{j=1}^k\mathcal{H}_j)}(\bigotimes_{j=1}^k X_{j}) \leq 1$ and therefore, $\bigotimes_{j=1}^k X_{j} \in \mathcal{F}(\bigotimes_{j=1}^k\mathcal{H}_j)^\circ$. This shows that $\bigotimes_{j=1}^k\mathcal{F}(\mathcal{H}_j)^\circ
\subseteq \mathcal{F}(\bigotimes_{j=1}^k\mathcal{H}_j)^\circ$.

To establish multiplicativity under the condition $\bigotimes_{j=1}^k\mathcal{F}(\mathcal{H}_j)
\subseteq \mathcal{F}(\bigotimes_{j=1}^k\mathcal{H}_j)$, it remains to prove the reverse inequality
$h_{\mathcal{F}(\bigotimes_{j=1}^k\mathcal{H}_j)}(\bigotimes_{j=1}^k X_{j}) \geq \prod_{j=1}^k h_{\cF(\mathcal{H}_j)}(X_{j})$. 
This inequality is immediate if $ h_{\cF(\mathcal{H}_j)}(X_{j})=0$ for some $j$. Otherwise, 
for any sufficiently small $\varepsilon > 0$, there exist positive semidefinite operators $\sigma_{j}\in \cF(\cH_j)$ such that
$\Tr(\sigma_{j} X_{j}) \geq  h_{\cF(\mathcal{H}_j)}(X_{j}) - \varepsilon> 0$ for all $j$.
The condition $\bigotimes_{j=1}^k\mathcal{F}(\mathcal{H}_j)
\subseteq \mathcal{F}(\bigotimes_{j=1}^k\mathcal{H}_j)$ guarantees that $\bigotimes_{j=1}^k \sigma_j$ is a feasible point, so we have
\begin{equation}
h_{\mathcal{F}(\bigotimes_{j=1}^k\mathcal{H}_j)}\left(\bigotimes_{j=1}^k X_{j}\right) \ge \Tr(\bigotimes_{j=1}^k(\sigma_{j}X_{j} )) \geq \prod_{j=1}^k \big(h_{\cF(\mathcal{H}_j)}(X_{j}) - \varepsilon\big).
\end{equation}
Taking the limit $\varepsilon \to 0$ completes the argument.
\end{proof}

\subsection{Examples of sets satisfying the assumptions}
We now provide several examples of sets that satisfy Assumption~\ref{assumption:polar} given in the previous section.

\begin{example}[Arbitrarily varying sources] \label{example:AV}
Fix $\cF_1=\cF$ to be a closed convex set of quantum states over a finite-dimensional Hilbert space. Let $\cF_n$ be the convex hull of $\cF^{\otimes n}$, i.e., $\cF_n$ is the convex hull of product states in $\cF$:
\begin{equation}
    \cF_n=\text{conv}\left\{\bigotimes_{i=1}^n\sigma_i: \sigma_1, \dots,\sigma_n \in \cF\right\} \,.
\end{equation}
Then, the collection $\{\cF_n\}_{n\in \N}$ satisfies both Assumption~\ref{assumption:polar} and Assumption~\ref{assumption:perm-inv}. In particular, it satisfies condition~\ref{A.4-polar}. This is because any linear function on $\cF_n$ takes its maximum on its extreme points, and by definition, the extreme points of $\cF_n$ are tensor product quantum states. Thus, the corresponding support functions are multiplicative, and by Lemma~\ref{lemma:multiplicativity of support function}, condition~\ref{A.4-polar} is met.    
\end{example}

\begin{example}[Rains set] \label{example:Rains}
The Rains set, which contains the set of separable states as a subset, is a relaxation introduced by Rains to obtain an upper bound on distillable entanglement~\cite{rains1999bound}. Given a composite system $AB$,
the Rains set is a subset of positive operators on the composite system defined by:
\begin{equation}
\mathrm{Rains}(A : B) := \left\{ \sigma_{AB} \geq 0 :\, \big\|\sigma^{\top_B}\big\|_1 \leq 1 \right\},
\end{equation}
where $\sigma_{AB}^{\top_B}$ denotes the partial transpose of $\sigma_{AB}$. Letting $\cF_n=\mathrm{Rains}(A^n: B^n)$, the collection $\mathrm{Rains}=\{\cF_n\}_{n\in \N}$ satisfies Assumptions~\ref{assumption:polar},~\ref{assumption:support}, and~\ref{assumption:perm-inv}. Crucially, it satisfies the condition on the polar sets~\ref{A.4-polar}~\cite{fang2024generalized}, unlike the set of separable states (SEP) or the set of states with a positive partial transpose (PPT), which do not~\cite{zhu2010additivity}. To prove this fact, we can use the SDP duality to write
\begin{equation}
h_{\mathrm{Rains}(A:B)}(X_{AB}) = \sup_{\sigma_{AB} \in \mathrm{Rains}(A:B)} \mathrm{Tr}[X_{AB} \sigma_{AB}] = \inf_{\gamma_{AB} \geq X_{AB}} \|\gamma_{AB}^{\top_B}\|_{\infty}, 
\end{equation}
where $\|\cdot\|_{\infty}$ is the spectral norm. Then, by the multiplicativity of $\|\cdot\|_{\infty}$, we can easily verify that the support functions for the Rains sets are multiplicative under the tensor product. By Lemma~\ref{lemma:multiplicativity of support function} this is equivalent to condition~\ref {A.4-polar}. The rest of the conditions are immediate. The Rains set also satisfies Assumption~\ref{ass:tensor-compatibility} when considering tensor products of different finite-dimensional Hilbert spaces.

The Rains set is particularly useful to bound the set of separable states and hence can be used to provide bounds on the relative entropy of entanglement. 
In particular, for a bipartite state $\rho$, we have the bounds
\begin{equation}
\label{bound entanglement set}
\min_{\sigma \in \mathrm{SEP} }D(\rho \|\sigma) \geq \min_{\sigma \in \mathrm{PPT} }D(\rho \| \sigma) \geq \min_{\sigma \in \mathrm{Rains} } D(\rho \| \sigma).
\end{equation}
Indeed, we have the inclusion $\mathrm{SEP} \subseteq \mathrm{PPT} \subseteq \mathrm{Rains}$.
Also, the same holds for the regularized versions, so the Rains set can be used to bound the regularized relative entropy of entanglement. The latter quantity is related to the error exponent in distinguishing a quantum state from the set of separable states~\cite{Mosonyi+22errorexponent, HayashiTomamichel2016correlation}.
\end{example}

\begin{example}[Subnormalized states with non-positive mana] \label{example:mana}
Inspired by the Rains bound, in~\cite{wang2020efficiently} the authors relaxed the set of stabilizer states to the set of subnormalized states with non-positive mana:
\begin{equation}
\mathcal{W}_n=\{\sigma  \geq 0: \|\sigma \|_{W,1} \leq 1\} \,,
\end{equation}
where $\| \cdot\|_{W,1}$ denotes the Wigner trace norm and the index $n$ represents the number of underlying subsystems. The set $\mathcal W_n$ is convex~\cite[Appendix B]{wang2020efficiently}. Moreover, similarly to the case of the Rains set, the family of sets $\{\mathcal{W}_n\}_{n\in \mathbb N}$ satisfies the assumption on the polar sets~\cite{fang2024generalized}. This is because we have, by the SDP duality, that
\begin{equation}
h_{\mathcal{W}_n}(X) = \sup_{\sigma \in \mathcal{W}_n} \mathrm{Tr}[X \sigma] = \inf_{\gamma \geq X} \|\gamma\|_{W,\infty}, 
\end{equation}
where $\|\cdot\|_{W,\infty}$ is the Wigner spectral norm. Then, by the multiplicativity of $\|\cdot\|_{W, \infty}$, we can easily check that $h_{\mathcal{W}_n}(\cdot)$ is multiplicative under tensor product, which by Lemma~\ref{lemma:multiplicativity of support function} is equivalent to the assumption~\ref{A.4-polar} on the polar sets.
The rest of the assumptions can also be easily verified. 

The relative entropy minimized over the set $\mathcal{W}_n$ is known as thauma and, along with the regularized thauma, can be used to provide bounds on magic state distillation. Since for the family of stabilizer states, which we denote by $\{{\rm{STAB}}_n\}_{n\in \N}$, we have the inclusion ${\rm{STAB}}_n \subseteq \mathcal{W}_n$, we obtain the bound
\begin{equation}
\min_{\sigma \in \mathrm{STAB}_1}D(\rho\|\sigma) \geq \min_{\sigma \in \mathcal{W}_1}D(\rho\|\sigma) \,,
\end{equation}
and the same holds for the regularized quantities. 
\end{example}

For other examples of sets satisfying Assumption~\ref{assumption:polar} or Assumption~\ref{ass:tensor-compatibility}  we refer to~\cite{fang2024generalized}.

\subsection{Generalized quantum hypothesis testing and Stein's lemma}
\label{sec: generalized hypothesis}
The additivity conditions we derive below find natural applications in the framework of generalized hypothesis testing. To this end, this section is dedicated to formally introducing the quantum hypothesis testing setting with a generalized alternative hypothesis.

Consider a device that outputs either $n$ copies of a state $\rho$ (null hypothesis), or a state $\sigma^{(n)}$ drawn from a closed and compact subset $\mathcal{F}_n$ of positive operators (alternative hypothesis). The goal is to distinguish between these two scenarios.

 In particular, our hypothesis testing problem is defined by a pair \((\rho, \{\mathcal{F}_n\}_{n\in \N} )\), where \(\rho\) is a quantum state on a finite-dimensional Hilbert space $\mathcal H$, and \(\mathcal{F}_n\) for any $n\in \N$ is a convex and compact subset of positive operators  $\mathcal H^{\otimes n}$.
The two hypotheses are defined as follows:
\begin{equation}\label{eq:problem}
\begin{aligned}
& \text{Null hypothesis: The state is }  \rho^{\otimes n} \\
    &\text{Alternative hypothesis: The state is some } \sigma^{(n)} \text{ belonging to } \cF_n.
\end{aligned}
\end{equation} 
The task is to perform a binary measurement $\{T_n,I-T_n\}$, where $0\leq T_n\leq I$ is a test operator, in order to distinguish between the two hypotheses. If the measurement outcome corresponds to $T_n$, we conclude that the null hypothesis is true; if the outcome corresponds to $I - T_n$, we decide that the alternative hypothesis is true. The Type I and Type II errors associated with this test are
 \begin{align}
 & \text{Type I error:} \quad \; \alpha_n(\rho|T_n):=\operatorname{Tr}\big[(I - T_n)\rho^{\otimes n}\big] ,\label{eq:type-I-error}\\
& \text{Type II error:} \quad \beta_n(\cF_n|T_n):=\sup_{\sigma^{\n} \in \cF_n} \operatorname{Tr}\big[\sigma^{\n}T_n\big] \label{eq:type-II-error} \,.
 \end{align}
These correspond, respectively, to the probability of incorrectly concluding that the state is a state $\sigma^{\n} \in \cF_n$ when the true state is $\rho^{\otimes n}$, and the worst-case probability of incorrectly concluding that the state is $\rho^{\otimes n}$ when it is in fact a state $\sigma^{\n} \in \cF$.
 We are interested in the asymptotic behavior of these errors when $n$ tends to infinity.

\bigskip

The trade-off between the type I and the type II error probabilities has been studied in various settings. One of the most explored and arguably relevant settings in the literature is the asymmetric case, where the objective is to minimize the type II errors while assuming that the type I error vanishes in the asymptotic limit.
To this end, the \emph{Stein exponent} is defined as
\begin{align}\label{eq:def-Stein-exp}
&\rms(\rho\| \{\cF_n\}_{n\in \N}) \nonumber\\
&:= \sup\left\{r\geq 0:\, \exists (T_n)_{n\in \N}, \, 0\leq T_n\leq I,\,   \liminf_{n\to +\infty} -\frac 1n\log \beta_n(\cF_n|T_n)\geq r,\,  \lim_{n\to +\infty} \alpha_n(\rho| T_n)=0  \right\}.
\end{align}
The Stein exponent, 
when each set in the alternative hypothesis \(\{\mathcal{F}_n\}_{ n\in \N}\) consists of a single element, i.e., \(\mathcal{F}_n = \{\sigma^{\otimes n}\}\), is according to the quantum Stein lemma~\cite{HiaiPetz1991proper, OgawaNagaoka00, audenaert+08asymptotic} equal to
\begin{align}\label{eq:Stein}
\rms(\rho\| \sigma) = D(\rho\| \sigma),
\end{align}
where $D(\rho\|\sigma)$ is the Umegaki relative entropy.
The generalized Stein’s lemma has been established in various settings and under broad conditions. For the purposes of this work, we state the result in the specific case where the assumption on the polar sets is met. A proof for this setting can be found in~\cite{fang2024generalized}.

\begin{theorem}[\cite{fang2024generalized}]\label{thm:Stein}
Consider the quantum hypothesis testing problem~\eqref{eq:problem} for some state $\rho$ and suppose that the alternative hypothesis $\{\cF_n\}_{n\in \N}$ satisfies Assumption~\ref{assumption:polar}, Assumption~\ref{assumption:support}, and Assumption~\ref{assumption:perm-inv}. 
Then, the Stein exponent is 
\begin{align}\label{eq:Stein-composite}
\rms(\rho\| \{\cF_n\}_{n\in \N}) = \lim_{n\to +\infty} \frac 1n \min_{\sigma^{\n}\in \cF_n}\,  D\big(\rho^{\otimes n}\big\| \sigma^{(n)}\big) = \Dreg(\rho)\,.
\end{align}
where $\rms(\rho\| \{\cF_n\}_{n\in \N})$ is defined in~\eqref{eq:def-Stein-exp} in terms of type I and type II errors given in~\eqref{eq:type-I-error} and~\eqref{eq:type-II-error}.
\end{theorem}
Observe that Assumption~\ref{assumption:support} ensures that the minimum relative entropy on the right-hand side of~\eqref{eq:Stein-composite} is finite and grows at most linearly.

\section{Additivity for the Umegaki relative entropy}
\label{sec:Additivity of Umegaki}
In this section, we establish several additivity results for the Umegaki relative entropy. In particular, we provide necessary and sufficient conditions for additivity, formulated in terms of properties of the single-copy optimizer.

Our approach is based on the fundamental property that a point is a minimizer of a convex function if and only if its derivative at that point is non-negative in all feasible directions. This condition is formulated using the framework of Fr\'echet derivatives (see~\cite{bhatia2013matrix} for a comprehensive review), and first appeared in the context of entanglement theory in~\cite{vedral1997quantifying}. It has since been further developed in subsequent works~\cite{friedland2011explicit,girard2014convex,rubboli2024new}.

\subsection{Weak-additivity of the Umegaki relative entropy}
In this section, we establish a weak additivity result for the Umegaki relative entropy, corresponding to the case of multiple copies of the same state. 

We first derive a necessary and sufficient condition that must be satisfied by the minimizer of the Umegaki relative entropy. This result combines ideas from~\cite[Theorem 1]{friedland2011explicit} (see also~\cite[Theorem 4]{rubboli2024new}), together with the integral representation of the derivative of the logarithm established in~\cite[Lemma 3.4]{Sutter+2017multivariate}. Let us define the positive operator
\begin{equation}\label{eq:comment-37}
\Xi(\rho,\sigma) = \int_{-\infty}^{+\infty} \sigma^{-\frac{1}{2}(1+it)}\rho \sigma^{-\frac{1}{2}(1-it)} \beta_0(t)\d t \,,
\end{equation}
where $\beta_0(t) = \frac{\pi}{2(\cosh(\pi t) +1)}$ is a probability density.
Here, negative powers are taken in the sense of generalized inverses. Then, we have
\begin{lemma}
\label{necessary and sufficient Umegaki}
    Let $\rho$ be a quantum state, and let $\cF$ be a convex and compact subset of positive operators that includes at least one state whose support contains the support of $\rho$. Then $\sigma_0 \in \argmin_{\sigma \in \mathcal{F}} D(\rho \| \sigma)$ if and only if $\supp(\rho) \subseteq \supp(\sigma_0)$ and
\begin{align}\label{eq:derivative-condition}
\sup_{\tau\in \cF} \textup{Tr}\big[\tau\, \Xi(\rho,\sigma_0)\big] =  1\,.
\end{align}
\end{lemma}
\begin{proof}
We note that by the assumption on the supports, we have $\min_{\sigma\in \cF}D(\rho\| \sigma)<+\infty$ and hence $\supp(\rho)\subseteq \supp(\sigma_0)$.
Since the Umegaki relative entropy is convex in the second argument, the condition $\sigma_0\in \argmin_{\sigma \in \mathcal{F}} D(\rho \| \sigma)$ is equivalent to (see, e.g.,~\cite[Theorem 1]{friedland2011explicit} or~\cite[Theorem 4]{rubboli2024new})
\begin{align}\label{eq:first-der-condition}
\frac{\d}{\d x} D\big(\rho\big\| (1-x)\sigma_0+x\tau)\big) \Bigg\vert_{x=0}\geq 0, \qquad \qquad \forall \tau\in \cF.
\end{align}
We consider the case where the support of $\tau$ is included in the support of $\sigma_0$.
In the case where this does not hold, we use~\cite[Section VII]{friedland2011explicit}; see also Appendix~\ref{app:derivative}.
The derivative on the left-hand side can be computed using the following identity
\begin{align}\label{eq:logx-logy}
\frac{\log{x}-\log{y}}{x-y} & = 
 \int_0^\infty \frac{1}{(x+\lambda)(y+\lambda)} \dd \lambda = \int_{-\infty}^{\infty} x^{-\frac{1}{2}(1+it)}y^{-\frac{1}{2}(1-it)} \beta_{0}(t) \dd t,
\end{align}
where $\beta_0(t) = \frac{\pi}{2(\cosh(\pi t) +1)}$ is a probability density. Here, the first equality is standard and can be verified by direct integration. The second equality is proven in~\cite[Eq.~(96)]{Sutter+2017multivariate}. Another proof of this equality is given in Lemma~\ref{lem:derivative-integral} below. We can then use these equalities to compute the Frech\'et derivative of the log function. To this end, let $A= \sum_k \mu_k \ketbra{k}{k}$ be the spectral decomposition of a positive definite operator $A$. Then, $(A+\lambda)^{-1} = \sum_k (\mu_k+\lambda)^{-1}\ketbra{k}{k}$ implies  
\begin{align}
D \log A[H] &= \lim_{s\to 0} \frac{\log (A+sH)-\log A}{s}  = \int_0^{\infty} (A+\lambda)^{-1}H(A+\lambda)^{-1}\dd \lambda \\
& =  \sum_{k, \ell }\int_0^{\infty} (\mu_k+\lambda)^{-1}(\mu_\ell +\lambda)^{-1} \ketbra{k}{k} H \ketbra{\ell}{\ell}\dd \lambda\\
& = \sum_{k, \ell }\int_{-\infty}^{+\infty}  \mu_k^{-\frac{1}{2}(1+it)}  \mu_{\ell}^{-\frac{1}{2}(1-it)}   \beta_{0}(t)  \ketbra{k}{k} H \ketbra{\ell}{\ell}\dd t\\
& =\int_{-\infty}^{+\infty}  A^{-\frac{1}{2}(1+it)} H A^{-\frac{1}{2}(1-it)} \beta_{0}(t) \dd t.
\end{align}
Here, the first line is standard, for the third line we use~\eqref{eq:logx-logy} and the last line follows from $A^{-\frac{1}{2}(1+it)} = \sum_k \mu_k^{-\frac{1}{2}(1+it)}\ketbra{k}{k}$. This equation has first been proven in~\cite[Lemma 3.4]{Sutter+2017multivariate}.
Therefore, we have 
\begin{align}
\frac{\d}{\d x} D\big(\rho\big\| (1-x)\sigma_0+x\tau)\big) \Bigg\vert_{x=0} & = -\frac{\d}{\d x} \textup{Tr}\Big[\rho \log((1-x)\sigma_0+x\tau)\Big]\Bigg\vert_{x=0} \label{eq:comment-41} \\
&= \int_{-\infty}^{+\infty} \Tr\Big[\rho\sigma_0^{-\frac{1}{2}(1+it)}(\sigma_0 -\tau)  \sigma_0^{-\frac{1}{2}(1-it)}\Big] \beta_0(t)\d t \label{eq:comment-42}\\
& = 1- \int_{-\infty}^{+\infty} \Tr\Big[\rho\sigma_0^{-\frac{1}{2}(1+it)}\tau \sigma_0^{-\frac{1}{2}(1-it)}\Big] \beta_0(t)\d t.
\end{align}
Therefore, the condition $\sigma_0\in \argmin_{\sigma \in \mathcal{F}} D(\rho \| \sigma)$ is equivalent to 
\begin{equation}
\int_{-\infty}^{+\infty} \Tr\Big[\rho\sigma_0^{-\frac{1}{2}(1+it)}\tau \sigma_0^{-\frac{1}{2}(1-it)}\Big] \beta_0(t)\d t \leq 1 \,\qquad \qquad \forall \tau\in \cF.
\end{equation}
Here, the powers are understood as generalized powers and hence taken only on the support of the states. Since the inequality must hold for all $\tau$, it is equivalent to consider the supremum over $\tau$. Furthermore, equality is attained at $\tau=\sigma_0$.

\end{proof}

The next result provides a necessary and sufficient single-copy criterion in the context of weak additivity.
\begin{theorem}
\label{thm:Stein-additivity}
Suppose that $\rho$ and $\{\cF_n\}_{n\in \N}$ satisfy Assumption~\ref{assumption:polar}, and Assumption~\ref{assumption:support}.
 Moreover, let $\sigma_0 \in \argmin_{\sigma \in \mathcal{F}_1} D(\rho \| \sigma)$. Then, 
\begin{equation}
\label{eq:additivity Umegaki}
\min_{\sigma^{\n} \in \mathcal{F}_n} D\big(\rho^{\otimes n} \big\| \sigma^{\n}\big) = n\min_{\sigma \in \mathcal{F}_1} D\big(\rho \big\| \sigma\big) \,, \quad \forall n \geq 1 ,
\end{equation}
if and only if
\begin{equation}
\label{eq:additivity-Stein-condition weak additivity}
\sup_{t \in \mathbb{R}, \tau \in \mathcal{F}_1}\textup{Tr}\Big[\tau \sigma_0^{-\frac{1}{2}(1+it)}\rho \sigma_0^{-\frac{1}{2}(1-it)}\Big] = 1.
\end{equation}
\end{theorem}

We note that if $\sigma_0$ is full-rank, then $\sigma_0^{-it}$ is unitary. In this case,   condition~\eqref{eq:additivity-Stein-condition weak additivity} can be written in terms of an optimization on a  commutative subgroup of unitaries, which by taking its closure can be assumed to be compact.

\begin{proof}
We first show the equivalence between~\eqref{eq:additivity Umegaki} and
\begin{equation}
\label{eq:additivity-Stein-condition weak additivity ineq}
\sup_{t \in \mathbb{R}, \tau \in \mathcal{F}_1}\textup{Tr}\Big[\tau \sigma_0^{-\frac{1}{2}(1+it)}\rho \sigma_0^{-\frac{1}{2}(1-it)}\Big] \leq  1.
\end{equation}
Indeed, by choosing $\tau =\sigma_0$, we note that 
\begin{equation}
\label{eq:additivity-Stein-condition}
\sup_{t \in \mathbb{R}, \tau \in \mathcal{F}_1}\textup{Tr}\Big[\tau \sigma_0^{-\frac{1}{2}(1+it)}\rho \sigma_0^{-\frac{1}{2}(1-it)}\Big] \geq 1 \,,
\end{equation}
and hence the inequality~\eqref{eq:additivity-Stein-condition weak additivity ineq} must hold as an equality.

Let us first show how~\eqref{eq:additivity Umegaki} implies~\eqref{eq:additivity-Stein-condition weak additivity ineq}. Let us assume additivity for all $n \geq 1$. This means that $\sigma_0^{\otimes n}$ is an optimizer for the $n$-letter optimization problem. Thus, the  necessary condition of Lemma~\ref{necessary and sufficient Umegaki} for the optimality of $\sigma_0^{\otimes n}$ for state $\rho^{\otimes n}$ and for the specific direction $\tau^{\otimes n}\in \cF_n$ yield
\begin{equation}
\left(\int_{-\infty}^\infty  \Trm\Big[\tau \sigma_0^{-\frac{1}{2}(1+it)}\rho \sigma_0^{-\frac{1}{2}(1-it)}\Big]^n \beta_{0}(t) \d t\right)^\frac{1}{n} \leq 1.
\end{equation}
 Taking the limit of $n\to \infty$, this inequality implies~\eqref{eq:additivity-Stein-condition weak additivity}.

We now show how~\eqref{eq:additivity-Stein-condition weak additivity ineq} implies~\eqref{eq:additivity Umegaki}. We need to verify that, under the assumption~\eqref{eq:additivity-Stein-condition weak additivity}, the sufficient condition of Lemma~\ref{necessary and sufficient Umegaki} holds for $\sigma_0^{\otimes n}$. Then, noting that $\beta_{0}(t)$ is a probability measure, we obtain
\begin{align}
&\sup_{\tau^{\n} \in \cF_n}\Trm\Bigg[\tau^{\n}\int_{-\infty}^\infty  (\sigma_0^{\otimes n})^{-\frac{1}{2}(1+it)}\rho^{\otimes n}(\sigma_0^{\otimes n})^{-\frac{1}{2}(1-it)}\Bigg] \beta_{0}(t)\dd t \\
&\qquad \qquad \qquad \leq \int_{-\infty}^\infty \sup_{\tau^{\n} \in \cF_n}\Trm\Big[\tau^{\n} (\sigma_0^{\otimes n})^{-\frac{1}{2}(1+it)}\rho^{\otimes n}(\sigma_0^{\otimes n})^{-\frac{1}{2}(1-it)}\Big] \beta_{0}(t)\dd t  \\
\label{critical}
& \qquad \qquad \qquad\leq \int_{-\infty}^\infty \sup_{\tau \in \cF_1}\Trm\Big[\tau \sigma_0^{-\frac{1}{2}(1+it)}\rho \sigma_0^{-\frac{1}{2}(1-it)}\Big]^n \beta_{0}(t)\dd t\\
&\qquad \qquad \qquad\leq 1.
\end{align}
where in~\eqref{critical} we used the assumption on the multiplicativity of the support function in Lemma~\ref{lemma:multiplicativity of support function}.  

\end{proof}
We observe that if the optimizer $\sigma_0$ commutes with $\rho$, then the additivity condition is automatically satisfied. In this case, the terms $\sigma_0^{\frac{it}{2}}$ and $\sigma_0^{-\frac{it}{2}}$ simplify, reducing the condition to the single-copy criterion given in Corollary~\ref{commuting} in Appendix~\ref{app:necessary and sufficient} for $\alpha = 1$. As discussed in Section~\ref{sec:examples}, this situation arises when the states exhibit certain symmetries.

In the following sections, we demonstrate that this observation directly yields the necessary and sufficient conditions for the additivity of the Stein exponent.

\subsection{Stein exponent and additivity}
The generalized Stein’s lemma in Theorem~\ref{thm:Stein} states that the Stein exponent is characterized by the regularized relative entropy. Since it is straightforward to see that the absence of regularization, i.e., $\Dreg(\rho) = \mathfrak{D}(\rho)$, implies additivity for all $n \in \mathbb{N}$, Theorem~\ref{thm:Stein-additivity} yields necessary and sufficient conditions under which the regularization in the Stein exponent can be omitted.

\begin{corollary}
\label{cor: Stein}
Consider the quantum hypothesis testing problem corresponding to the state $\rho$ and
collection $\{\cF_n\}_{n\in \N}$ satisfying Assumption~\ref{assumption:polar}, Assumption~\ref{assumption:support}, and Assumption~\ref{assumption:perm-inv}.
 Moreover, let $\sigma_0 \in \argmin_{\sigma \in \mathcal{F}_1} D(\rho \| \sigma)$. Then, the regularization in~\eqref{eq:Stein-composite} can be removed, i.e.
\begin{equation}
s(\rho\|\cF) =D(\rho\|\sigma_0) 
\end{equation}
if and only if
\begin{equation}
\sup_{t \in \mathbb{R}, \tau \in \mathcal{F}_1}\textup{Tr}\Big[\tau \sigma_0^{-\frac{1}{2}(1+it)}\rho \sigma_0^{-\frac{1}{2}(1-it)}\Big] = 1.
\end{equation}
\end{corollary}

\subsection{Strong additivity of the Umegaki relative entropy}
In this section, we show that the regularized relative entropy is strongly additive. As we demonstrate below, this is a consequence of the generalized Stein’s lemma (Theorem~\ref{thm:Stein}), together with the assumption on the polar sets.

\begin{prop}
\label{prop: Additivity regularized}
Let $\rho_1, \dots, \rho_k$ be quantum states on  finite-dimensional Hilbert spaces $\cH_1, \dots, \cH_k$, respectively. Let $\bigl\{\,\mathcal{F}\bigl(\bigotimes_{j=1}^k \mathcal{H}_j^{\otimes m_j}\bigr)\,\bigr\}_{(m_1,\dots,m_k)\in\mathbb{N}^k}$ be a family of sets of positive semidefinite operators that
satisfies Assumption~\ref{ass:tensor-compatibility}.
Then, we have
\begin{align}
\Dreg\left(\bigotimes_{j=1}^k \rho_j\right) = \sum_{j=1}^k \Dreg(\rho_j)\,.
\end{align}
\end{prop}
\begin{proof}
We begin by proving subadditivity and then demonstrate superadditivity, thereby establishing additivity.
Subadditivity follows from the assumption $\bigotimes_{j=1}^k\mathcal{F}(\mathcal{H}^{\otimes n}_j)
\subseteq \mathcal{F}(\bigotimes_{j=1}^k\mathcal{H}^{\otimes n}_j)$ which gives
    \begin{equation}
        \min_{\sigma^{(n)} \in \mathcal{F}(\bigotimes_{j=1}^k\mathcal{H}^{\otimes n}_j) }D\left(\bigotimes_{j=1}^k \rho_j^{\otimes n}\middle\|\sigma^{(n)}\right) \leq \sum_{j=1}^k\min_{\sigma_j^{(n)} \in \mathcal F(\mathcal{H}^{\otimes n}_j) }D\big(\rho_j^{\otimes n}\big\|\sigma_j^{(n)}\big).
    \end{equation}
Dividing by $n$ and taking the limit $n \rightarrow +\infty$ on both sides, we obtain the desired inequality.

 To prove superadditivity, we invoke the generalized Stein's lemma in~\eqref{thm:Stein}, which provides a variational characterization of the regularized relative entropy. Explicitly, we have that
    \begin{align}
    &\Dreg\left(\bigotimes_{j=1}^k \rho_j\right) =
      \sup\Big\{r\geq 0:\, \exists \{T_{n}\}_{n\in \N}, \, 0\leq T_{n}\leq I,\nonumber\\
      &  \qquad \qquad  \,   \liminf_{n\to +\infty} -\frac 1n\log \beta_{n}\Big(\mathcal F\Big(\bigotimes \nolimits_{j=1}^k\mathcal{H}^{\otimes n}_j\Big)\Big| T_{n}\Big)\geq r,\,  \lim_{n\to +\infty} \alpha_n\Big(\bigotimes\nolimits_{j=1}^k \rho_j\Big| T_{n}\Big)=0  \Big\}.
    \end{align}
We note that, here, $T_n$ is an operator acting on the Hilbert space $\bigotimes_{j=1}^k\mathcal{H}_j^{\otimes n}$ which is chosen as follows. Applying the generalized Stein's lemma for the state $\rho_j$, there are operators $T_{n, j}$ acting on  $\cH_{j}^{\otimes n}$ such that 
\begin{align}
\lim_{n\to +\infty}\alpha_n(\rho_j| T_{n, j})=0, \qquad\text{ and } \qquad \liminf_{n\to+\infty} -\frac 1n \log \beta_n(\cF(\cH_j^{\otimes n})| T_{n, j}) =  \Dreg(\rho_j).
\end{align} 
Now we let $T_n=T_{n, 1}\otimes\cdots\otimes T_{n,k}$.   We note that this product test satisfies
    \begin{equation}
        \lim_{n\to +\infty}\alpha_n\Big (\bigotimes\nolimits_{j=1}^k \rho_j\Big| T_{n}\Big) = 1-\lim_{n\to +\infty}\prod_{j=1}^k\Trm\big[\rho_j T_{n,j}\big] =0\,.
    \end{equation}
On the other hand, we can apply the assumption $\bigotimes_{j=1}^k\mathcal{F}(\mathcal{H}^{\otimes n}_j)^\circ
\subseteq \mathcal{F}(\bigotimes_{j=1}^k\mathcal{H}^{\otimes n}_j)^\circ$ and Lemma~\ref{lemma:multiplicativity of support function} to conclude the multiplicativity of the support functions and write
    \begin{equation}
        \beta_{n}\Big(\mathcal F\Big(\bigotimes\nolimits_{j=1}^k\mathcal{H}^{\otimes n}_j\Big)\Big|T_{n}\Big) = \prod_{j=1}^k \beta_n\big(\mathcal F(\mathcal{H}^{\otimes n}_j)\big|T_{n, j}\big).
    \end{equation}
This implies 
\begin{align}
\Dreg\left(\bigotimes_{j=1}^k \rho_j\right) &\geq \liminf_{n\to +\infty} -\frac 1n\log \beta_{n}\Big(\mathcal F\Big(\bigotimes\nolimits_{j=1}^k\mathcal{H}^{\otimes n}_j\Big)\Big|T_{n}\Big) \\
& =  \sum_{j=1}^m \liminf_{n\to +\infty} -\frac 1n \log \beta_n(\cF(\cH_j^{\otimes n})|T_{n, j}) = \sum_{j=1}^k \Dreg(\rho_j).
\end{align}    
\end{proof}

Using the proposition above, we determine the conditions under which the non-regularized, minimized Umegaki relative entropy is additive.
\begin{theorem}
\label{th: strong additivity Umegaki}
  Let $\rho_1, \dots, \rho_k$ be quantum states on finite-dimensional Hilbert spaces $\cH_1, \dots, \cH_k$, respectively.  Let $\bigl\{\,\mathcal{F}\bigl(\bigotimes_{j=1}^k \mathcal{H}_j^{\otimes m_j}\bigr)\,\bigr\}_{(m_1,\dots,m_k)\in\mathbb{N}^k}$ be a family of sets of positive semidefinite operators that
satisfies Assumption~\ref{ass:tensor-compatibility}.
Let $\sigma_{0,j} \in \argmin_{\sigma_j \in \cF(\mathcal{H}_j)} D(\rho_j\|\sigma_j)$ be an optimizer for $\rho_j$, $j=1,\dots, k$.
Then, we have
\begin{align}
\label{eq:strong additivity}
\mathfrak{D}\left(\bigotimes_{j=1}^k \rho_j^{\otimes m_j}\right) = \sum_{j=1}^k m_j\mathfrak{D}(\rho_j) \qquad \forall m_j \in \mathbb{N}, \,j=1,\dots,k
\end{align}
if and only if 
\begin{align}
\label{eq:additivity first state}
 \sup_{t \in \mathbb{R}, \tau_j \in \cF(\mathcal{H}_j)}\textup{Tr}\Big[\tau_j \sigma_{0,j}^{-\frac{1}{2}(1+it)}\rho_j \sigma_{0,j}^{-\frac{1}{2}(1-it)}\Big] = 1, \qquad \forall j=1,\dots, k\,.
\end{align}
\end{theorem}
\begin{proof}
    We establish the two directions separately. We first show how the conditions in~\eqref{eq:additivity first state} imply~\eqref{eq:strong additivity}.
    By invoking the subadditivity of the minimized relative entropy, the additivity of the regularized relative entropy stated in Proposition~\ref{prop: Additivity regularized}, and the sufficient conditions for weak additivity in Theorem~\ref{thm:Stein-additivity}, we find that
\begin{equation}
        \mathfrak{D}\left(\bigotimes_{j=1}^k \rho_j^{\otimes m_j}\right) \geq \Dreg\left(\bigotimes_{j=1}^k \rho_j^{\otimes m_j}\right) = \sum_{j=1}^k m_j\Dreg(\rho_j)  = \sum_{j=1}^k m_j\mathfrak{D}(\rho_j) \,.
    \end{equation}
We now turn to the converse direction. Suppose that condition~\eqref{eq:additivity first state} is not satisfied (say) for $j=1$. Then, by Theorem~\ref{thm:Stein-additivity}, it follows that additivity must fail for some value $m_1$. Setting $m_2=\cdots=m_k = 1$, we thus obtain
\begin{equation}
        \mathfrak{D}(\rho_1^{\otimes m_1} \otimes \rho_2\otimes \cdots \otimes \rho_k) \leq \mathfrak{D}(\rho_1^{\otimes m_1} )+  \sum_{j=2}^k\mathfrak{D}(\rho_j) < m_1 \mathfrak{D}(\rho_1) + \sum_{j=2}^k  \mathfrak{D}(\rho_j)\,,
    \end{equation}
where in the first inequality we use subadditivity. This shows that additivity is violated for $m_1$ and $m_2=\cdots=m_k = 1$. Thus,~\eqref{eq:strong additivity} implies~\eqref{eq:additivity first state}.
    
\end{proof}


\subsection{Examples}
\label{sec:examples}

In this section, we illustrate the above conditions by presenting concrete examples of both additivity and non-additivity, based on Theorem~\ref{thm:Stein-additivity}.

Our first example is a qubit scenario in the context of hypothesis testing with an arbitrarily varying source. In this case, we identify parameter ranges where additivity holds and others where it fails.
Remarkably, by tuning a single parameter, we construct examples where the relative entropy remains additive for $n$ copies but loses additivity at 
$n+1$ copies, even for arbitrarily large $n$.

We then turn to the resource theories of entanglement and magic, focusing on symmetric states for which the optimizer can be shown to commute with the input state.
As shown in Corollary~\ref{commuting} in Appendix~\ref{app:necessary and sufficient}, if the optimizer $\sigma_0 \in \argmin_{\sigma \in \mathcal{F}} D(\rho \| \sigma)$ commutes with $\rho$, then the supremum over $t \in \mathbb{R}$ in~\eqref{eq:additivity-Stein-condition} can be omitted. In this case, the resulting inequality holds by the optimality of $\sigma_0$, and regularization is not needed. This typically occurs when the states exhibit symmetries. Moreover, if all elements of $\mathcal{F}_1$ commute with each other (e.g., in the classical case), and in particular with $\sigma_0$, then a similar argument shows that additivity holds.

\subsubsection{A qubit example for additivity and non-additivity}
We start with an example for which the regularization in the Stein exponent can or cannot be removed, depending on the value of a real parameter $\lambda$. In particular, in the non-additivity case, we show that while for small $n$, the minimized relative entropy is additive, there exists $n_0$ such that for some $n \geq n_0$ it fails to be additive. 

Let $\rho = \ketbra{+}{+}$ with $\ket{\pm} = \frac{1}{\sqrt{2}}(\ket0 \pm \ket 1)$. Also, let $\cF=\cF_1=\text{conv}\{\sigma, \ketbra{-}{-} \}$ with 
\begin{align}\label{eq:example-def-sigma}
\sigma = \frac{1+\lambda}{2}\ketbra{0}{0} + \frac{1-\lambda}{2}\ketbra{1}{1},
\end{align}
for some fixed parameter $\lambda\in [0,1)$. We then define $\cF_n$, for $n\geq 2$, as in Example~\ref{example:AV} so that $\{\cF_n\}_{n\in \N}$ is an arbitrarily varying source.
We note that $\rho$ is orthogonal to $\ketbra{-}{-}$, so one expects that $\argmin_{\tau\in \cF} D(\rho\| \tau)=\sigma$ for small values of $\lambda$. In fact, this holds at least for $\lambda \leq \frac{3}{4}$. 
In the case $\lambda = 0$, the regularization is not needed since in this case $\sigma$ equals the fully mixed state and commutes with $\rho$, in which case additivity holds as discussed above.
When $\lambda\neq 0$, there exists $t\in \mathbb R$ such that $\sigma^{-it/2} \ket -$ equals $\ket +$ up to a phase. Indeed, such a value can be found by solving for \( t \) such that \(\left(1+\lambda\right)^{it/2} = -\left(1-\lambda\right)^{it/2} \), which yields 
\begin{equation}
    t = 2(\pi + 2\pi m)\left(\log\left(\frac{1+\lambda}{1-\lambda}\right)\right)^{-1} \,, \quad  m \in \mathbb{Z}\,.
\end{equation}
 In this case
\begin{align}
\Trm\Big[ \rho \sigma^{-\frac 12 (1+it)} \ketbra{-}{-} \sigma^{-\frac 12 (1-it)} \Big]= \big|\bra + \sigma^{-\frac12}\ket +  \big|^2=\frac{1}{1-\lambda^2}+\frac{1}{\sqrt{1-\lambda^2}}>1.
\end{align}
Thus, by Theorem~\ref{thm:Stein-additivity}, regularization is necessary.

The above argument shows that $\sigma^{\otimes n}$ is not an optimizer of $\min_{\sigma^{\n}\in \cF^{\n}}  D\big(\rho^{\otimes n}\big\| \sigma^{\n}\big)$  for some values of $n$ if $\lambda\neq 0$. One may wonder for which values of $n$ this quantity is not additive. To this end, one can check for which values of $n$ the optimality condition in Theorem~\ref{thm:derivative-condition} is violated for $\sigma^{\otimes n}$, i.e.,
\begin{align}
    \sup_{\tau^{\n}\in \cF^{\n}} \Trm\Bigg[\tau^{\n}\int_{-\infty}^\infty  (\sigma^{\otimes n})^{-\frac{1}{2}(1+it)}\rho^{\otimes n}(\sigma^{\otimes n})^{-\frac{1}{2}(1-it)}\Bigg] \beta_{0}(t)\dd t > 1.
\end{align}
We need to examine the above inequality for extreme points of $\cF_n$, which up to the permutation of subsystems are of the form $\sigma^{\otimes m}\otimes \ketbra{-}^{\otimes (n-m)}$. Writing down the expression for these states, we realize that we may consider only the case of $\tau^{\n} = \ketbra{-}{-}^{\otimes n}$. Expanding the operators in the computational basis, a straightforward calculation yields
\begin{align}
\label{eq:violation n}
    f_p(n):&= 
    \int_{-\infty}^\infty \big|\big\langle-\big|  \sigma^{-\frac{1}{2}(1+it)}\big|+\big\rangle\big|^{ n} \beta_{0}(t)\dd t \\
    &= \frac{1}{2^{2n}}\sum_{k,\ell=1}^n \binom{n}{k}\binom{n}{\ell} (-1)^{k+\ell} \Delta_{\log}(p^k(1-p)^{n-k},p^\ell(1-p)^{n-\ell}) \, ,
\end{align}
where $\Delta_{\log}(x,y) = (\log{x}-\log{y})/(x-y)$ denotes the logarithmic divided difference and $p=(1+\lambda)/2$. Thus, we need to find values of $n$ for which $f_p(n)>1$.

\begin{figure}
\centering
\begin{tikzpicture}
\node at (0,0) {\includegraphics[width=.6\textwidth]{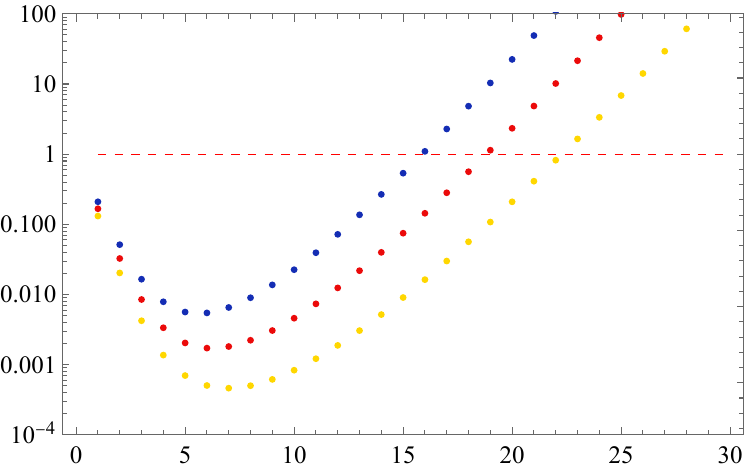}};

\draw[myyellow, thick] (1.8,-1.1) -- (2.5,-1.1);
\node[right] at (2.6,-1.1) { $p=0.70$};

\draw[myred, thick] (1.8,-1.5) -- (2.5,-1.5);
\node[right] at (2.6,-1.5) { $p=0.72$};

\draw[myblue, thick] (1.8,-1.9) -- (2.5,-1.9);
\node[right] at (2.6,-1.9) { $p=0.74$};

\node at (0.5,-3.4) {Number of copies $n$};
\node at (-5.4,0) {\rotatebox{90}{Value of $f_p(n)$}};
\end{tikzpicture}
\caption{The plot represents the value $f_p(n)$ as a function of $n$ for $p = 0.70, 0.72, 0.74$. The horizontal red line is the constant function $1$. Once the curve becomes bigger than $1$, the minimized relative entropy is not additive anymore. As $p$ tends to $1/2$ (commuting case), the additivity is violated for larger and larger values of $n$.}
\label{fig:Violation_n}
\end{figure}
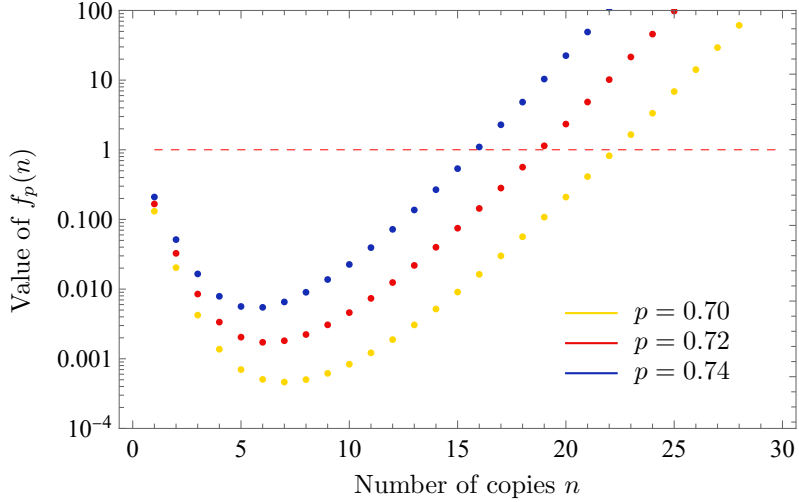

Fig.~\ref{fig:Violation_n} illustrates the validity of this inequality for sufficiently large $n$. We note that the violation of additivity occurs at larger values as $p$ tends to $1/2$ (and $\lambda$ tends to $0$), in which case the quantities are additive.

This example shows that in some cases the limit in~\eqref{eq:Stein-composite} is indeed necessary and the problem of computing the Stein exponent $\rms(\rho\|\{\cF_n\}_{n\in \N})$ is not easy in general. However, we note that $\frac 1n \min_{\sigma^{\n}\in \cF_n} D\big(\rho^{\otimes n}\big\|\sigma^{\n}\big)$ for any $n$ is an upper bound on $\rms(\rho\|\cF)$. In Appendix~\ref{app:lower-Stein-exp} we also show that $\min_{\sigma\in \cF} D^\downarrow(\rho\|\sigma)$ is a lower bound on $\rms(\rho\|\cF)$ where  $D^{\downarrow}\big(\rho\big\|\sigma)= \lim_{\alpha \rightarrow 1^-}D_{\alpha,1-\alpha}(\rho\|\sigma)$ is the reversed sandwiched R\'enyi relative entropy. For the specific examples considered in Fig.~\ref{fig:Violation_n} and $n=1$ we obtain $\min_{\sigma\in \cF} D^\downarrow(\rho\|\sigma)\leq   \rms(\rho\|\cF) \leq \min_{\sigma\in \cF} D(\rho\|\sigma)$ with 
$\min_{\sigma\in \cF} D(\rho\|\sigma)= 0.78, 0.80, 0.82$, for $p=0.70,0.72,0.74$, respectively, and $\min_{\sigma\in \cF} D^\downarrow(\rho\|\sigma)= 0.69$ for any $p$.

\subsubsection{Additivity for Werner states}
A bipartite quantum state $\rho_{A B}$, with subsystems $A, B$ having the same dimension $d$, is called a (generalized) Werner state if it is invariant under the following twirling map~\cite{werner1989quantum, bennett1996purification,vollbrecht2001entanglement} 
\begin{equation}
\mathcal{E}: \rho \rightarrow \int dU (U \otimes U) \rho (U^\dagger \otimes U^\dagger) \,,
\end{equation} 
where $\d U$ is the Haar measure over the group of unitary matrices of size $d\times d$. 
A Werner state can be written as
\begin{equation}
\label{W}
\rho^W(p) = p \frac{2}{d(d+1)}P^{\text{SYM}}_{AB} + (1-p)\frac{2}{d(d-1)}P^{\text{AS}}_{AB} \,,
\end{equation}
where $P^{\text{SYM}}_{AB}$ and $P^{\text{AS}}_{AB}$ are the projections onto the symmetric and antisymmetric subspaces, respectively. The Werner state is separable for $p \geq 1/2$ and entangled for $p<1/2$~\cite{vollbrecht2001entanglement}.

The relative entropy $D(\rho^{W}(0)\| \sigma)$ when minimized over $\sigma$ belonging to the set of separable states (SEP) or positive partial transpose states (PPT) is non-additive~\cite{vollbrecht2001entanglement}. Therefore, by continuity (see Proposition~\ref{Continuity} in Appendix~\ref{app:continuity}), additivity does not hold even in a neighborhood of $p=0$. However, as we discussed above, according to Theorem~\ref{thm:Stein-additivity} and Lemma~\ref{lem:additivity-commuting}, the relative entropy minimized over the Rains set is additive. The point is that by the data processing inequality $D(\rho^W(p)\| \sigma)\geq D(\rho^W(p)\| \mathcal E(\sigma))$. Moreover, the Rains set is invariant under the twirling map because it is an LOCC operation~\cite{regula2019one}. Thus, when minimizing over $\sigma$ belonging to the Rains set, we may replace $\sigma$ with $\mathcal E(\sigma)$ which we know is a Werner state. On the other hand, Werner states constitute a commuting family. Hence, weak additivity follows from Theorem~\ref{thm:Stein-additivity} and the discussion thereafter. Moreover, in Lemma~\ref{lem:additivity-commuting}, we show that additivity holds more generally in this setting—specifically, for any pair of states where at least one is a Werner state.

Interestingly, it is shown in~\cite{audenaert2002asymptotic} that for Werner states, the regularized relative entropy minimized over PPT states coincides with the relative entropy minimized over the Rains set (even without regularization as discussed above). Hence, the approximation of the set of separable or PPT states with the Rains set is sometimes tight. 

We conclude by noting that the Umegaki relative entropy minimized over the Rains set is not additive in general~\cite{wang2017nonadditivity}. Analogously, violations of the additivity of the Rains relative entropy can be found by identifying a single-copy optimizer and constructing a phase $t$ and direction $\tau$ such that the inequality in Theorem~\ref{thm:Stein-additivity} is violated.

\subsubsection{Additivity for noisy strange states}
We consider the class of qutrit strange states subject to depolarizing noise
\begin{equation}
\rho(p)=pS + (1-p) \frac{I}{3} \,,
\end{equation}
where $S = \ketbra{S}{S}$ is the strange state with $\ket{S} = (\ket{1}-\ket{2})/\sqrt{2}$ (see, e.g.,~\cite{veitch2014resource}).
This class of states is invariant under stabilizer operations (see, e.g.,~\cite[proof of Proposition 14]{takagi2022one} as well as the operation~\cite[Section III.A]{veitch2014resource})
\begin{equation}
\mathcal{E}(\rho) = \Tr[S \rho] S + \Tr[(I-S) \rho] \frac{I-S}{2} \,.
\end{equation}
The pure qutrit strange state serves as a general counterexample to the additivity of magic monotones defined via quantum divergences minimized over the set of stabilizer states (see~\cite{veitch2014resource} and~\cite[Proposition 12]{rubboli2024mixed}). As before, we can leverage the data processing inequality to restrict the optimization over the set $\mathcal W_1$ given in Example~\ref{example:mana} to the set of states of the form $\mathcal{E}(\sigma)$ where $\sigma \in \mathcal{W}_1$. On the other hand, these states form a commuting family. Hence, weak additivity for the family $\{\mathcal{W}_n\}_{n\in \mathbb{N}}$ follows from Theorem~\ref{thm:Stein-additivity} and the discussion thereafter. In Lemma~\ref{lem:additivity-commuting}, we further show that a stronger result holds in this case: additivity holds for any pair of states as long as at least one of them is a strange state.


\section{Additivity problem for the $\alpha$-$z$-R\'enyi relative entropy}\label{sec:additivity}
In the preceding section, we established how the additivity conditions for the Umegaki relative entropy yield necessary and sufficient conditions for removing the regularization in Stein's exponent. In the following sections, we examine additional error exponents, which we will express in terms of the Petz and sandwiched R\'enyi divergences. In this section, we investigate the additivity of these quantities under the same assumptions imposed in the previous section.

In this section, we fix a quantum state $\rho$ and a collection $\{\cF_n\}_{n\in \N}$ satisfying Assumption~\ref{assumption:polar} and Assumption~\ref{assumption:support}, and discuss the additivity of the $\alpha$-$z$-R\'enyi relative entropies minimized over the set $\cF_n$. More precisely, we ask under what conditions we have 
\begin{equation}\label{eq:main-additivity-question}
\min_{\sigma^{\n} \in \mathcal{F}_n} D_{\alpha,z}\big(\rho^{\otimes n} \big\| \sigma^{\n}\big) =n \min_{\sigma \in \mathcal{F}_1} D_{\alpha,z}(\rho\| \sigma), \qquad \quad \forall n.
\end{equation}
To solve this problem, we first find a necessary and sufficient condition to characterize the optimizer. 
To this end, we start with a key technical lemma which allows us to write the necessary and sufficient conditions for the minimizers on the right-hand side of~\eqref{eq:main-additivity-question} in a convenient form. This lemma might also be of independent interest. 

\begin{lemma}
\label{lem:derivative-integral}
For any $\theta \in (-1,1)$ and positive real number $x, y$ we have
\begin{align}
\frac{x^{\theta}-y^{\theta}}{\theta(x-y)} & = 
\sinc(\theta\pi) \int_0^\infty  \frac{1}{(x+t)(y+t)} t^\theta\dd t \\
& =  \int_{-\infty}^{\infty} x^{-\frac{1}{2}(1-\theta+it)}y^{-\frac{1}{2}(1-\theta-it)} \beta_{\theta}(t) \dd t,
\end{align}
where
\begin{align}\label{eq:def-beta-dist}
\beta_{\theta}(t)=\frac{\sin(\theta \pi)}{2\theta(\cosh(\pi t) + \cos(\theta \pi))}
\end{align}
is a probability measure. Here, the left-hand side for $\theta=0$ is understood in the limit and is equal to $\frac{\log x - \log y}{x-y}$. We also have $\beta_0(t) = \lim_{\theta\to 0} \beta_\theta(t) = \frac{\pi}{2(\cosh(\pi t) +1)}$.
\end{lemma}

\begin{proof} Multiplying these equations by $x^{-\theta}y^{-\theta}$,  a straightforward calculation shows that it suffices to prove the lemma for non-negative $\theta$. Also, the case of $\theta=0$ can be proven by taking the limit of $\theta\searrow 0$, and in fact is already known~\cite[Lemma 3.4]{Sutter+2017multivariate}. Thus, we focus on $\theta\in (0,1)$. In this case, the first equation follows by direct computation and is also a consequence of the following well-known integral representation~\cite{Bhatia-P}:  
\begin{align}\label{eq:x-alpha-integral}
x^\theta = \frac{\sin(\theta \pi)}{\pi}\int_0^{\infty}   \frac{x}{t+x} t^{\theta-1}\dd t\,, \qquad \theta\in (0,1).
\end{align}
To prove the second equation, recall that the Poisson kernel of the strip $\{z=\theta+i\xi:\, 0\leq \theta\leq 1, \xi\in \mathbb R\}$ is equal to~\cite{BerghLofstrom1976interpolation} 
\begin{align}
K(\theta+i\xi, b+it) = \frac{\sin(\pi \theta)}{2\big( \cosh(\pi(t-\xi)) + \cos(\pi (\theta+b-1))  \big)},\qquad b=0,1.
\end{align} 
Then, using $\sin(\pi \theta) = \sin(\pi(1-\theta))$ for $\xi=0$ we have
\begin{align}
    K(\theta, it) = (1-\theta) \beta_{1-\theta}(t), \qquad K(\theta, 1+it) = \theta \beta_\theta(t) .
\end{align}
This means that for any analytic function $f(z)$ on the strip we have
\begin{align}
    f(\theta) = (1-\theta)\int_{-\infty}^{+\infty}  f(it) \beta_{1-\theta}(t)\dd t +  \theta\int_{-\infty}^{+\infty}  f(1+it) \beta_{\theta}(t)\dd t \,.
\end{align}
Then, taking $f(z) = e^{cz}$ for some $c\in \mathbb R$ we find that
\begin{align}
e^{c(1+\theta)} - e^{c(1-\theta)}  = e^c f(\theta) - f(1-\theta)
 = \theta\int_{-\infty}^{+\infty} \big( e^{2c} -1\big) e^{ict} \beta_\theta(t)\dd t.
\end{align}
Writing down this equation for $e^c=\sqrt{\frac yx}$ where $x,y$ are positive real numbers and multiplying both sides by  $(xy)^{-\frac 12(1-\theta)} (\frac yx-1)^{-1}$ we obtain 
\begin{align}\label{eq:comment-84}
\frac{x^{\theta} - y^{\theta}}{x-y}  
 = \theta\int_{-\infty}^{+\infty} x^{-\frac{1}{2}(1-\theta+it)} y^{-\frac{1}{2}(1-\theta-it)}\beta_\theta(t)\dd t.
\end{align}
\end{proof}
As a corollary, we obtain an integral expression for the Frech\'et derivative of the power function. This result extends the corresponding statement for the logarithm established in~\cite[Lemma 3.4]{Sutter+2017multivariate}. For a positive definite $A$ and Hermitian $H$, we define
\begin{equation}
      D x^\theta(A)[H] = \lim_{t \rightarrow 0} \frac{(A+tH)^\theta -A^\theta}{t} \,.
\end{equation}
\begin{corollary}
\label{Frechet derivative power}
    Let $A$ be a positive definite operator and $H$ be a Hermitian operator. Then, for any $\theta \in (-1,1)$, we have 
    \begin{equation}\label{eq:comment-86}
        D x^\theta(A)[H] = \theta \sinc\left(\theta \pi\right) \int_0^\infty \lambda^\theta(A + \lambda)^{-1}H(A + \lambda)^{-1}  \dd \lambda = \theta\int_{-\infty}^{\infty} A^{-\frac{1}{2}(1-\theta+it)} H A^{-\frac{1}{2}(1-\theta-it)} \beta_{\theta}(t) \d t \,.
    \end{equation}
\end{corollary}
\begin{proof}
    The first equality is proven, e.g. in~\cite[Lemma 1]{rubboli2024new} and follows from the integral representation~\eqref{eq:x-alpha-integral}. For the second equality we use Lemma~\ref{lem:derivative-integral}. To this end, we let  $A = \sum_{k}\mu_k \ketbra{k}{k}$ be the spectral decomposition of $A$ with $\mu_k>0$. Then, $(A+\lambda)^{-1}= \sum_k \frac{1}{\mu_k+\lambda}\ketbra{k}{k}$ yields 
   \begin{align}
     D x^\theta(A)[H] & = \theta \sinc\left(\theta \pi\right) \int_0^\infty \lambda^\theta(A + \lambda)^{-1}H(A + \lambda)^{-1}  \dd \lambda \\
 &= \theta\sinc(\theta\pi )\sum_{k,\ell}  \left(\int_0^\infty \lambda^\theta \frac{1}{(\mu_k+\lambda)(\mu_\ell + \lambda)}   \dd \lambda \right)   \ketbra{k}{k}H\ketbra{\ell}{\ell} \\
&  = \theta \sum_{k,\ell}  \left(  \int_{-\infty}^{\infty} \mu_k^{-\frac{1}{2}(1-\theta+it)}\mu_\ell^{-\frac{1}{2}(1-\theta-it)} \beta_{\theta}(t) \dd t \right)   \ketbra{k}{k}H\ketbra{\ell}{\ell} \\
& = \theta \int_{-\infty}^{\infty} A^{-\frac{1}{2}(1-\theta+it)}H A^{-\frac{1}{2}(1-\theta-it)} \beta_{\theta}(t) \dd t,
\end{align}
where in the third line we use Lemma~\ref{lem:derivative-integral}, and the last line follows from $A^{-\frac{1}{2}(1-\theta\pm it)} = \sum_k \mu_k^{-\frac{1}{2}(1-\theta\pm it)}\ketbra{k}{k}$. 
\end{proof}
Next, we use the latter result to derive necessary and sufficient conditions for the optimizer of the $\alpha$-$z$ R\'enyi relative entropies.
For quantum states $\rho, \sigma$ define
\begin{align}\label{eq:chi def}
\chi_{\alpha,z}(\rho,\sigma):=  \rho^\frac{\alpha}{2z}\big( \rho^{\frac{\alpha}{2z}}\sigma^{\frac{1-\alpha}{z}} \rho^{\frac{\alpha}{2z}}\big)^{z-1} \rho^\frac{\alpha}{2z},
\end{align}
where negative powers are taken in the sense of generalized inverses; also, $\rho^0$ equals the orthogonal projection on the support of $\rho$, which we denote by $\supp(\rho)$. Let
\begin{align}
\label{equation problem}
&\Xi_{\alpha,z}(\rho,\sigma) :=  \begin{dcases}
\chi_{\alpha,1-\alpha}(\rho,\sigma) &  \textup{if} \; z=1-\alpha, \\
\sigma^{-1}  \chi_{\alpha,\alpha-1}(\rho, \sigma) \sigma^{-1}  &  \textup{if} \; z=\alpha-1, \\
\int_{-\infty}^{\infty} \sigma^{-\frac{1}{2}(1-\frac{1-\alpha}{z}+it)}\chi_{\alpha,z}(\rho,\sigma)\sigma^{-\frac{1}{2}(1-\frac{1-\alpha}{z}-it)} \beta_{\frac{1-\alpha}{z}}(t) \d t & \textup{if} \; |1-\alpha|/z \neq 1,
\end{dcases} 
\end{align}
where $\beta_{\theta}(t)$ is the probability measure defined in Lemma~\ref{lem:derivative-integral}. Moreover, we define
\begin{align}
\label{sets}
S_{\alpha,z}(\rho) =  \begin{cases} 
\left\{ \sigma \; \text{is a state} : \supp(\rho) = \supp\!\left(\rho^0\sigma \rho^0\right) \right\}  &  \textup{if} \; (1-\alpha)/z = 1,\\
\left\{ \sigma\; \text{is a state} : \supp(\rho) \subseteq \supp(\sigma)  \right\} &  \textup{otherwise} .
\end{cases} 
\end{align} 
Finally, following~\cite{rubboli2024new} we let $\mathcal D$ be the set of pairs $(\alpha, z)$ for which the $\alpha$-$z$-R\'enyi relative entropy satisfies the data processing inequality:
\begin{align}
\mathcal D&=  \{(\alpha, z): 0<\alpha<1, z\geq \max\{\alpha, 1-\alpha\}\}\\
&\quad\,\, \cup \{(\alpha, z): \alpha>1, \max\{\alpha/2,\alpha-1\}\leq z\leq \alpha\}\\
 &\quad\,\, \cup \{(\alpha, z): \alpha=1,z>0 \}.
\end{align}

The next result combines Corollary~\ref{Frechet derivative power} and the necessary and sufficient conditions derived in~\cite{rubboli2024new}.

\begin{theorem}
\label{thm:derivative-condition}
Let $(\alpha,z) \in \mathcal{D}$, $z \neq \alpha-1$, $\rho$ be a quantum state and $\cF$ be a convex and compact subset of positive operators that includes at least one state whose support contains the support of $\rho$.
Then, $\sigma_0 \in \argmin_{\sigma \in \mathcal{F}} D_{\alpha,z}(\rho \| \sigma)$ if and only if $\sigma_0 \in S_{\alpha,z}(\rho)$ and
\begin{align}\label{eq:derivative-condition alpha-z}
\sup_{\tau\in \cF} \textup{Tr}\big[\tau\, \Xi_{\alpha,z}(\rho,\sigma_0)\big] =  Q_{\alpha,z}(\rho \|\sigma_0) \,.
\end{align}
Moreover, if $\rho$ is full-rank, the same holds for the case $z=\alpha-1$.
\end{theorem}

\begin{proof}
The proof combines the result of Theorem~\ref{main Theorem} from Appendix~\ref{app:necessary and sufficient} with the integral representation in Corollary~\ref{Frechet derivative power}.
The support conditions remain the same.  Moreover, by the same theorem (as well as Lemma~\ref{lem:derivative-support} in Appendix~\ref{app:derivative}), the necessary and sufficient condition~\eqref{eq:derivative-condition} holds if $\Xi_{\alpha,z}(\rho,\sigma_0)$ is replaced with $\Xi'_{\alpha,z}(\rho,\sigma_0)$ given by
\begin{align}\label{eq:comment-98}
&\Xi'_{\alpha,z}(\rho,\sigma) :=  \begin{dcases}
\Xi_{\alpha,z}(\rho,\sigma)  &  \quad  |(1-\alpha)/z| = 1, \\
\sinc\left(\frac{{(1-\alpha)\pi}}{z}\right) \int_0^\infty \lambda^\frac{1-\alpha}{z}(\sigma + \lambda)^{-1}\chi_{\alpha,z}(\rho,\sigma)(\sigma + \lambda)^{-1}  \dd \lambda & \quad  |(1-\alpha)/z| \neq 1.
\end{dcases} 
\end{align}
Thus, we only need to show that $\Xi'_{\alpha,z}(\rho,\sigma_0)=\Xi_{\alpha,z}(\rho,\sigma_0)$. 
To this end, let $\theta = \frac{1-\alpha}{z}$ and observe that for $(\alpha, z)\in \mathcal D$ we have $|\theta|\leq 1$. Assuming that $|\theta|\neq 1$, the second equality in Corollary~\ref{Frechet derivative power} gives $\Xi'_{\alpha,z}(\rho,\sigma) = \Xi_{\alpha,z}(\rho,\sigma)$.
\end{proof}

We can now state our main result regarding the additivity problem~\eqref{eq:main-additivity-question}.

\begin{theorem}\label{thm:additivity-derivative}
Suppose that the quantum state $\rho$ and $\{\cF_n\}_{n\in \N}$ satisfy Assumption~\ref{assumption:polar} and Assumption~\ref{assumption:support}.
Let $(\alpha,z) \in \mathcal{D}$ be such that $|(1-\alpha)|/z \neq 1$  
Moreover, let $\sigma_0 \in \argmin_{\sigma \in \mathcal{F}_1} D_{\alpha,z}(\rho \| \sigma)$. Then, 
\begin{equation}
\label{eq:additivity}
\min_{\sigma^{\n} \in \mathcal{F}_n} D_{\alpha,z}\big(\rho^{\otimes n} \big\| \sigma^{\n}\big) = nD_{\alpha,z}(\rho\| \sigma_0)\, \quad \forall n \geq 1 ,
\end{equation}
if and only if
\begin{equation}
\label{condition on t}
\sup_{t \in \mathbb{R}, \tau \in \mathcal{F}_1}\textup{Tr}\Bigg[\tau \sigma_0^{-\frac{1}{2}\big(1-\frac{1-\alpha}{z}+it\big)}\chi_{\alpha,z}(\rho,\sigma_0) \sigma_0^{-\frac{1}{2}\big(1-\frac{1-\alpha}{z}-it\big)}\Bigg] = Q_{\alpha,z}(\rho\|\sigma_0).
\end{equation}
\end{theorem}

\begin{proof}
We derive the equivalence between~\eqref{eq:additivity} and
\begin{equation}
\label{condition on t ineq}
\sup_{t \in \mathbb{R}, \tau \in \mathcal{F}_1}\textup{Tr}\Bigg[\tau \sigma_0^{-\frac{1}{2}\big(1-\frac{1-\alpha}{z}+it\big)}\chi_{\alpha,z}(\rho,\sigma_0) \sigma_0^{-\frac{1}{2}\big(1-\frac{1-\alpha}{z}-it\big)}\Bigg] \leq Q_{\alpha,z}(\rho\|\sigma_0)
\end{equation}
and then note that the latter inequality is an equality by choosing $\tau = \sigma_0$.

We first start by showing how~\eqref{eq:additivity} implies~\eqref{condition on t ineq}. Let us assume additivity for all $n \geq 1$. This means that $\sigma_0^{\otimes n}$ is an optimizer for the $n$-letter optimization problem. Thus, writing the necessary and sufficient condition of Theorem~\ref{thm:derivative-condition} for $\rho^{\otimes n}, \sigma_0^{\otimes n}$ and in the direction of $\tau^{\otimes n}\in \cF_n$, we find that
\begin{equation}
\left(\int_{-\infty}^\infty  \Trm\Big[\tau \sigma_0^{-\frac{1}{2}(1-\frac{1-\alpha}{z}+it)}\chi_{\alpha,z}(\rho,\sigma_0) \sigma_0^{-\frac{1}{2}(1-\frac{1-\alpha}{z}-it)}\Big]^n \beta_{\frac{1-\alpha}{z}}(t) \dd t\right)^\frac{1}{n} \leq Q_{\alpha,z}(\rho\|\sigma_0).
\end{equation}
Here, we use $\chi_{\alpha,z}(\rho^{\otimes n},\sigma_0^{\otimes n}) = \chi_{\alpha,z}(\rho,\sigma_0)^{\otimes n}$ and $Q_{\alpha,z}(\rho^{\otimes n},\sigma_0^{\otimes n}) = Q_{\alpha,z}(\rho,\sigma_0)^{ n}$. Taking the limit of $n\to \infty$ and using Theorem~\ref{thm:derivative-condition}, this inequality implies~\eqref{condition on t}.

We now show how~\eqref{condition on t} implies~\eqref{eq:additivity}. We need to verify that, under the assumption~\eqref{condition on t}, the necessary and sufficient conditions of Theorem~\ref{thm:derivative-condition} hold for $\sigma_0^{\otimes n}$. To this end, using condition~\ref{A.4-polar}, the identity 
$\chi_{\alpha,z}\big(\rho^{\otimes n},\sigma_0^{\otimes n}\big)=\chi_{\alpha,z}(\rho,\sigma_0)^{\otimes n}$ and the fact that $\beta_{\frac{1-\alpha}{z}}(t)$ is a density function, we find that
\begin{align}
& \sup_{\tau^{\n}\in \cF_n} \Trm\big[\tau^{\n}\, \Xi_{\alpha,z}\big(\rho^{\otimes n},\sigma_0^{\otimes n}\big)\big] \\
& \qquad\qquad= \sup_{\tau^{\n}\in \cF_n} \int_{-\infty}^\infty  \Trm\Big[\tau^{\n} \Big(\sigma_0^{-\frac{1}{2}(1-\frac{1-\alpha}{z}+it)}\chi_{\alpha,z}(\rho,\sigma_0) \sigma_0^{-\frac{1}{2}(1-\frac{1-\alpha}{z}-it)}\Big)^{\otimes n}\Big] \beta_{\frac{1-\alpha}{z}}(t)\dd t \\
&\qquad\qquad \leq \sup_{t\in \mathbb R, \tau^{\n}\in \cF_n}\Trm\Big[\tau^{\n} \Big(\sigma_0^{-\frac{1}{2}(1-\frac{1-\alpha}{z}+it)}\chi_{\alpha,z}(\rho,\sigma_0) \sigma_0^{-\frac{1}{2}(1-\frac{1-\alpha}{z}-it)}\Big)^{\otimes n}\Big]  \\
&\qquad\qquad =
\sup_{t\in \mathbb R, \tau\in \cF_1} \Trm\Big[\tau \sigma_0^{-\frac{1}{2}(1-\frac{1-\alpha}{z}+it)}\chi_{\alpha,z}(\rho,\sigma_0) \sigma_0^{-\frac{1}{2}(1-\frac{1-\alpha}{z}-it)}\Big]^n \\
&\qquad\qquad \leq Q_{\alpha,z}(\rho\|\sigma_0)^n,
\end{align}
where in the last line we use~\eqref{condition on t}. The desired result then follows from Theorem~\ref{thm:derivative-condition}.

\end{proof}

As for the case $\alpha=1$, we note that in the case where an optimizer $\sigma_0$ for $\rho$ commutes with $\rho$, the above condition for additivity is automatically satisfied. Indeed, in this case, the phases can be simplified, leading to a condition that is automatically satisfied, as it coincides with the single-copy necessary and sufficient condition stated in Corollary~\ref{commuting} in Appendix~\ref{app:necessary and sufficient}. Moreover, given two states, additivity holds when at least one of them commutes with its optimizer.

\begin{lemma}
\label{lem:additivity-commuting}
Let $(\alpha,z)\in\mathcal{D}$ be such that $\lvert 1-\alpha\rvert/z \neq 1$,
and let $\rho_1$ and $\rho_2$ be quantum states on the finite-dimensional
Hilbert spaces $\mathcal{H}_1$ and $\mathcal{H}_2$, respectively.
Let $\mathcal{F}(\mathcal{H}_1)$, $\mathcal{F}(\mathcal{H}_2)$, and
$\mathcal{F}(\mathcal{H}_1\otimes\mathcal{H}_2)$ be sets of positive
semidefinite operators, where $\mathcal{F}(\mathcal{H}_j)$ contains a state
whose support includes that of $\rho_j$ for $j=1,2$, and assume that
$\mathcal{F}(\mathcal{H}_1)^\circ \otimes \mathcal{F}(\mathcal{H}_2)^\circ
\subseteq \mathcal{F}(\mathcal{H}_1\otimes\mathcal{H}_2)^\circ$.
Let $\sigma_{0,1}\in\arg\min_{\sigma_1\in\mathcal{F}(\mathcal{H}_1)}
D_{\alpha,z}(\rho_1\|\sigma_1)$.
If $[\sigma_{0,1},\rho_1]=0$, then we have
\begin{equation}
\mathfrak{D}_{\alpha,z}(\rho_1 \otimes \rho_2) = \mathfrak{D}_{\alpha,z}(\rho_1)+ \mathfrak{D}_{\alpha,z}(\rho_2) \,.
\end{equation}
\end{lemma}

\begin{proof}
According to Corollary~\ref{commuting} in Appendix~\ref{app:necessary and sufficient}, we need to show that
\begin{align}
\sup_{\tau  \in \mathcal{F}(\mathcal{H}_1\otimes \mathcal{H}_2)} \Trm \Big[\tau \Xi_{\alpha,z}(\rho_1 \otimes \rho_2,\sigma_{0,1} \otimes \sigma_{0,2})  \Big] = Q_{\alpha,z}(\rho_1 \otimes \rho_2 \|\sigma_{0,1}\otimes \sigma_{0,2})\,,
\end{align}
where $\sigma_{0,j} \in \argmin_{\sigma_j \in \mathcal{F}(\mathcal{H}_j)}D_{\alpha,z}(\rho_j\|\sigma_{j})$ for $j=1,2$.
Now the point is that $Q_{\alpha,z}(\rho_1 \otimes \rho_2 \|\sigma_{0,1}\otimes \sigma_{0,2}) = Q_{\alpha,z}(\rho_1  \|\sigma_{0,1})Q_{\alpha,z}(\rho_2\|\sigma_{0,2})$ and $\Xi_{\alpha,z}(\rho_1 \otimes \rho_2,\sigma_{0,1} \otimes \sigma_{0,2}) =\Xi_{\alpha,z}(\rho_1,\sigma_{0,1})\otimes \Xi_{\alpha,z}(\rho_2,\sigma_{0,2})$ since $[\sigma_{0,1},\rho_1]=0$ (see~\cite[Proof of Theorem 7]{rubboli2024new}). Therefore, by the assumption $\mathcal{F}(\mathcal{H}_1)^\circ \otimes \mathcal{F}(\mathcal{H}_2)^\circ
\subseteq \mathcal{F}(\mathcal{H}_1\otimes\mathcal{H}_2)^\circ$ and Lemma~\ref{lemma:multiplicativity of support function} for $k=2$, the support function is multiplicative and we have
\begin{align}
&\sup_{\tau \in \mathcal{F}(\mathcal{H}_1\otimes \mathcal{H}_2)} \Trm \Big[\tau \Xi_{\alpha,z}(\rho_1 \otimes \rho_2,\sigma_{0,1} \otimes \sigma_{0,2})  \Big] \\
&\qquad\qquad  = \sup_{\tau_1 \in \mathcal{F}(\mathcal{H}_1)} \Trm \Big[\tau_1 \Xi_{\alpha,z}(\rho_1,\sigma_{0,1})  \Big] \sup_{\tau_2 \in \mathcal{F}(\mathcal{H}_2)} \Trm \Big[\tau_2 \Xi_{\alpha,z}(\rho_2,\sigma_{0,2})  \Big] \\
&\qquad \qquad = Q_{\alpha,z}(\rho_1  \|\sigma_{0,1})Q_{\alpha,z}(\rho_2\|\sigma_{0,2}),
\end{align}
where the last inequality follows from the optimality of $\sigma_{0,1}$ and $\sigma_{0,2}$ and Corollary~\ref{commuting}.
\end{proof}

\begin{remark}\label{rem:az additivity simplified}
Assume that $z>1$, and define $X:=\rho^{\frac{\alpha}{2z}}\sigma_0^{\frac{1-\alpha}{2z}}$. Substituting \eqref{eq:chi def}, the condition in \eqref{condition on t} can be rewritten as
\begin{align}
&Q_{\alpha,z}(\rho\|\sigma_0)\\
& =
\sup_{t \in \mathbb{R}, \tau \in \mathcal{F}_1}\textup{Tr}\bigg[\tau 
\sigma_0^{-\frac{it}{2}}\sigma_0^{\frac{z+\alpha-1}{z}}
\rho^\frac{\alpha}{2z}\Big( \rho^{\frac{\alpha}{2z}}\sigma_0^{\frac{1-\alpha}{z}} \rho^{\frac{\alpha}{2z}}\Big)^{z-1} \rho^\frac{\alpha}{2z}
\sigma_0^{\frac{z+\alpha-1}{z}} \sigma_0^{-\frac{it}{2}}\bigg]\\
&\ds=
\sup_{t \in \mathbb{R}, \tau \in \mathcal{F}_1}\textup{Tr}\bigg[\tau 
\sigma_0^{-\frac{it}{2}}\sigma_0^{\frac{z+\alpha-1}{z}}\sigma_0^{\frac{\alpha-1}{2z}}\underbrace{\sigma_0^{\frac{1-\alpha}{2z}}
\rho^\frac{\alpha}{2z}\Big( \rho^{\frac{\alpha}{2z}}\sigma_0^{\frac{1-\alpha}{z}} \rho^{\frac{\alpha}{2z}}\Big)^{z-1} \rho^\frac{\alpha}{2z}
\sigma_0^{\frac{1-\alpha}{2z}}}_{=X^*(XX^*)^{z-1}X=(X^*X)^z}\sigma_0^{\frac{\alpha-1}{2z}}\sigma_0^{\frac{z+\alpha-1}{z}} \sigma_0^{-\frac{it}{2}}\bigg]\\
&=
\sup_{t \in \mathbb{R}, \tau \in \mathcal{F}_1}\textup{Tr}\bigg[\tau 
\sigma_0^{-\frac{1+it}{2}}\Big(\sigma_0^{\frac{1-\alpha}{2z}} \rho^{\frac{\alpha}{z}}\sigma_0^{\frac{1-\alpha}{2z}} \Big)^{z}
\sigma_0^{-\frac{1-it}{2}}\bigg].
\end{align}
\end{remark}

\medskip
In the following theorem, we prove that in the remaining case of $|(1-\alpha)|/z =1$, the additivity~\eqref{eq:main-additivity-question} always holds. Note that these cases include, among many others, quantities such as the log-fidelity, reversed sandwiched R\'enyi relative entropy, min-relative entropy, and max-relative entropy.\footnote{The latter two quantities are recovered in the limits $\alpha \rightarrow 0$ and $\alpha \rightarrow \infty$, respectively~\cite[Section D]{rubboli2024mixed}.} Interestingly, as shown in Appendix~\ref{app:lower-Stein-exp}, the additivity of these quantities can be used to provide bounds on some error exponents that are tight in the classical case.

\begin{theorem}
\label{Additivity in the line}
Let $(\alpha,z)\in\mathcal{D}$ be such that $\lvert 1-\alpha\rvert/z = 1$,
and let $\rho_1$ and $\rho_2$ be quantum states on the finite-dimensional
Hilbert spaces $\mathcal{H}_1$ and $\mathcal{H}_2$, respectively.
Let $\mathcal{F}(\mathcal{H}_1)$, $\mathcal{F}(\mathcal{H}_2)$, and
$\mathcal{F}(\mathcal{H}_1\otimes\mathcal{H}_2)$ be sets of positive
semidefinite operators, where $\mathcal{F}(\mathcal{H}_j)$ contains a state
whose support includes that of $\rho_j$ for $j=1,2$, and assume that
$\mathcal{F}(\mathcal{H}_1)^\circ \otimes \mathcal{F}(\mathcal{H}_2)^\circ
\subseteq \mathcal{F}(\mathcal{H}_1\otimes\mathcal{H}_2)^\circ$. Then, we have
\begin{equation}
\mathfrak{D}_{\alpha,z}(\rho_1 \otimes \rho_2) = \mathfrak{D}_{\alpha,z}(\rho_1)+ \mathfrak{D}_{\alpha,z}(\rho_2) \,.
\end{equation}
\end{theorem}

\begin{proof}
It suffices to prove the result in the case where $\rho$ is full-rank as the general case then follows by a continuity argument using Proposition~\ref{Continuity} in Appendix~\ref{app:continuity}. 
In this case, let us denote with $\mathcal{H}_1$ and $\mathcal{H}_2$ the dimensions of $\rho_1$ and $\rho_2$, respectively. Moreover, let $\sigma_{0,1} \in \argmin_{\sigma_1 \in \cF(\mathcal{H}_1)} D(\rho_1\|\sigma_1)$ and $\sigma_{0,2} \in \argmin_{\sigma_2 \in \cF(\mathcal{H}_2)} D(\rho_2\|\sigma_2)$. Writing down the necessary and sufficient condition of Theorem~\ref{thm:derivative-condition} which holds in our setting, we need to show that
\begin{align}
\sup_{\tau \in \mathcal{F}(\mathcal{H}_1\otimes \mathcal{H}_2)} \Trm \Big[\tau \Xi_{\alpha,z}(\rho_1 \otimes \rho_2,\sigma_{0,1} \otimes \sigma_{0,2})  \Big] = Q_{\alpha,z}(\rho_1 \otimes \rho_2 \|\sigma_{0,1}\otimes \sigma_{0,2}).
\end{align}
Now the point is that $Q_{\alpha,z}(\rho_1 \otimes \rho_2 \|\sigma_{0,1}\otimes \sigma_{0,2}) = Q_{\alpha,z}(\rho_1  \|\sigma_{0,1})Q_{\alpha,z}(\rho_2\|\sigma_{0,2})$ and since $|(1-\alpha)/z|=1$ we have $\Xi_{\alpha,z}(\rho_1 \otimes \rho_2,\sigma_{0,1} \otimes \sigma_{0,2}) = \Xi_{\alpha,z}(\rho_1,\sigma_{0,1})\otimes \Xi_{\alpha,z}(\rho_2,\sigma_{0,2})$. Therefore, by the assumption $\mathcal{F}(\mathcal{H}_1)^\circ \otimes \mathcal{F}(\mathcal{H}_2)^\circ
\subseteq \mathcal{F}(\mathcal{H}_1\otimes\mathcal{H}_2)^\circ$ and Lemma~\ref{lemma:multiplicativity of support function} for $k=2$, the support function is multiplicative and we have
\begin{align}
&\sup_{\tau \in \mathcal{F}(\mathcal{H}_1\otimes \mathcal{H}_2)} \Trm \Big[\tau \Xi_{\alpha,z}(\rho_1 \otimes \rho_2,\sigma_{0,1} \otimes \sigma_{0,2})  \Big] \\
&\qquad\qquad  = \sup_{\tau_1 \in \mathcal{F}(\mathcal{H}_1)} \Trm \Big[\tau_1 \Xi_{\alpha,z}(\rho_1,\sigma_{0,1})  \Big] \sup_{\tau_2 \in \mathcal{F}(\mathcal{H}_2)} \Trm \Big[\tau_2 \Xi_{\alpha,z}(\rho_2,\sigma_{0,2})  \Big] \\
&\qquad \qquad = Q_{\alpha,z}(\rho_1  \|\sigma_{0,1})Q_{\alpha,z}(\rho_2\|\sigma_{0,2}).
\end{align}
where the last inequality follows from the optimality of $\sigma_{0,1}$ and $\sigma_{0,2}$.   
\end{proof}

\subsubsection{The conditional entropy}

As an example to illustrate an application of the above results, we show that the conditional entropy is weakly additive.

Fix some bipartite state $\rho_{AB}$ as the null hypothesis and let the alternative hypothesis be determined by $\cF_n=\big\{\pi_A^{\otimes n}\otimes \sigma_{B}^{\n}:\, \sigma_{B}^{\n} \; \text{is a state}\big\}$, where $\pi_A = I_A/d_A$ is the fully mixed state.  It is not hard to verify that $\{\cF_n\}_{n\in \N}$ satisfies Assumption~\ref{assumption:polar}.

The $\alpha$-$z$-R\'enyi conditional entropy is defined by 
\begin{align}\label{eq:def-conditional-ent}
   H^\uparrow_{\alpha,z}(A|B)_{\rho_{AB}}= \inf_{\sigma_B}D(\rho_{AB}\|\pi_A \otimes \sigma_B)-\log{d_A}  \,.
\end{align}
We claim that this quantity is additive in the sense that 
\begin{align}
    H^\uparrow_{\alpha,z}(A^n|B^n)_{\rho_{AB}^{\otimes n}}
    = nH^\uparrow_{\alpha,z}(A|B)_{\rho_{AB}} \,.
\end{align} 
To prove this, let $\sigma_B=\sigma_0$ be an optimizer in~\eqref{eq:def-conditional-ent}, which is known to satisfy~\cite{rubboli2024quantum}
\begin{equation}\label{eq:conditional-entropy-optimizer}
    \sigma_B = \frac{1}{\gamma} \Tr_A\Big[(\pi_A \otimes \sigma_B)^\frac{1-\alpha}{2z}\rho_{AB}^{\frac{\alpha}{z}}(\pi_A \otimes \sigma_B)^\frac{1-\alpha}{2z}\Big]^{z} \,,
\end{equation}
where $\gamma = Q_{\alpha,z}(\rho_{AB}\|\pi_A \otimes \sigma_B)$ is a normalization factor. Next, we use the argument of Remark~\ref{rem:az additivity simplified} to write
\begin{equation}
    \chi_{\alpha,z}(\rho_{AB},\pi_A \otimes \sigma_B) =  (\pi_A \otimes \sigma_B)^{-\frac{1-\alpha}{2z}} \Big((\pi_A \otimes \sigma_B)^\frac{1-\alpha}{2z}\rho_{AB}^{\frac{\alpha}{z}}(\pi_A \otimes \sigma_B)^\frac{1-\alpha}{2z}\Big)^{z} (\pi_A \otimes \sigma_B)^{-\frac{1-\alpha}{2z}}. 
\end{equation}
Then, to verify condition~\eqref{condition on t}, we compute
\begin{align}
&\Trm\Big[ \chi_{\alpha,z}(\rho_{AB},\pi_A \otimes \sigma_B) (\pi_A\otimes \sigma_B)^{-\frac 12 \big(1-\frac{1-\alpha}{z}-it\big)} (\pi_A\otimes \tau_B) (\pi_A\otimes \sigma_B)^{-\frac 12 \big(1-\frac{1-\alpha}{z}+it\big)} \Big]\\
&\qquad\qquad \qquad  = \Trm\Big[ \Big((\pi_A \otimes \sigma_B)^\frac{1-\alpha}{2z}\rho_{AB}^{\frac{\alpha}{z}}(\pi_A \otimes \sigma_B)^\frac{1-\alpha}{2z}\Big)^{z} \,  \Big(I_A\otimes \sigma_B^{-\frac 12 (1-it)}  \tau_B  \sigma_B^{-\frac 12 (1+it)} \Big)\Big] \label{eq:comment-127}\\
&\qquad\qquad \qquad  = \Trm\Big[ \Tr_A\Big((\pi_A \otimes \sigma_B)^\frac{1-\alpha}{2z}\rho_{AB}^{\frac{\alpha}{z}}(\pi_A \otimes \sigma_B)^\frac{1-\alpha}{2z}\Big)^{z} \,  \Big( \sigma_B^{-\frac 12 (1-it)}  \tau_B  \sigma_B^{-\frac 12 (1+it)} \Big)\Big] \label{Fixed-point condition-0}\\
&\qquad\qquad \qquad =  \gamma\Trm\Big[ \sigma_{B} \sigma_B^{-\frac 12 (1-it)}  \tau_B  \sigma_B^{-\frac 12 (1+it)} \Big]\label{Fixed-point condition} \\ 
&\qquad\qquad\qquad  = Q_{\alpha,z}(\rho_{AB}\|\pi_A \otimes \sigma_B)\Trm[   \tau_B  ] \label{eq:gamma-definition}\\
&\qquad \qquad\qquad =Q_{\alpha,z}(\rho_{AB}\|\pi_A \otimes \sigma_B).
\end{align} 
where in~\eqref{Fixed-point condition} we use~\eqref{eq:conditional-entropy-optimizer} and in~\eqref{eq:gamma-definition} we use $\gamma = Q_{\alpha,z}(\rho_{AB}\|\pi_A \otimes \sigma_B)$. This gives the desired additivity.

Similarly, one could also consider the $\alpha$-$z$-R\'enyi mutual information corresponding to the family of sets $\cF_n=\{\rho_1^{\otimes n}\otimes \sigma_{B}^{\n}: \sigma_{B}^{\n} \; \text{is a state}\}$ and proceed analogously to prove additivity once deriving a fixed-point condition as in~\eqref{eq:conditional-entropy-optimizer}.

\section{Chernoff exponent}\label{sec:Chernoff}
In this section, we first characterize the Chernoff exponent for the generalized quantum hypothesis testing problem under Assumption~\ref{assumption:polar} based on a regularized quantity, and then explore when the regularization can be removed. 

In the setting of symmetric hypothesis testing, we are interested in the summation of errors $\alpha_n(\rho|T_n) + \beta_n(\cF_n|T_n)$, and our goal is to compute the associated error exponent:
\begin{align}\label{eq:def-Chernoff-exp}
\rmc(\rho\| \{\cF_n\}_{n\in \N}) := \liminf_{n\to +\infty} -\frac{1}{n} \log\inf \big\{\alpha_n(\rho|T_n) + \beta_n(\cF_n|T_n):\, 0\leq T_n\leq I   \big\}.
\end{align}
This error exponent is called the \emph{Chernoff exponent}.  
When the sets in the alternative hypothesis consist of a single state, i.e., $\mathcal{F}_n=\{\sigma^{\otimes n}\}$, the Chernoff exponent is given in terms of Petz R\'enyi relative entropy by~\cite{NussbaumSzkola2009chernoff, Audenaert+07discriminating}
\begin{align}\label{eq:Chernoff}
\rmc(\rho\| \sigma) = \sup_{\alpha\in (0,1)} (1-\alpha)D_\alpha(\rho\| \sigma).
\end{align}

The regularization problem for this error exponent, as well as the one considered later, requires further analysis due to the additional optimization over the parameter \( \alpha \). To address this, we first introduce some technical lemmas.

\subsection{Technical lemmas}

We begin with a technical lemma that will be useful for the achievability proof.
\begin{lemma}\label{lem:Audenaert-ineq}
Let $A$ be a positive operator and $\cF$ be a closed convex set of positive operators on a finite-dimensional Hilbert space. Then, for any $\alpha\in (0,1)$, there exists a test $0\leq T\leq I$ such that
\begin{align}\label{eq:Audenaert-ineq-lem}
 \textup{Tr}[A (I-T)] + \textup{Tr}[B T]\leq \max_{B'\in \cF} \,  \textup{Tr}[A^\alpha B'^{1-\alpha}], \qquad \forall B\in \cF.
 \end{align}
\end{lemma}

\begin{proof}
This theorem is a generalization of the celebrated result of~\cite{Audenaert+07discriminating}, establishing the same inequality in the special case where $\cF$ consists of a single operator. Besides the original proof in~\cite{Audenaert+07discriminating}, there is an elegant proof of this inequality in~\cite{Audenaert2014comparisons, Jakvsic+2012quantum}. See also~\cite{BeigiHircheTomamichel2025divergence} for a more recent proof. We give yet another proof of this inequality in Appendix~\ref{app:proof-Audenaert-ineq}. Having this special case, to prove~\eqref{eq:Audenaert-ineq-lem} we write
\begin{align}
\min_{0\leq T\leq I}\,\, \max_{B\in \cF} \Trm[A (I-T)] + \Trm[B T] 
& = \max_{B\in \cF} \,\,\min_{0\leq T\leq I}  \Trm[A (I-T)] + \Trm[B T]\\
& \leq \max_{B\in \cF} \,  \Trm[A^\alpha B^{1-\alpha}]. 
\end{align} 
Here, the equality follows from the minimax theorem, and the inequality follows from the special case of the theorem discussed above. 

\end{proof}

The following lemma is an extension of standard minimax theorems and will be used in the sequel.

\begin{lemma}\label{lem:minimax Farkas} \cite[Theorem 5.2]{farkas2006potential}
Let $X$ be a compact convex set in a topological vector space, and $Y$ be a convex subset of another vector space. Let $f : X \times Y \to \mathbb{R} \cup \{+\infty\}$ be such that
\medskip
\begin{enumerate}[itemsep=.8pt, topsep=1pt]
    \item  $f(x,\cdot)$ is concave on $Y$ for each $x \in X$, and
    \label{ciao}
    \item $f(\cdot,y)$ is convex and lower semi-continuous on $X$ for each $y \in Y$.
\end{enumerate}
\medskip
Then
\begin{equation}
\label{minimax Farkas}
	\inf_{x \in X} \sup_{y \in Y} f(x,y) = \sup_{y \in Y} \inf_{x \in X} f(x,y),
\end{equation}
and the infima in~\eqref{minimax Farkas} can be replaced by minima.
\end{lemma}

In the following lemma, we collect some properties of the reliability function of the Chernoff exponent.

\begin{lemma}\label{lem:func-chernoff}
Let $\rho$ be a quantum state and $\cF$ be a closed convex set. Define $\phi_{\rho}:[0,1]
\times \cF\to  \mathbb{R} \cup \{+\infty\}$ by
\begin{align}
\phi_{\rho}(\alpha, \sigma) = (1-\alpha) D_\alpha(\rho\| \sigma)= -\log \textup{Tr}\big[\rho^\alpha \sigma^{1-\alpha}\big],
\end{align}
if $\rho, \sigma$ are not orthogonal, and $\phi_{\rho}(\alpha, \sigma)=+\infty$ otherwise.
Then the followings hold:
\medskip
\begin{enumerate}[itemsep=1pt, topsep=1pt]
\item $\alpha \mapsto \phi_{\rho}(\alpha, \sigma)$ is concave and continuous on $[0,1]$, and  $\sigma \mapsto \phi_{\rho}(\alpha, \sigma)$ is convex and continuous for all $\alpha \in [0, 1)$. 
\label{item: chernoff (i)}
\item $\alpha\mapsto \inf_{\sigma \in \cF}\phi_{\rho}(\alpha, \sigma)$ is concave and continuous on $[0,1]$, and $\sigma\mapsto \sup_{\alpha \in [0,1)} \phi_{\rho}(\alpha, \sigma)$ is convex and lower semi-continuous on $\cF$. 
\label{item: chernoff (ii)}
\item $\inf_{\sigma\in \cF} \sup_{\alpha\in [0,1)} \phi_{\rho}(\alpha, \sigma) = \sup_{\alpha\in [0,1)} \inf_{\sigma\in \cF}  \phi_{\rho}(\alpha, \sigma)$ and the infima can be replaced by minima.
\label{item: chernoff (iii)}
\item The supremum over $\alpha \in [0,1]$ of the function $\inf_{\sigma\in \cF}\phi_{\rho}(\alpha, \sigma)$ and the infimum over $\cF$ of $ \sup_{\alpha\in [0,1)}\phi_{\rho}(\alpha, \sigma)$ are achieved. 
\label{item: chernoff (iv)}
\item Suppose that $\supp(\rho)\subseteq \supp(\sigma_0)$. Then, either $\argmax_{\alpha\in [0,1]}\phi_{\rho}(\alpha, \sigma_0)$ consists of a single point (i.e., $\alpha_0$ is unique) or $\rho=\sigma_0$.   
\label{item: chernoff (v)}
\end{enumerate}
\end{lemma}

\begin{proof}
We prove each part separately.

\medskip
\noindent
 \eqref{item: chernoff (i)} The continuities are clear from the definition of $\phi_{\rho}(\alpha, \sigma)$. The concavity and convexity properties are well known (see, e.g.,~\cite{Mosonyi+22errorexponent,Tomamichel}). 

\medskip
 \noindent
 \eqref{item: chernoff (ii)} This follows from the fact that the infimum of a concave function is concave and upper semi-continuous. Since $[0,1]$ is a closed interval, concavity and upper semi-continuity imply continuity on the whole interval.
Moreover, the pointwise supremum of a continuous function is lower semi-continuous, and the pointwise supremum of a convex function is convex.

\medskip
\noindent
 \eqref{item: chernoff (iii)} This is a consequence of the minimax theorem in Lemma~\ref{lem:minimax Farkas}. 

\medskip
\noindent
 \eqref{item: chernoff (iv)} By~\eqref{item: chernoff (ii)} the function $\alpha \mapsto \inf_{\sigma \in \cF}\phi_{\rho}(\alpha, \sigma)$ is continuous on the compact set $[0,1]$, and the supremum of a continuous function on a compact set is always achieved in the set. Moreover, the infimum of the function  $\sigma \mapsto \sup_{\alpha \in [0,1)}\phi_{\rho}(\alpha, \sigma)$ is achieved since the infimum of a lower semicontinuous function is always achieved in a compact set.

\medskip
\noindent
 \eqref{item: chernoff (v)} Suppose that $\alpha_1\neq \alpha_0$ is such that $\phi_{\rho}(\alpha_1, \sigma_0)=\phi_{\rho}(\alpha_0, \sigma_0)$. Then, the concavity of $\alpha\mapsto \phi_{\rho}(\alpha, \sigma_0)$ implies that $\phi_{\rho}(\alpha, \sigma_0)$ is a constant for any $\alpha$ between $\alpha_0, \alpha_1$. Hence, for any $\alpha$ between $\alpha_0,\alpha_1$ we must have
\begin{align}
0
  = \frac{\dd^2}{\dd \alpha^2}  \Trm\big[\rho^{\alpha}\sigma_0^{1-\alpha}\big]  = \Trm\big[(\log \rho - \log \sigma_0) \rho^{\alpha}(\log \rho - \log \sigma_0) \sigma_0^{1-\alpha}\big]  .
\end{align}
where the logarithms are taken on the support of the states. Letting $X=\rho^{\alpha/2}(\log \rho-\log \sigma_0)$, this is equivalent to $\Tr[X^\dagger X \sigma_0^{1-\alpha}]=0$. Since $\supp(\rho)\subseteq \supp(\sigma_0)$, this implies $X^\dagger X=0$ and $X=\rho^{\alpha/2}(\log \rho-\log \sigma_0)=0$, which in turn gives $\rho(\log \rho-\log \sigma_0)=0$. As a result, $D(\rho\| \sigma_0)=0$ and then $\rho=\sigma_0$.

\end{proof}

In the next section, we derive an additivity result for the Chernoff exponent, based on a single-copy criterion that can be verified at a saddle point. To this end, we provide a convenient characterization of such saddle points as solutions to specific outer optimization problems. We recall that a pair $(\alpha_0,\sigma_0)$ is called a saddle point of $\phi_{\rho}(\cdot, \cdot)$ if 
\begin{equation}
\label{saddle point Chernoff}
\sup_{\alpha \in [0,1]} \phi_{\rho}(\alpha,\sigma_0) =  \phi_{\rho}(\alpha_0,\sigma_0) = \inf_{\sigma\in \cF} \phi_{\rho}(\alpha_0,\sigma) \,,
\end{equation}
i.e., if “$\sigma_0$ minimizes against $\alpha_0$” and “$\alpha_0$ maximizes against $\sigma_0$”.
The next lemma, based on arguments similar to those presented in~\cite[Proposition 3.4.1]{bertsekas2009convex}, provides a procedure to construct such saddle points. 

\begin{lemma}
\label{lem: saddle point chernoff}
Let $\rho$ be a quantum state. A pair $(\alpha_0,\sigma_0)$ is a saddle point of $\phi_{\rho}(\cdot, \cdot)$ if and only if $\alpha_0$ and $\sigma_0$ are optimal solutions of the optimization problems
\begin{align}
\label{outer alpha0}
   & \alpha_0\in \argmax_{\alpha\in [0,1]} \inf_{\sigma\in \cF}\phi_{\rho}(\alpha, \sigma) ,\\
   \label{outer sigma0}
  & \sigma_0\in \argmin_{\sigma\in \cF} \sup_{\alpha\in [0,1)}\phi_{\rho}(\alpha, \sigma)\,.
\end{align} 
\end{lemma}

\begin{proof}
We first prove that if $\alpha_0$ and $\sigma_0$ are solutions of the outer optimizations in ~\eqref{outer alpha0} and~\eqref{outer sigma0}, then $(\alpha_0,\sigma_0)$ is a saddle point. Exploiting the fact that the optimizers are attained as stated in part~\eqref{item: chernoff (iv)} of Lemma~\ref{lem:func-chernoff}, we obtain 
\begin{align}
\sup_{\alpha \in [0,1)} \inf_{\sigma \in \cF}  \phi_{\rho}(\alpha,\sigma) &=\inf_{\sigma \in \cF} \phi_{\rho}(\alpha_0,\sigma) \\
&\leq \phi_{\rho}(\alpha_0,\sigma_0) \\
&\leq  \sup_{\alpha \in [0,1]} \phi_{\rho}(\alpha,\sigma_0) \\
&=\sup_{\alpha \in [0,1)} \phi_{\rho}(\alpha,\sigma_0) \\
&=  \inf_{\sigma \in \cF} \sup_{\alpha \in [0,1)} \phi_{\rho}(\alpha,\sigma) ,
\end{align}
where in the first and last equality, we used the definition of $\sigma_0$ and $\alpha_0$. By the minimax theorem of part~\eqref{item: chernoff (iii)}, all inequalities must be equalities. Hence, the inner equalities show that $(\alpha_0,\sigma_0)$ is a saddle point.

We now prove that if $(\alpha_0,\sigma_0)$ is a saddle point, then  $\alpha_0$ and $\sigma_0$ are solutions of the outer optimizations in~\eqref{outer alpha0} and~\eqref{outer sigma0}. We have 
\begin{align}
\label{outer optimization 1}
\inf_{\sigma \in \cF} \sup_{\alpha \in [0,1)}  \phi_{\rho}(\alpha,\sigma) &\leq \sup_{\alpha \in [0,1)} \phi_{\rho}(\alpha,\sigma_0) \\
&= \phi_{\rho}(\alpha_0,\sigma_0) \\
&= \inf_{\sigma \in \cF} \phi_{\rho}(\alpha_0,\sigma) \\
\label{outer optimization 2}
&\leq \sup_{\alpha \in [0,1)} \inf_{\sigma \in \cF} \phi_{\rho}(\alpha,\sigma_0) \,,
\end{align}
where in the inner equalities we used that $(\alpha_0,\sigma_0)$ is a saddle point.  Then, by the standard minimax inequality
\begin{equation}
	\sup_{\alpha \in [0,1)} \inf_{\sigma \in \cF} \phi_{\rho}(\alpha,\sigma_0) = \inf_{\sigma \in \cF} \phi_{\rho}(\alpha_0,\sigma_0) \leq \inf_{\sigma \in \cF} \sup_{\alpha \in [0,1)}  \phi_{\rho}(\alpha,\sigma)\,,
\end{equation}
all inequalities must be equalities. This means that $\alpha_0$ and $\sigma_0$ are solutions of the outer optimization problems.
\end{proof}

\subsection{Chernoff exponent and the additivity problem}

We first derive a characterization of the Chernoff exponent.

\begin{theorem}\label{thm:Chernoff}
Consider the quantum hypothesis testing problem~\eqref{eq:problem} for some state $\rho$ and suppose that the alternative hypothesis $\{\cF_n\}_{n\in \N}$ satisfies Assumption~\ref{assumption:polar} and Assumption~\ref{assumption:support}. Then, the Chernoff exponent is
\begin{align}\label{eq:Chernoff-composite}
\rmc(\rho\| \{\cF_n\}_{n\in \N}) & = \lim_{n\to +\infty} \frac 1n \sup_{\alpha\in (0,1)}\min_{\sigma^{\n}\in \cF_n}\,  (1-\alpha)D_\alpha\big(\rho^{\otimes n}\big\| \sigma^{\n}\big)\\
&=\lim_{n\to +\infty} \frac 1n \min_{\sigma^{\n}\in \cF_n}\sup_{\alpha\in (0,1)}\,  (1-\alpha)D_\alpha\big(\rho^{\otimes n}\big\| \sigma^{\n}\big),
\end{align}
where $\rmc(\rho\| \{\cF_n\}_{n\in \N})$ is defined in~\eqref{eq:def-Chernoff-exp} in terms of type I and type II errors given in~\eqref{eq:type-I-error} and~\eqref{eq:type-II-error}.
\end{theorem}

\begin{proof}
The fact that the order of supremum over $\alpha\in (0,1)$ and minimum over $\sigma^{\n}\in \cF_n$ can be exchanged is established in Lemma~\ref{lem:func-chernoff}. Also, the fact that the limits exist (and indeed can be replaced by infima) follows from Fekete's lemma.
Moreover, for the converse part, for any $\sigma^{(n_0)}\in \cF_{n_0}$, we may consider the hypothesis testing problem with the alternative hypothesis being determined by $\cF'_m=\{(\sigma^{(n_0)})^{\otimes m}\}$. By Assumption~\ref{assumption:polar} we have $\cF'_m\subseteq \cF_{mn_0}$, which by the definition of the Chernoff exponent implies 
\begin{align}
\rmc(\rho\| \{\cF_n\}_{n\in \N})\leq \frac 1{n_0} \rmc(\rho\| \{\cF'_m\}_{m\in \N})= \frac{1}{n_0} \sup_{\alpha\in (0,1)} (1-\alpha)D_\alpha(\rho^{\otimes n_0}\| \sigma^{(n_0)}),
\end{align}
where the equality follows from~\eqref{eq:Chernoff}. Taking the infimum over $n_0$ and $\sigma^{(n_0)}$, the desired converse bound is obtained.
To establish the achievability part, it suffices to apply Lemma~\ref{lem:Audenaert-ineq} for $A=\rho^{\otimes n}$ and the set $\cF_n$. 
\end{proof}

We note that the above theorem holds without requiring the condition \ref{A.4-polar} of Assumption~\ref{assumption:polar} on the polar sets as this condition is not used in the proof.

The next result addresses the additivity of the Chernoff exponent. In this setting, the additivity conditions must be verified at a (single-copy) saddle point of the reliability function. A convenient characterization of such saddle points is provided in Lemma~\ref{lem: saddle point chernoff}, which outlines a practical procedure to construct them by evaluating specific “sup” and “inf” functions. We finally emphasize that a tuple achieving the optimal value in the minimax inequality does not necessarily correspond to a saddle point.

\begin{theorem} \label{thm:additivity-Chernoff}
Consider the quantum hypothesis testing problem corresponding to the state $\rho$ and collection $\{\cF_n\}_{n\in \N}$ satisfying Assumption~\ref{assumption:polar} and Assumption~\ref{assumption:support}.
Let $(\alpha_0, \sigma_0)$ be a saddle point of the reliability function of the Chernoff exponent $\phi_\rho(\cdot, \cdot)$ defined in~\eqref{saddle point Chernoff}.  Then, the regularization in~\eqref{eq:Chernoff-composite} can be removed, i.e., 
\begin{equation}
\label{eq:no-regul-Chernoff}
\rmc(\rho\|\{\cF_n\}_{n\in \N}) = (1-\alpha_0)D_{\alpha_0}(\rho\| \sigma_0)
\end{equation}
if and only if 
\begin{equation}
\label{eq:add-ineq-Chernoff}
\alpha_0=0 \quad \text{ or }\quad \sup_{t \in \mathbb{R}, \tau \in \mathcal{F}_1}\textup{Tr}\Big[\tau \sigma_0^{-\frac{1}{2}(\alpha_0+it)}\rho^{\alpha_0} \sigma_0^{-\frac{1}{2}(\alpha_0-it)}\Big] = \textup{Tr}\big[\rho^{\alpha_0}\sigma^{1-\alpha_0}_0\big] \,.
\end{equation}
\end{theorem}

\begin{proof}
We first show how~\eqref{eq:add-ineq-Chernoff} implies~\eqref{eq:no-regul-Chernoff}. Considering the two cases $\alpha_0=0$ and $\alpha_0\neq 0$, by Theorem~\ref{Additivity in the line} and Theorem~\ref{thm:additivity-derivative}, condition~\eqref{eq:add-ineq-Chernoff} implies that
\begin{equation}
\label{additivity}
\min_{\sigma^{\n} \in \mathcal{F}_n }D_{\alpha_0}\big(\rho^{\otimes n}\big\| \sigma^{\n}\big) = n \min_{\sigma \in \mathcal{F}_1} D_{\alpha_0}(\rho\| \sigma),
\end{equation}
for every $n$. 
Moreover, since by Lemma~\ref{lem:func-chernoff} the tuple $(\alpha_0,\sigma_0)$ is a saddle point, we have
\begin{align}
\label{by now}
&\max_{\alpha \in [0,1]} \phi_{\rho}(\alpha,\sigma_0)=\phi_{\rho}(\alpha_0,\sigma_0),
\end{align}
where as before for simplicity of notation we use $\phi_{\rho}(\alpha,\sigma) =(1-\alpha)D_\alpha(\rho\|\sigma)$. 
Next, replacing the sets with an instance of the sets, we obtain
\begin{align}
\label{interline}
&\min_{\sigma^{\n} \in \cF_n} \max_{\alpha \in [0,1]} \phi_{\rho^{\otimes n}}\big(\alpha, \sigma^{\n}\big) \leq \max_{\alpha \in [0,1]} \phi_{\rho^{\otimes n}}\big(\alpha, \sigma_0^{\otimes n}\big) = n \phi_{\rho}(\alpha_0,\sigma_0), \\
\label{additivity using}
& \max_{\alpha \in [0,1]} \min_{\sigma^{\n} \in \cF_n} \phi_{\rho^{\otimes n}}\big(\alpha,\sigma^{\n}\big)  \geq  \min_{\sigma^{\n} \in \cF_n} \phi_{\rho^{\otimes n}}\big(\alpha_0,\sigma^{\n}\big) = n \phi_{\rho}(\alpha_0,\sigma_0),
\end{align}
where the equalities follow from~\eqref{additivity} and~\eqref{by now}. By the minimax theorem, the two quantities are equal, implying that
\begin{equation}
\rmc(\rho\|\{\cF_n\}_{n\in \N}) = \lim_{n \rightarrow \infty}\frac{1}{n} \min_{\sigma^{\n} \in \cF_n }\max_{\alpha \in [0,1]} \phi_{\rho^{\otimes n}}\big(\alpha,\sigma^{\n}\big) = \phi_{\rho}(\alpha_0,\sigma_0).
\end{equation}

We now prove that~\eqref{eq:no-regul-Chernoff} implies \eqref{eq:add-ineq-Chernoff}. 
By the additivity of $D_\alpha(\cdot\| \cdot)$ under tensor product states, for any $n$, it holds the upper bound
\begin{equation}
    \frac 1n \min_{\sigma^{\n}\in \cF_n} \max_{\alpha} \phi_{\rho^{\otimes n}}\big(\alpha,\sigma^{\n}\big) \leq \phi_{\rho}(\alpha_0,\sigma_0),
\end{equation}
and by Theorem~\ref{thm:Chernoff} it tends to $\rmc(\rho\| \{\cF_n\}_{n\in \N})$. Then,~\eqref{eq:no-regul-Chernoff} implies  
\begin{equation}
\label{finite number}
\frac{1}{n} \min_{\sigma^{\n} \in \cF_n }\max_{\alpha \in [0,1]} \phi_{\rho^{\otimes n}}\big(\alpha, \sigma^{\n}\big) = \phi_{\rho}(\alpha_0,\sigma_0).
\end{equation} 
Define $\big(\alpha_0^{\n}, \sigma_0^{\n}\big)\in [0,1]\times \cF_n$ as saddle points for the function $\phi_{\rho^{\otimes n}}(\cdot, \cdot)$. We refer to Lemma~\ref{lem: saddle point chernoff} for a characterization as solution of outer optimization problems. Then,~\eqref{finite number} implies 
\begin{align}\label{eq:add-implication-Chernoff}
n\phi_{\rho}(\alpha_0,\sigma_0) = \phi_{\rho^{\otimes n}}\big(\alpha_0^{\n},\sigma^{\n}_0\big).
\end{align}
Since both $(\alpha_0, \sigma_0)$ and $\big(\alpha_0^{\n}, \sigma_0^{\n}\big)$ are saddle points, we have 
\begin{equation}
    \phi_{\rho}(\alpha_0,\sigma_0) \geq \phi_{\rho}\big(\alpha_0^{\n},\sigma_0\big),
\end{equation}
and
\begin{equation}
    \phi_{\rho^{\otimes n}}\big(\alpha_0^{\n},\sigma_0^{\n}\big) \leq \phi_{\rho^{\otimes n}}\big(\alpha_0^{\n},\sigma_0^{\otimes n}\big).
\end{equation}
Comparing these two equations with~\eqref{eq:add-implication-Chernoff}, we find that both the above inequalities are tight. Hence, we have
\begin{equation}
\label{equalities}
\phi_{\rho}(\alpha_0,\sigma_0) = \phi_{\rho}\big(\alpha_0^{\n},\sigma_0\big).
\end{equation} 
Therefore, by part~\eqref{item: chernoff (v)} of Lemma~\ref{lem:func-chernoff} either $\alpha_0^{\n}= \alpha_0$ or $\rho=\sigma_0$. In the latter case,~\eqref{eq:add-ineq-Chernoff} is obtained using Corollary~\ref{commuting}. On the other hand, the former case implies that
\begin{equation}
\inf_{\sigma^{\n}\in \cF_n}\phi_{\rho^{\otimes n}}\big(\alpha_0,\sigma^{\n}\big)= n\phi_{\rho}(\alpha_0,\sigma_0).
\end{equation}
Therefore, by Theorem~\ref{thm:additivity-derivative}, the desired inequality in~\eqref{eq:add-ineq-Chernoff} holds if $\alpha_0\neq 0$.

\end{proof}

\section{Hoeffding exponent and the additivity problem}\label{sec:Hoeffding}
In this section, we deal with the characterization of the Hoeffding exponent and its additivity problem.

The trade-off between the two errors, assuming that both of them decay at an exponential rate, is given by
\begin{align}\label{eq:def-Hoeffding-exp}
\rmd_r(\rho\| \{\cF_n\}_{n\in \N}) := \sup\left\{ \liminf_{n\to +\infty} -\frac 1n \log\alpha_n(\rho| T_n):\,   \liminf_{n\to +\infty} -\frac 1n\log \beta_n(\cF_n|T_n)\geq r, \, 0\leq T_n\leq I\right\}.
\end{align}
This error exponent is called the \emph{direct exponent} or the \emph{Hoeffding exponent}. 

When $\mathcal{F}_n=\{\sigma^{\otimes n}\}$ for some state $\sigma$, the Hoeffding exponent is given by~\cite{Hayashi2007error, Nagaoka2006converse}
\begin{align}\label{eq:Hoeffding}
\rmd_r(\rho\| \sigma) = \sup_{\alpha\in (0,1]} \frac{\alpha-1}{\alpha}\big( r-D_\alpha(\rho\| \sigma)\big).
\end{align}

\subsection{A technical lemma}

We now establish some properties of the reliability function associated with the Hoeffding exponent. 
This lemma is analogous to Lemma~\ref{lem:func-chernoff}.

\begin{lemma}\label{lem:func-Hoeffding}
Let $\rho$ be a quantum state and $\cF$ be a compact convex subset of quantum states. Fix $r>0$ and define $\psi_{r, \rho}:(0,1]
\times \cF\to \mathbb R\cup\{+\infty\}$ by
\begin{align}\label{eq:psi-hoeffidng-def}
\psi_{r, \rho}(\alpha, \sigma) = \frac{\alpha-1}{\alpha}\big(r- D_\alpha(\rho\| \sigma)\big)= \frac{\alpha-1}{\alpha} r -\frac 1\alpha \log \textup{Tr}\big[\rho^\alpha \sigma^{1-\alpha}\big],
\end{align}
if $\rho, \sigma$ are not orthogonal, and $\psi_{r, \rho}(\alpha, \sigma)=+\infty$ otherwise.
Then the followings hold:
\begin{enumerate}
\item $\alpha\mapsto \psi_{r, \rho}(\alpha, \sigma)$ is continuous on $(0,1]$, and $\sigma\mapsto \psi_{r, \rho}(\alpha, \sigma)$ is continuous on $\cF$ if $\alpha\in (0,1)$ and lower semi-continuous if $\alpha=1$. 
\label{item: Hoeffding (i)}
\item$\sigma\mapsto \psi_{r, \rho}(\alpha, \sigma)$ is convex for any $\alpha\in (0,1]$.
\label{item: Hoeffding (ii)}

\item $\beta\mapsto \psi_{r, \rho}(1/\beta, \sigma)$ is concave on $[1, +\infty)$. Moreover, if $\supp(\rho)\subseteq \supp(\sigma)$, then
\begin{align}\label{eq:psi-equality-case}
\psi_{r, \rho}\Big(\frac{1-\theta}{\beta_0} + \frac{\theta}{\beta_1}, \sigma\Big)=(1-\theta)\psi_{r, \rho}\Big(\frac{1}{\beta_0}, \sigma\Big) +\theta\psi_{r, \rho}\Big(\frac{1}{\beta_1}, \sigma\Big),
\end{align}
for some $\beta_0\neq \beta_1$ and $\theta\in(0,1)$ imply that $\rho, \sigma$ commute.
\label{item: Hoeffding (iii)}
 
\item $\min_{\sigma\in \cF} \sup_{\alpha\in (0,1]} \psi_{r, \rho}(\alpha, \sigma) = \sup_{\alpha\in (0,1]} \min_{\sigma\in \cF}  \psi_{r, \rho}(\alpha, \sigma)$.
\label{item: Hoeffding (iv)}
\item Suppose that $r>D_0(\rho\| \cF)= \min_{\sigma\in \cF} D_0(\rho\| \sigma)$. Then, the supremum over $\alpha\in (0,1]$ of the function $\min_{\sigma\in \cF} \psi_{r, \rho}(\alpha, \sigma)$ is achieved. 
\label{item: Hoeffding (v)}
\end{enumerate}
\end{lemma}

\bigskip 
\begin{proof}
We prove each part separately.

\medskip
 \noindent
\eqref{item: Hoeffding (i)} This is clear from the definition of $\psi_{r, \rho}(\cdot, \cdot)$.

\medskip
 \noindent
\eqref{item: Hoeffding (ii)} This is equivalent to the convexity of $\sigma\mapsto D_\alpha(\rho\|\sigma)$ and is well known.

\medskip
 \noindent
\eqref{item: Hoeffding (iii)} The concavity of part~\eqref{item: Hoeffding (iii)} is also well known, yet here, in order to handle the equality case, we give a proof. 
Let $\rho = \sum_i \lambda_i P_i$ and $\sigma = \sum_k \mu_k Q_k$ with $\{P_i:\, i\}$ and $\{Q_k:\, k\}$ being orthogonal projections and $\lambda_i, \mu_k>0$ being distinguished non-zero eigenvalues of $\rho, \sigma$ respectively.  Then, we have 
\begin{equation}
    \psi_{r, \rho}(\alpha, \sigma) = r - \log \|fg^{1/\alpha}\|_\alpha,
\end{equation}
where $f, g$ are vectors with coordinates $f_{ik} = \lambda_i/\mu_k$ and $g_{ik} = e^{r}\mu_k \tr[P_iQ_k]$, with $\|\cdot\|_\alpha$ being the $\alpha$-norm: $\|v\|_\alpha = (\sum_j |v_j|^\alpha)^{1/\alpha}$. Now suppose that $\alpha_0, \alpha_1\in (0,1]$ and for $\theta\in (0,1)$ define $\alpha_\theta$ by 
\begin{align}
    \frac{1}{\alpha_\theta} = \frac{1-\theta}{\alpha_0} + \frac{\theta}{\alpha_1}.
\end{align}
By H\"older's inequality, we have 
\begin{align}
\label{eq:concavity-Holder}
\big\|fg^{1/\alpha_\theta}\big\|_{\alpha_\theta} & \leq \big\|f^{1-\theta}g^{(1-\theta)/\alpha_0}\big\|_{\alpha_0/(1-\theta)} \big\|f^\theta g^{\theta/\alpha_1}\big\|_{\alpha_1/\theta} \leq  \big\|fg^{1/\alpha_0}\big\|_{\alpha_0}^{1-\theta} \big\|f g^{1/\alpha_1}\big\|_{\alpha_1}^{\theta}.
\end{align}
Taking the logarithm of both sides, the concavity of $\beta\mapsto \psi_{r, \rho}(1/\beta, \sigma)$ is implied. 

Now suppose we have equality as in~\eqref{eq:psi-equality-case}. This means that we have equality in the above H\"older's inequality, and the equality condition of H\"older's inequality yields
\begin{equation}
    \left(f^{1-\theta}g^{(1-\theta)/\alpha_0}\right)^{\alpha_0/(1-\theta)} = c  \left(f^{\theta}g^{\theta/\alpha_1}\right)^{\alpha_1/\theta} 
\end{equation}
for some constant $c$. This equivalently means that  $f_{ik}^{\alpha_0}g_{ik} = c f_{ik}^{\alpha_1} g_{ik}$ for any $i, k$. Therefore, for any $i, k$ satisfying $\tr[P_iQ_k]\neq 0$ we have $(\lambda_i/\mu_k)^{\alpha_0-\alpha_1} =c$. 
Since by definition $\lambda_i$'s and $\mu_k$'s are unequal, this means that for any $i$ there is at most one $k$ such that $\tr[P_iQ_k]\neq 0$. On the other hand, by the assumption $\supp(\rho)\subseteq \supp(\sigma)$, for any $i$ there is at least one $k$ satisfying $\tr[P_iQ_k]\neq 0$. Putting these together, we conclude that for any $i$ there is a unique $k_i$ such that $P_i\leq Q_{k_i}$ and $P_iQ_k=Q_kP_i=0$ for any $k\neq k_i$. In any case, $P_i$ commutes with $Q_k$ for all $k$, which implies that $\rho, \sigma$ commute. 

\medskip
 \noindent
\eqref{item: Hoeffding (iv)} This is implied by Lemma~\ref{lem:minimax Farkas} and~\eqref{item: Hoeffding (ii)} and~\eqref{item: Hoeffding (iii)}.

\medskip
 \noindent
 \eqref{item: Hoeffding (v)} The fact that $\sigma_0$ is well-defined is already established in~\eqref{item: Hoeffding (iv)}. To prove that $\alpha_0$ is well-defined we use the condition $r>D_0(\rho\| \cF)$. First note that for any $\sigma\in \cF$ and $\alpha\in (0, 1]$ we have $\sigma^{1-\alpha}\leq I$. Therefore, 
\begin{align}
    \sup_{\alpha\in (0, 1]}\min_{\sigma\in \cF} \psi_{r, \rho}(\alpha, \sigma)\geq \sup_{\alpha\in (0, 1]} \frac{\alpha-1}{\alpha}r - \frac{1}{\alpha}\log \Trm[\rho^\alpha]\geq 0,
\end{align}
where the last inequality follows by taking $\alpha=1$. Next, we note that $\alpha\mapsto D_\alpha(\rho\|\sigma)$ is continuous, particularly at $\alpha=0$. Thus, $D_0(\rho\| \cF)<r$ implies that for a minimizer $\sigma'\in \argmin_{\sigma\in \cF}D_0(\rho\| \sigma)$ we have $D_\alpha(\rho\|\sigma')<r$ in some neighborhood $\alpha\in [0, \epsilon)$. Comparing to the above inequality, we find that the supremum of $\alpha\mapsto \min_{\sigma\in \cF}\psi_{r, \rho}(\alpha, \sigma)=\frac{\alpha-1}{\alpha}(r-D_\alpha(\rho\|\sigma))$ cannot be achieved in the interval $(0,\epsilon)$. We conclude that
\begin{align}
    \sup_{\alpha\in (0, 1]}\min_{\sigma\in \cF} \psi_{r, \rho}(\alpha, \sigma) = \sup_{\alpha\in [\epsilon, 1]}\min_{\sigma\in \cF} \psi_{r, \rho}(\alpha, \sigma),
\end{align} 
implying that the supremum is achieved.
\end{proof}

Suppose that $r>D_0(\rho\| \cF)= \min_{\sigma\in \cF} D_0(\rho\| \sigma)$. A pair $(\alpha_0,\sigma_0)$ is called a saddle point of $\psi_{r,\rho}(\cdot, \cdot)$ if 
\begin{equation}
\label{saddle point Hoeffding}
\sup_{\alpha \in (0,1]} \psi_{r,\rho}(\alpha,\sigma_0) =  \psi_{r,\rho}(\alpha_0,\sigma_0) = \inf_{\sigma\in \cF} \psi_{r,\rho}(\alpha_0,\sigma) \,,
\end{equation}
Similarly to the Chernoff exponent, the next lemma provides a procedure to construct saddle points by evaluating some specific “sup” and “inf” functions.

\begin{lemma}
\label{lem: saddle point Hoeffding}
Let $\rho$ be a quantum state and suppose that $r>D_0(\rho\| \cF)= \min_{\sigma\in \cF} D_0(\rho\| \sigma)$. A pair $(\alpha_0,\sigma_0)$ is a saddle point of $\psi_{r,\rho}$ if and only if $\alpha_0$ and $\sigma_0$ are optimal solutions of the optimization problems
\begin{align}
   & \alpha_0\in \argmax_{\alpha\in (0,1]} \inf_{\sigma\in \cF}\psi_{r,\rho}(\alpha, \sigma) ,\\
  & \sigma_0\in \argmin_{\sigma\in \cF} \sup_{\alpha\in (0,1]}\psi_{r,\rho}(\alpha, \sigma)\,.
\end{align} 
\end{lemma}
\begin{proof}
    The proof follows the same reasoning as the one for the Chernoff exponent in Lemma~\ref{lem: saddle point chernoff} and builds on the observations established in Lemma~\ref{lem:func-Hoeffding}.
\end{proof}

\subsection{Hoeffding exponent and  regularization}

\begin{theorem}\label{thm:Hoeffding}
Consider the quantum hypothesis testing problem~\eqref{eq:problem} for some state $\rho$ and suppose that the alternative hypothesis $\{\cF_n\}_{n\in \N}$ satisfies Assumption~\ref{assumption:polar} and Assumption~\ref{assumption:support}.
Let $\max\{0, D_0(\rho\| \cF_1)\}<r<\rms(\rho\|\{\cF_n\}_{n\in \N})$ where $D_0(\rho\| \cF_1) = \inf_{\sigma\in \cF_1} D_0(\rho\| \sigma)$ and $\rms(\rho\|\{\cF_n\}_{n\in \N})$ is the Stein exponent. Then, the Hoeffding exponent is
\begin{align}\label{eq:Hoeffding-composite}
\rmd_r(\rho\| \{\cF_n\}_{n\in \N}) &= \lim_{n\to +\infty}  \sup_{\alpha\in (0,1]}\min_{\sigma^{\n}\in \cF_n}\, \frac{\alpha-1}{\alpha}\Big(r-\frac 1nD_\alpha\big(\rho^{\otimes n}\big\|\sigma^{\n} \big)\Big)\\
&= \lim_{n\to +\infty} \min_{\sigma^{\n}\in \cF_n} \sup_{\alpha\in (0,1]}\, \frac{\alpha-1}{\alpha}\Big(r-\frac 1nD_\alpha\big(\rho^{\otimes n}\big\|\sigma^{\n} \big)\Big),
\end{align}
where $\rmd_r(\rho\| \{\cF_n\}_{n\in \N}\})$ is defined in~\eqref{eq:def-Hoeffding-exp} in terms of type I and type II errors given in~\eqref{eq:type-I-error} and~\eqref{eq:type-II-error}.
\end{theorem}

Note that if $r>\rms(\rho\| \{\cF_n\}_{n\in \N})$, then $\rmd_r(\rho\| \{\cF_n\}_{n\in \N})=0$. Also, it is not hard to verify that $D_0(\rho\| \cF_1)$ is additive and if $r<D_0(\rho\| \cF_1)$ then $\rmd_r(\rho\| \{\cF_n\}_{n\in \N})=+\infty$. See~\cite[Lemma II.9]{Mosonyi+22errorexponent} for more details. This is why in the above theorem we restrict to $D_0(\rho\| \cF_1)<r<\rms(\rho\|\{\cF_n\}_{n\in \N})$.

\begin{proof}
The fact that the order of supremum and minimum can be exchanged is already established in Lemma~\ref{lem:func-Hoeffding}. Also, the limits exist by Fekete's lemma. 
We skip the proof of the converse as it is similar to the proof of the converse part of Theorem~\ref{thm:Chernoff} and is based on~\eqref{eq:Hoeffding}. For the achievability we apply Lemma~\ref{lem:Audenaert-ineq} for $A=\rho^{\otimes n}$ and the set $e^x\cF_n = \big\{e^x \sigma^{\n}:\, \sigma^{\n}\in \cF_n\big\}$ for some real parameter $x$ to be determined. We find that there is a test $0\leq T_n\leq I$ such that for any $\sigma^{\n}\in \cF_n$,
\begin{align}
1- \tr[T_n\rho^{\otimes n}] + e^x \tr\big[T_n \sigma^{\n}\big]\leq    e^{(1-\alpha) x} \sup_{\tau^{\n}\in \cF_n} \Trm\big[\big(\rho^{\otimes n}\big)^\alpha \big(\tau^{\n}\big)^{1-\alpha}\big].
\end{align}
Letting
\begin{equation}
    x = \frac{1}{\alpha} \Big(r+\frac 1n\sup_{\tau^{\n}\in \cF_n}\log \Trm\big[\big(\rho^{\otimes n}\big)^\alpha \big(\tau^{\n}\big)^{1-\alpha}\big] \Big),
\end{equation}
yields the desired bound. 
\end{proof}

We once again note that the above theorem holds without requiring the condition \ref{A.4-polar} of Assumption~\ref{assumption:polar} on the polar sets as this condition is not used in the proof.

We now present the additivity condition for the Hoeffding exponent. As in the case of the Chernoff exponent, verifying additivity requires identifying a (single-copy) saddle point of the reliability function. A practical procedure for constructing such saddle points is again based on evaluating specific “sup” and “inf” functions, as detailed in Lemma~\ref{lem: saddle point Hoeffding}.

\begin{theorem}\label{thm:additivity-Hoeffding}
Consider the quantum hypothesis testing problem corresponding to the state $\rho$ and collection $\{\cF_n\}_{n\in \N}$ satisfying Assumption~\ref{assumption:polar} and Assumption~\ref{assumption:support}. Let $\max\{0, D_0(\rho\| \cF_1)\}<r<\rms(\rho\|\{\cF_n\}_{n\in \N})$ and $(\alpha_0, \sigma_0)$ be a saddle point of the reliability function of the Hoeffding exponent defined in~\eqref{saddle point Hoeffding}.  Then, the regularization in~\eqref{eq:Hoeffding-composite} can be removed, i.e., 
\begin{align}\label{eq:additivity-Hoeffding}
\rmd_r(\rho\| \{\cF_n\}_{n\in \N}) = \frac{\alpha_0-1}{\alpha_0}\big(r- D_{\alpha_0}(\rho\| \sigma_0)\big),
\end{align}
if and only if 
\begin{align}\label{eq:additivity-Hoeffding-condition}
\sup_{t\in \mathbb R, \tau\in \cF_1}\textup{Tr}\Big[ \rho^{\alpha_0} \sigma_0^{-\frac 12 (\alpha_0+it)} \tau \sigma_0^{-\frac 12 (\alpha_0-it)} \Big] =  \textup{Tr}\big[\rho^{\alpha_0}\sigma_0^{1-\alpha_0}\big].
\end{align}
\end{theorem}

\begin{proof}
The proof that~\eqref{eq:additivity-Hoeffding-condition} implies~\eqref{eq:additivity-Hoeffding} is similar to that of Theorem~\ref{thm:additivity-Chernoff}, and we skip it. For the other direction, suppose that additivity~\eqref{eq:additivity-Hoeffding} holds.
Let $\alpha_0^{\n}, \sigma_0^{\n}$ be saddle points for the $n$-th level optimization problem for the constant $nr$ and set $\cF_n$, i.e., for the function $\psi_{nr, \rho^{\otimes n}}:(0, 1]\times \cF_n\to \mathbb R\cup\{+\infty\}$ defined in~\eqref{eq:psi-hoeffidng-def}. We refer to Lemma~\ref{lem: saddle point Hoeffding} for a characterization as solutions of outer optimization problems. We note that the assumptions of part~\eqref{item: Hoeffding (v)} are valid for $\psi_{nr, \rho^{\otimes n}}(\cdot, \cdot)$ since by Theorem~\ref{Additivity in the line} we have $nr> nD_0(\rho\| \cF_1)= D_0(\rho^{\otimes n}\| \cF_n)$. 

Now, by the additivity assumption~\eqref{eq:additivity-Hoeffding} we have 
\begin{align}\label{eq:equality-additivity-Hoeffding-pr}
\psi_{nr, \rho^{\otimes n}}\big(\alpha_0^{\n}, \sigma_0^{\n}\big) = n\psi_{r, \rho}(\alpha_0,\sigma_0)\,.
\end{align} 
Since $(\alpha_0, \sigma_0)$ and $\big(\alpha_0^{\n}, \sigma_0^{\n}\big)$ are saddle points, we obtain 
\begin{align}\label{eq:ii-1st-eq-condition}
\psi_{r, \rho}(\alpha_0,\sigma_0) \geq  \psi_{r, \rho}\big(\alpha_0^{\n}, \sigma_0\big),
\end{align}
and
\begin{align}\label{eq:ii-2dn-eq-condition}
\psi_{nr, \rho^{\otimes n}}\big(\alpha_0^{\n},\sigma^{\n}_0\big) \leq  \psi_{nr, \rho^{\otimes n}}\big(\alpha_0^{\n},\sigma_0^{\otimes n}\big).
\end{align}
Comparing these two equations with the equality~\eqref{eq:equality-additivity-Hoeffding-pr}, we conclude that both the above inequalities are tight, and we have 
\begin{align}
\psi_{r, \rho}(\alpha_0,\sigma_0) = \psi_{r, \rho}\big(\alpha_0^{\n}, \sigma_0\big).
\end{align}
Suppose that for some $n$, $\alpha_0^{\n}\neq \alpha_0$. Then, the optimality of $\alpha_0$ and part~\eqref{item: Hoeffding (iii)} of Lemma~\ref{lem:func-Hoeffding} together with $\supp(\rho)\subseteq \supp(\sigma_0)$ that follows from Theorem~\ref{thm:derivative-condition}, imply that $\psi_{r, \rho}(\alpha, \sigma_0)$ is a constant for any $\alpha$ between $\alpha_0, \alpha_0^{\n}$ and $\rho, \sigma_0$ commute.  
Then,~\eqref{eq:additivity-Hoeffding-condition} holds since by Theorem~\ref{thm:derivative-condition} it is equivalent to the optimality condition for $\sigma_0$. Therefore, assume that for all $n$, $\alpha_0^{\n}=\alpha_0$. This means that 
\begin{align}
\min_{\sigma^{\n}}     \psi_{nr, \rho^{\otimes n}}\big(\alpha_0,\sigma^{\n}\big)  = n\psi_{r, \rho}(\alpha_0,\sigma_0). 
\end{align}
Then, the desired inequality~\eqref{eq:additivity-Hoeffding-condition} follows by Theorem~\ref{thm:additivity-derivative}. 

\end{proof}

\section{Strong converse exponent}\label{sec:Strong-Converse}

In this section, we introduce conditions under which the strong converse exponent of a quantum hypothesis testing problem with an alternative hypothesis satisfying Assumption~\ref{assumption:polar} and Assumption~\ref{assumption:support}
is given by a single-letter formula.

The \emph{strong converse exponent} is defined as follows. By the definition of the Stein exponent, if the type II error decays with an exponential rate larger than the Stein exponent, then the type I error does not go to zero. The strong converse property says that in this case, the type I error tends to $1$. The strong converse exponent characterizes this convergence rate:
\begin{align}\label{eq:def-strong-coverse-exp}
\rmsc_r(\rho\| \{\cF_n\}_{n\in \N}) := \inf\left\{ \limsup_{n\to +\infty} -\frac 1n \log(1-\alpha_n(\rho| T_n)):\,   \liminf_{n\to +\infty} -\frac 1n\log \beta_n(\cF_n|T_n)\geq r\right\},
\end{align}
where the infimum is taken over sequences of tests $\{T_n:\, n\in \N\}$ satisfying $0\leq T_n\leq I$.
It is not hard to verify that under Assumption~\ref{A.4-polar}, the strong converse exponent can be equivalently written as 
\begin{align}\label{eq:def-strong-coverse-exp2}
\rmsc_r(\rho\| \{\cF_n\}_{n\in \N}) 
&= \inf\left\{ \liminf_{n\to +\infty} -\frac 1n \log(1-\alpha_n(\rho| T_n)):\,   \liminf_{n\to +\infty} -\frac 1n\log \beta_n(\cF_n|T_n)\geq r\right\}\\
&= \inf\left\{ \lim_{n\to +\infty} -\frac 1n \log(1-\alpha_n(\rho| T_n)):\,   \liminf_{n\to +\infty} -\frac 1n\log \beta_n(\cF_n|T_n)\geq r\right\},
\end{align}
where in the second line, the infimum is taken only over test sequences for which the indicated limit exists; see, e.g., \cite{bunth2025}. 
For more details on these definitions, we refer to~\cite{Mosonyi+22errorexponent}. The strong converse exponent when $\mathcal{F}_n=\{\sigma^{\otimes n}\}$ for a fixed state $\sigma$, is given in terms of the sandwiched R\'enyi relative entropy by~\cite{mosonyi2015quantum}
\begin{align}\label{eq:Strong-Converse}
\rmsc_r(\rho\| \sigma) = H_r^*(\rho\|\sigma):=\max_{\alpha\in[1,+\infty]} \frac{\alpha-1}{\alpha}\Big( r-\widetilde D_\alpha(\rho\| \sigma)\Big),
\end{align}
where $H_r^*(\cdot \| \cdot)$ is the \emph{Hoeffding anti-divergence}, and $\frac{\alpha-1}{\alpha}$ for $\alpha=+\infty$ is defined to be $1$. For a more general alternative hypothesis represented by a family $\{\cF_n\}_{n\in \N}$, expressions for the strong converse exponent are naturally sought in terms of 
\begin{align}
 H_r^*(\rho^{\otimes n}\|\cF_n)&:=\max_{\sigma^{\n}\in\cF_n} H_r^*(\rho^{\otimes n}\|\sigma^{(n)})=
\max_{\alpha\in[1,+\infty]}\max_{\sigma^{(n)}\in\cF_n} \frac{\alpha-1}{\alpha}\Big( r-\widetilde D_\alpha\big(\rho^{\otimes n}\| \sigma^{(n)}\big)\Big)\\
&=
\max_{\alpha\in[1,+\infty]} \frac{\alpha-1}{\alpha}\Big( r-\min_{\sigma^{(n)}\in\cF_n}\widetilde D_\alpha\big(\rho^{\otimes n}\| \sigma^{(n)}\big)\Big).
\end{align}
For the existence of the maxima in the above expressions see \cite[Corollary V.17]{Mosonyi+22errorexponent}.

 \subsection{Strong converse exponent and regularization}

Characterizing the strong converse exponent $\rmsc_r(\rho\| \{\cF_n\}_{n\in \N})$ does not seem as easy as in the case of the Chernoff and Hoeffding exponents. This is because the strong converse exponent, even for simple hypotheses, is achieved only in the asymptotic limit, while for the Chernoff and the Hoeffding exponents, we obtained one-shot achievability bounds using Lemma~\ref{lem:Audenaert-ineq}. This is similar to the Stein exponent, which is also achieved only in the asymptotic limit. Thus, one idea towards the characterization of the strong converse exponent is to follow the proof idea of Theorem~\ref{thm:Stein} and use measured quantities. Taking this approach,
we prove the following proposition in  Appendix~\ref{app:strong-converse}.

\begin{prop}\label{prop:sc bounds}
Consider the quantum hypothesis testing problem~\eqref{eq:problem}, and assume that Assumption~\ref{assumption:polar} and Assumption~\ref{assumption:support} are satisfied. For any $r>0$, the strong converse exponent of this problem can be bounded as
\begin{align}
 &\lim_{n\to +\infty}\frac{1}{n}H_{nr}^*(\rho\|\cF_n)\label{eq:sc lower}\\
&\ds =\sup_{n\in\bN}\sup_{\alpha>1} \sup_{\sigma^{\n}\in \cF_n} \frac{\alpha-1}{\alpha}\Big(r-\frac 1n \widetilde D_\alpha\big(\rho^{\otimes n} \big\| \sigma^{\n}\big)  \Big)\\
&\ds\le\rmsc_r(\rho\| \{\cF_n\}_{n\in \N})\\
&\ds\le
\inf_{n\in \N}  \inf_{\mathcal M_n} \sup_{\alpha>1} \sup_{\sigma^{\n}\in \cF_n} \frac{\alpha-1}{\alpha}\Big(r-\frac 1n D_\alpha\big(\mathcal M_n(\rho^{\otimes n}) \big\| \mathcal M_n(\sigma^{\n})\big)  \Big), \label{eq:sc upper}
\end{align}
where $\cM_n$ denotes a measurement applied on the tensor product space, and $\cM_n(\rho^{\otimes n}), \cM_n(\sigma^{\n})$ are the output distributions of this measurement applied on $\rho^{\otimes n}, \sigma^{\n}$, respectively.
\end{prop}
 
If the order of the infimum over $\cM_n$ and the suprema over $\alpha, \sigma^{\n}$ could be exchanged in \eqref{eq:sc upper}, then 
the inequalities in \eqref{eq:sc lower}--\eqref{eq:sc upper} would become equalities, giving a full characterization of the strong converse exponent. However, the function on the right-hand side (or any of its reparametrizations) is not known to be jointly concave in $\big(\alpha, \sigma^{\n}\big)$, and therefore it is not clear whether the obvious idea of using a minimax argument works.

Not being able to characterize the strong converse exponent $\rmsc_r(\rho\| \{\cF_n\}_{n\in \N})$, we cannot derive a necessary and sufficient condition for the corresponding additivity problem. We can, however, prove a sufficient such condition.
To formulate it, let us introduce 
\begin{align}
 r_{\infty}(\rho\|\cF_1)
 &:=\sup_{\alpha\in[1,+\infty)}\left\{\alpha\widetilde D_\infty(\rho\|\cF_1)-\log\sandq_\alpha(\rho\|\cF_1)\right\}.
\end{align}
Note that by substituting $\alpha=1$, we get $r_{\infty}(\rho\|\cF_1)\ge \widetilde D_\infty(\rho\|\cF_1)$. According to~\cite[Corollary~V.12]{Mosonyi+22errorexponent} 
 and~\cite[Corollary~V.20]{Mosonyi+22errorexponent}, for every $D(\rho\|\cF_1)<r<r_{\infty}(\rho\|\cF_1)$, there exist $\alpha_0\in(1,+\infty)$ and $\sigma_0\in\cF_1$
such that 
\begin{align}\label{eq:ar def}
H_r^*(\rho\|\cF_1)=\frac{\alpha_0-1}{\alpha_0}\Big( r-\widetilde D_{\alpha_0}(\rho\| \sigma_0)\Big).
\end{align}

\begin{theorem}\label{thm:additivity-strong-converse}
Consider the quantum hypothesis testing problem~\eqref{eq:problem}, and assume that  Assumption~\ref{assumption:polar} and Assumption~\ref{assumption:support} are satisfied, and that
$D(\rho\|\cF_1) <   \widetilde D_\infty(\rho\|\cF_1)$. Let
\begin{align}\label{eq:sc rate condition}
D(\rho\|\cF_1) < r<  r_{\infty}(\rho\|\cF_1)
\end{align}
and
\begin{align}\label{eq:sc optimizers}
(\alpha_0, \sigma_0)\in \argmax_{(\alpha, \sigma)\in (1, +\infty)\times \cF_1} \frac{\alpha-1}{\alpha} \big(r - \widetilde D_\alpha(\rho\| \sigma)\big)
\end{align}
as in \eqref{eq:ar def},
so that $\sigma_0=\argmax_{\sigma\in\cF_1}H_r^*(\rho\|\sigma)=\argmin_{\sigma\in\cF_1}\widetilde D_{\alpha_0}(\rho\|\sigma)$.
Then, 
\begin{align}
&\rmsc_r(\rho\| \{\cF_n\}_{n\in \N}) = H_r^*(\rho\|\cF_1)=\frac{\alpha_0-1}{\alpha_0}\big( r - \widetilde D_{\alpha_0}(\rho\| \sigma_0)\big)
\label{eq:sc-exponent-thm}
\end{align}
implies
\begin{align}
\sup_{t \in \mathbb{R}, \tau \in \mathcal{F}_1}\textup{Tr}\Bigg[\tau 
\sigma_0^{-\frac{1+it}{2}}\big(\sigma_0^{\frac{1-\alpha_0}{2\alpha_0}} \rho\sigma_0^{\frac{1-\alpha_0}{2\alpha_0}} \big)^{\alpha_0}
\sigma_0^{-\frac{1-it}{2}}\Bigg] = \widetilde Q_{\alpha_0}(\rho\|\sigma_0).
\label{eq:sc-exponent-thm2}
\end{align}
Moreover, the converse implication is also true if at least one of the following conditions is satisfied:
\begin{enumerate}
\item\label{sc assumption1}
$1<\alpha_0\le 2$ and $r<\sand_{\infty}(\rho\|\cF_1)$.
\item\label{sc assumption2}
$1<\alpha_0\le 2$ and $\frac{\dd}{\dd\alpha} \log \widetilde Q_\alpha(\rho\| \sigma_0)\Big|_{\alpha=\alpha_0}<\sand_{\infty}(\rho\|\sigma_0)$.
\end{enumerate}
Furthermore, if $\rho$ and $\sigma_0$ commute, then $\frac{\dd}{\dd\alpha} \log \widetilde Q_\alpha(\rho\| \sigma_0)\Big|_{\alpha=\alpha_0}<\sand_{\infty}(\rho\|\sigma_0)$, and~\eqref{eq:sc-exponent-thm},~\eqref{eq:sc-exponent-thm2} hold without any further assumption. 
\end{theorem}

\begin{proof}
According to Theorem \ref{thm:additivity-derivative} and Remark \ref{rem:az additivity simplified}, the failure of the inequality in \eqref{eq:sc-exponent-thm2} is 
equivalent to the existence of an $n\in\bN$ such that $\sand_{\alpha_0}(\rho^{\otimes n}\|\cF_n)<n\sand(\rho\|\cF_1)$. By Proposition \ref{prop:sc bounds}, this implies
\begin{align}
\rmsc_r(\rho\| \{\cF_n\}_{n\in \N}) 
\geq 
\frac{\alpha_0-1}{\alpha_0}\left( r - \frac{1}{n}\widetilde D_{\alpha_0}(\rho^{\otimes n}\| \cF_n)\right)
>
\frac{\alpha_0-1}{\alpha_0}\left( r - \widetilde D_{\alpha_0}(\rho\| \sigma_0)\right)=H_r^*(\rho\|\sigma_0)\,.
\end{align}
This shows that~\eqref{eq:sc-exponent-thm} implies~\eqref{eq:sc-exponent-thm2}.

Assume now that the inequality in \eqref{eq:sc-exponent-thm2} holds; note that this is always satisfied if 
$\rho$ commutes with $\sigma_0$, according to Corollary~\ref{commuting}.
Note also that by
Proposition \ref{prop:sc bounds},
\begin{equation}
    \rmsc_r(\rho\| \{\cF_n\}_{n\in \N}) \geq \rmsc_r(\rho\| \sigma_0) =\frac{\alpha_0-1}{\alpha_0}\big( r - \widetilde D_{\alpha_0}(\rho\| \sigma_0)\big)=H_r^*(\rho\|\sigma_0)\,,
\end{equation}
and hence our goal is to prove the converse inequality.

Taking the derivative of $\frac{\alpha-1}{\alpha} \big(r - \widetilde D_\alpha(\rho\| \sigma_0)\big)$ with respect to $\alpha$, the optimality of $\alpha_0$ in \eqref{eq:sc optimizers} implies  
\begin{align}
r&= \gamma\alpha_0 -\log \widetilde Q_{\alpha_0}(\rho\| \sigma_0),\label{eq:scr-gamma-1}\\
H_r^*(\rho\|\sigma_0)&=
\gamma(\alpha_0-1)-\log \widetilde Q_{\alpha_0}(\rho\| \sigma_0)
=
\max_{\alpha\in(1,+\infty)}\left\{\gamma(\alpha-1)-\log \widetilde Q_{\alpha}(\rho\| \sigma_0)\right\},
\label{eq:scr-gamma-1-1}
\end{align}
where
\begin{align}\label{eq:scr-gamma-2}
\gamma := \frac{\dd}{\dd\alpha} \log \widetilde Q_\alpha(\rho\| \sigma_0)\Big|_{\alpha=\alpha_0},
\end{align}
and we used the convexity of the function $(1,+\infty)\ni\alpha\mapsto \log \widetilde Q_{\alpha}(\rho\| \sigma_0)$; see, e.g., \cite[Corollary 3.11]{mosonyi2015quantum}.

Assume first that $1<\alpha_0\le 2$, or that $\rho, \sigma_0$ commute.
Let $\mathcal P_{\sigma_0^{\otimes n}}(\cdot )$ be the pinching map with respect to $\sigma_0^{\otimes n}$. Then, there exists a 
probability measure $\mu_n$ on $\bR$ such that~\cite[Lemma 2.1]{Sutter+2017multivariate} 
\begin{align}
\widehat \rho_n := \mathcal P_{\sigma_0^{\otimes n}}(\rho^{\otimes n} ) = \int_{\bR}   \big(\sigma_0^{\otimes n}\big)^{it/2} \rho^{\otimes n}  \big(\sigma_0^{\otimes n}\big)^{-it/2} \dd \mu_n(t). 
\end{align}
Let 
\begin{align}
T_n = \Pi_{\widehat \rho_n> e^{\gamma n} \sigma_0^{\otimes n}} =\Pi_{\widehat \rho_n^{\alpha_0}> e^{\alpha_0\gamma n} \big(\sigma_0^{\otimes n}\big)^{\alpha_0}} 
\end{align}
be the orthogonal projection onto the span of the eigenvectors of $\widehat \rho_n- e^{\gamma n} \sigma_0^{\otimes n}$ with positive eigenvalues.
Using the fact that $\widehat \rho_n$ and $\sigma_0^{\otimes n}$ commute, we find that $T_n\leq e^{-\alpha_0\gamma n} \widehat \rho_n^{\alpha_0} \big(\sigma_0^{\otimes n}\big)^{-\alpha_0} $. Hence, for any $\tau^{\n} \in \cF_n$, the type II error is bounded as
\begin{align}
&e^{\alpha_0\gamma n} \Trm\big[\tau^{\n}  T_n \big]  \\
&\qquad\leq  
\Trm\Big[   \tau^{\n} \widehat \rho_n^{\alpha_0} \big(\sigma_0^{\otimes n}\big)^{-\alpha_0}  \Big] \\
&\qquad =
\Trm\Big[   \tau^{\n}  \big(\sigma_0^{\otimes n}\big)^{-\frac 12} \Big(  \big(\sigma_0^{\otimes n}\big)^{\frac{1-\alpha_0}{2\alpha_0}}   \widehat \rho_n \big(\sigma_0^{\otimes n}\big)^{\frac{1-\alpha_0}{2\alpha_0}}  \Big)^{\alpha_0}  \big(\sigma_0^{\otimes n}\big)^{-\frac 12} \Big]\\
&\qquad=  
\Trm\bigg[   \tau^{\n}  \big(\sigma_0^{\otimes n}\big)^{-\frac 12} \bigg(\int_{\bR}  \big(\sigma_0^{\otimes n}\big)^{\frac{1-\alpha_0}{2\alpha_0} + \frac{it}{2}}    \rho^{\otimes n} \big(\sigma_0^{\otimes n}\big)^{\frac{1-\alpha_0}{2\alpha_0}-\frac{it}{2}}  \dd \mu_n(t) \bigg)^{\alpha_0}  \big(\sigma_0^{\otimes n}\big)^{-\frac 12} \bigg]\\
&\qquad\leq  
\int_{\bR} \Trm\Big[   \tau^{\n}  \big(\sigma_0^{\otimes n}\big)^{-\frac 12} \Big(  \big(\sigma_0^{\otimes n}\big)^{\frac{1-\alpha_0}{2\alpha_0} + \frac{it}{2}}    \rho^{\otimes n} \big(\sigma_0^{\otimes n}\big)^{\frac{1-\alpha_0}{2\alpha_0}-\frac{it}{2}}   \Big)^{\alpha_0}  \big(\sigma_0^{\otimes n}\big)^{-\frac 12} \Big] \dd \mu_n(t)\label{eq:op-convexity-x-alpha0}\\
&\qquad\le
\sup_{t\in\bR} \Trm\Big[   \tau^{\n}  \big(\sigma_0^{\otimes n}\big)^{-\frac 12+ \frac{it}{2}} \Big(  \big(\sigma_0^{\otimes n}\big)^{\frac{1-\alpha_0}{2\alpha_0} }    \rho^{\otimes n} \big(\sigma_0^{\otimes n}\big)^{\frac{1-\alpha_0}{2\alpha_0}}   \Big)^{\alpha_0}  \big(\sigma_0^{\otimes n}\big)^{-\frac 12-\frac{it}{2}} \Big]\\
&\qquad=
\sup_{t \in \mathbb{R}}\textup{Tr}\bigg[\tau^{(n)} 
\left(\sigma_0^{-\frac{1+it}{2}}\Big(\sigma_0^{\frac{1-\alpha_0}{2\alpha_0}} \rho\sigma_0^{\frac{1-\alpha_0}{2\alpha_0}} \Big)^{\alpha_0}
\sigma_0^{-\frac{1-it}{2}}\right)^{\otimes n}\bigg].
\end{align}
Here,~\eqref{eq:op-convexity-x-alpha0} follows from the convexity of $x\mapsto x^{\alpha_0}$ 
when $\rho$ and $\sigma_0$ commute, 
and from the operator convexity of the same function when $1<\alpha_0\le 2$,
and the rest of the steps are straightforward. Taking now the supremum over $\tau^{(n)}\in\cF_n$
and using condition~\ref{A.4-polar} yield
\begin{align}
\sup_{\tau^{(n)}\in\cF_n}\Trm\big[\tau^{\n}  T_n \big]  
&\le 
e^{-\alpha_0\gamma n} \sup_{t \in \mathbb{R},\,\tau\in\cF_1}\textup{Tr}\bigg[\tau 
\sigma_0^{-\frac{1+it}{2}}\Big(\sigma_0^{\frac{1-\alpha_0}{2\alpha_0}} \rho\sigma_0^{\frac{1-\alpha_0}{2\alpha_0}} \Big)^{\alpha_0}
\sigma_0^{-\frac{1-it}{2}}\bigg]^n\\
&\le
e^{-\alpha_0\gamma n}Q_{\alpha_0}(\rho\|\sigma_0)^n=e^{-nr},
\end{align}
where the second inequality follows from \eqref{eq:sc-exponent-thm2} and the last equality from \eqref{eq:scr-gamma-1}.

The proof of the equality in \eqref{eq:sc-exponent-thm} will be complete if we can show that 
\begin{align}
\lim_{n\to \infty}  -\frac 1n \log \tr\big[ \rho^{\otimes n} T_n  \big] &=
H_r^*(\rho\|\sigma_0).
\end{align}
This follows immediately from \cite[Theorem IV.5]{MosonyiOgawa15two} if $D(\rho\|\sigma_0)<\gamma<\sand_{\infty}(\rho\|\sigma_0)$, and hence our aim is to show these inequalities. 

First, assume that $\gamma\le D(\rho\|\sigma_0)$. Then 
$\gamma\le \sand_{\alpha_0}(\rho\|\sigma_0)$ according to Lemma \ref{lemma:psi properties}, whence, by  
\eqref{eq:scr-gamma-1}, we get 
\begin{align}\label{eq:gamma-equiv}
\gamma &= \frac{1}{\alpha_0} r+ \frac{\alpha_0-1}{\alpha_0} \widetilde D_{\alpha_0}(\rho\| \sigma_0)  
\ge
\frac{1}{\alpha_0}r+\frac{\alpha_0-1}{\alpha_0}\gamma,
\end{align}
or equivalently $\gamma\ge r$. This yields 
$r\le\gamma\le D(\rho\|\sigma_0)\le \sand_{\alpha_0}(\rho\|\sigma_0)$, whence 
$H_r^*(\rho\|\cF_1)=\frac{\alpha_0-1}{\alpha_0}\left(r-\sand_{\alpha_0}(\rho\|\sigma_0)\right)\le 0$.
This, however, leads to a contradiction as the assumption $r>D(\rho\|\cF_1)$ implies that $r>\sand_{\alpha}(\rho\|\cF_1)$ for every 
$\alpha$ close enough to $1$, and for any such $\alpha$,
$H_r^*(\rho\|\cF_1)\ge\frac{\alpha-1}{\alpha}\left(r-\sand_{\alpha}(\rho\|\sigma)\right)> 0$.
The inequality $\gamma<\sand_{+\infty}(\rho\|\sigma_0)$ holds trivially under condition \eqref{sc assumption2}, and from 
condition \eqref{sc assumption1} it follows due to 
\eqref{eq:scr-gamma-1} as
\begin{align}
\gamma &= \frac{1}{\alpha_0} r+ \left(1-\frac{1}{\alpha_0}\right) \widetilde D_{\alpha_0}(\rho\| \sigma_0) 
<\sand_{\infty}(\rho\|\sigma_0),
\end{align}
where we used that 
$\sand_{\alpha_0}(\rho\|\sigma_0)\le \sand_{\infty}(\rho\|\sigma_0)$, according to Lemma~\ref{lemma:psi properties} below.

Finally, assume that $\rho$ commutes with $\sigma_0$. By \cite[Lemma II.10]{Mosonyi+22errorexponent}, there are two possibilities.
One is that the second derivative of $\alpha\mapsto\psi(\alpha):=\log\sand_{\alpha}(\rho\|\sigma_0)$ is strictly positive at every 
$\alpha\in(0,+\infty)$, and therefore $D(\rho\|\sigma_0)=\psi'(1)<\psi'(\alpha_0)=\gamma<\psi'(+\infty)=\sand_{\infty}(\rho\|\sigma_0)$.
Thus, the above argument can be used to derive the strong converse exponent. 
Alternatively, if $\psi''(\alpha)=0$ at some $\alpha\in(0,+\infty)$, then 
$\rho=\kappa\sigma_0\rho^0$ for some $\kappa>0$, whence
$\log\sand_{\alpha}(\rho\|\sigma_0)=(\alpha-1)\log\kappa$ for every $\alpha\in[1,+\infty]$, and 
therefore $\gamma=\psi'(\alpha_0)=\log\kappa=D(\rho\|\sigma_0)$.
This, however, contradicts the assumption that $D(\rho\|\cF_1)<r$, as shown above.

\end{proof}

The following lemma is well known in the literature; see, e.g., \cite{muller2013quantum,MosonyiOgawa15two}.

\begin{lemma}\label{lemma:psi properties}
Let $\rho,\sigma$ be such that $\supp(\rho)\subseteq \supp(\sigma)$, and let
\begin{align}\label{eq:psi def}
\psi(\alpha):=\log\sandq_{\alpha}(\rho\|\sigma),\ds\ds\ds\alpha\in(0,+\infty).
\end{align}
Then $\psi(\cdot)$ is a continuously differentiable convex function on $(0,+\infty)$ with $\psi(1)=0$.
In particular, $\alpha\mapsto\sand_{\alpha}(\rho\|\sigma)=\frac{\psi(\alpha)}{\alpha-1}$ and 
$\alpha\mapsto\psi'(\alpha)$ are monotone increasing on $(0,+\infty]$, where 
\begin{align}
\frac{\psi(\alpha)}{\alpha-1}\Bigg\vert_{\alpha=1}:=\lim_{\alpha\to 1}\frac{\psi(\alpha)}{\alpha-1}=D(\rho\|\sigma)
=\lim_{\alpha\to 1}\psi'(\alpha)=\psi'(1),
\end{align}
and 
\begin{align}\label{eq:infty limits}
\frac{\psi(\alpha)}{\alpha-1}\Bigg\vert_{\alpha=+\infty}:=\lim_{\alpha\to +\infty}\frac{\psi(\alpha)}{\alpha-1}=\sand_{\infty}(\rho\|\sigma)
=\lim_{\alpha\to +\infty}\psi'(\alpha)=:\psi'(+\infty).
\end{align}
\end{lemma}

\section*{Acknowledgements}
SB is supported by the Iran National Science Foundation (INSF) under project No.\ 4031370. 
MT and RR are supported by the Ministry of Education, Singapore, through grant T2EP20124-0005 and by the National Research Foundation, Singapore through the National Quantum Office, hosted in A*STAR, under its Centre for Quantum Technologies Funding Initiative (S24Q2d0009). The authors are indebted to Mil{\'a}n Mosonyi for his comments on an earlier version of the paper which, in particular, strengthened Theorem~\ref{thm:additivity-strong-converse} and the results in Appendix~\ref{app:strong-converse}.

\bibliographystyle{ultimate}
\bibliography{my}


\appendix

\section{Support condition in computing derivatives}\label{app:derivative}

This appendix is devoted to proving a lemma about the derivative of certain matrix functions when some desired support condition is not met. This lemma indeed says that in computing certain derivatives we may restrict to the support of the underlying operator. To prove this lemma, we borrow ideas from~\cite[Section VII]{friedland2011explicit}. 

\begin{lemma}\label{lem:derivative-support}
Let $\phi:\mathbb R_{>0}\to \mathbb R$ be an analytic function such that 
\begin{align}\label{eq:lemma-cond-phi}
t^2\phi(\lambda(t)) = o(t), \qquad \text{ as } t\searrow 0,
\end{align}
for any function $\lambda(t)>0$ satisfying $\lambda(t)=O(t)$. Let $\sigma$ be a non-zero positive semidefinite operator and $\xi$ be a hermitian matrix such that $\sigma+t\xi$ is positive semidefinite for small values of $t>0$. Let $P$ be the orthogonal projection on the support of $\sigma$. Then, we have\footnote{As in the main text, $\phi(X)$ for a positive semidefinite operator $X$, is defined with respect to the support of $X$.}
\begin{align}\label{lem:phi-derivative-support}
\frac{\dd}{\dd t}  P\phi(\sigma + t\xi)P \bigg|_{t=0}=\frac{\dd}{\dd t}  \phi(\sigma + tP\xi P)\bigg|_{t=0}.
\end{align}
\end{lemma}

We note that condition \eqref{eq:lemma-cond-phi} cannot be completely omitted in the lemma
as~\eqref{lem:phi-derivative-support} does not hold, e.g., for $\phi(t)=1/t$ (see Example~\ref{example:derivative} below). We also note that this condition holds for functions $\log t$ and $t^\theta$ for $0<|\theta|<1$. 

\begin{proof}
Assume with no loss of generality that
\begin{align} \sigma = \begin{pmatrix}
\sigma_0 & 0 \\ 0 & 0
\end{pmatrix},\qquad \xi = \begin{pmatrix}
\xi_0 & \nu \\ \nu^\dagger  & \xi_1
\end{pmatrix},
\end{align}
with $\sigma_0$ being full-rank matrix of size $r\times r$. 
In this case, we need to show that the top-left block of size $r\times r$ of the derivative of $\phi(\sigma + t\xi)$ at $t=0$ is equal to the derivative of $\phi(\sigma_0 + t\xi_0)$ at $t=0$.

By Rellich's theorem~\cite{Rellich} (see also~\cite{Wimmer_perturbation}), there are analytic functions $\lambda_1(t), \dots, \lambda_d(t)$ and analytic unitary matrix $U(t)$ such that 
\begin{align}
\sigma + t\xi = U(t) \diag\big(\lambda_1(t), \dots, \lambda_d(t)\big)U(t)^\dagger,
\end{align}
where we take $d$ to be the dimension of the underlying space.  
We note that by assumptions $\lambda_i(t)\geq 0$ for small values of $t\geq 0$. Moreover, since $\sigma$ has $r$ positive eigenvalues, we can assume that $\lambda_1(0), \dots, \lambda_r(0)>0$ and $\lambda_{r+1}=\cdots=\lambda_d(0)=0$. This assumption implies that $U(0)$ which diagonalizes $\sigma$, is of the form $\diag(V_0, V_1)$ where $V_0$ and $V_1$ are unitaries of size $r\times r$ and $(d-r)\times (d-r)$, respectively. Now suppose that we replace $U(t)$ with $\diag(V^\dagger_0, V^\dagger_1)U(t)$. This would change $\sigma_0$ to $V_0^\dagger \sigma_0V_0$, which is diagonal since $U(0)^\dagger \sigma U(0)$ is diagonal. As a result, we can assume with no loss of generality that
\begin{align}
U(0)=I, \quad \text{ and } \quad \sigma_0 = \diag(\lambda_1, \dots, \lambda_r),\qquad \lambda_1=\lambda_1(0), \dots, \lambda_r=\lambda_r(0)>0\,.
\end{align}

By the analyticity of $U(t)$, we can expand $U(t)$ around $t=0$ as 
\begin{equation}
    U(t)= I + t X + O(t^2).
\end{equation}
We note that $U(t)U(t)^\dagger=I$ implies 
\begin{align}\label{eq:X-anti-herm}
X^\dagger = - X.
\end{align}
That is, $X$ is anti-hermitian. In particular, its diagonal entries vanish. 
We can also expand the eigenvalues as
\begin{equation}
    \lambda_i(t) = \lambda_i + t\mu_i +O(t^2),
\end{equation}
where $\lambda_{i}=0$ for $i\geq r+1$. We then have
\begin{align}
\sigma + t\xi & = U(t) \diag(\lambda_1(t), \dots, \lambda_d(t))U(t)^\dagger\\
& = \sigma + t\bigg(  X\diag(\lambda_1, \dots, \lambda_d) + \diag(\lambda_1, \dots, \lambda_d)X^\dagger + \diag(\mu_1, \dots, \mu_d)   \bigg) + O(t^2). 
\end{align}
This implies
\begin{align}
\xi & =   X\diag(\lambda_1, \dots, \lambda_d) + \diag(\lambda_1, \dots, \lambda_d)X^\dagger + \diag(\mu_1, \dots, \mu_d).
\end{align}
Comparing the diagonal entries on both sides, we find that
\begin{equation}
    \bra i \xi \ket i = \mu_i.
\end{equation}
Also, letting $\xi'= \xi - \diag(\mu_1, \dots, \mu_d)$ we have 
\begin{align}
\xi' & =   X\diag(\lambda_1, \dots, \lambda_d) + \diag(\lambda_1, \dots, \lambda_d)X^\dagger.
\end{align}
This equation together with~\eqref{eq:X-anti-herm} implies that if $\lambda_i\neq \lambda_j$, then 
\begin{align}\label{eq:X-entries}
\bra i X\ket j = \frac{1}{\lambda_j - \lambda_i} \bra i \xi' \ket j.
\end{align}

We now turn to the analysis of the top-left block of the derivative of $\phi(\sigma + t\xi)$. Denoting the top-left $r\times r$ block of a matrix $M$ by $[M]_{r\times r}$, we have
\begin{align}
\big[\phi(\sigma + t\xi)\big]_{r\times r} & =  \big[ U(t) \diag\big(\phi\circ \lambda_1(t), \dots, \phi\circ \lambda_d(t)\big) U(t)^\dagger\big]_{r\times r}\\
& = \big[ U(t) \diag\big(\phi\circ \lambda_1(t), \dots, \phi\circ \lambda_r(t) , 0 \dots, 0\big)U(t)^\dagger\big]_{r\times r} \\
&\quad\, + \big[ U(t) \diag\big(0, \dots, 0,\phi\circ \lambda_{r+1}(t), \dots, \phi\circ \lambda_d(t)\big)U(t)^\dagger\big]_{r\times r}.
\end{align}
We first show that the second term is of order $o(t)$ and does not affect the derivative at $t=0$. To this end, letting 
\begin{equation}
    U(t) = \begin{bmatrix}
W_0(t) & V(t) \\ V'(t) & W_1(t)
\end{bmatrix},
\end{equation}
we have 
\begin{align}
&\big[ U(t) \diag\big(0, \dots, 0,\phi\circ \lambda_{r+1}(t), \dots, \phi\circ \lambda_d(t)\big)U(t)^\dagger\big]_{r\times r} \\
& =   V(t) \diag\big(\phi\circ \lambda_{r+1}(t), \dots, \phi\circ \lambda_d(t)\big)   V(t)^\dagger .
\end{align}
On the other hand, since $U(0)=I$, we have $V(t) = tZ + O(t^2),$ where $Z$ comes from
\begin{align}
X= \begin{pmatrix}
X_0 & Z\\
-Z^\dagger & X_1
\end{pmatrix}.
\end{align}
Moreover, $\lambda_i=0$ for $r+1\leq i\leq d$, which implies $\lambda_i(t)=O(t)$. Thus, by assumption, for $r+1\leq i\leq d$ we have $t^2\phi\circ\lambda_i(t) = o(t)$.
Therefore, 
\begin{align}
&\big[U(t) \diag(0, \dots, 0,\phi\circ \lambda_{r+1}(t), \dots, \phi\circ \lambda_d(t))U(t)^\dagger\big]_{r\times r} \\
& =  t^2 \big(Z + O(t)\big) \diag\big(\phi\circ \lambda_{r+1}(t), \dots, \phi\circ \lambda_d(t)\big)  \big( Z^\dagger + O(t)\big)  \\
& = o(t).
\end{align}
We conclude that 
\begin{align}
\big[ \phi(\sigma + t\xi)\big]_{r\times r} &  = \big[ U(t) \diag\big(\phi\circ \lambda_1(t), \dots, \phi\circ \lambda_r(t) , 0 \dots, 0\big)U(t)^\dagger\big]_{r\times r} + o(t)\\
&  =  W_0(t) \diag\big(\phi\circ \lambda_1(t), \dots, \phi\circ \lambda_r(t) \big)W_0(t)^\dagger + o(t).
\end{align}
We have $\lambda_i(t) = \lambda_i + t\mu_i + O(t^2)$ with $\lambda_i>0$ for $1\leq i\leq r$. Then, by the analyticity of $\phi(\cdot)$ and $\lambda_i(t)$, we have $\phi\circ \lambda_i(t) = \phi(\lambda_i) + t\mu_i\phi'(\lambda_i) + O(t^2)$. We also have $W_0(t) = I + tX_0 + O(t^2)$. Putting these together, we arrive at
\begin{align}
&\big[ \phi(\sigma + t\xi)\big]_{r\times r} \\
&  = \phi(\sigma_0) + t  \,\diag\big(\mu_1\phi'(\lambda_1), \dots, \mu_rr\phi'(\lambda_r)\big)  \\
&\quad  + t\Big( X_0\, \diag\big(\phi( \lambda_1), \dots, \phi( \lambda_r)\big)+ \diag\big(\phi(\lambda_1), \dots, \phi(\lambda_r)\big)X_0^\dagger  \Big)   
 + o(t).
\end{align}
We note that $\mu_i=\bra i\xi\ket i = \bra i\xi_0\ket i$ for $1\leq i\leq r$ depends only on $\xi_0$.  Also, by~\eqref{eq:X-entries} if $\lambda_i\neq \lambda_j$, then $\bra i X_0\ket j$ depends only on $\xi_0$. On the other hand, if $\lambda_i=\lambda_j$, by $X_0^\dagger =-X_0$ we find that 
\begin{align} \big\langle  i\big|\Big( X_0 \diag\big(\phi\circ \lambda_1(t), \dots, \phi\circ \lambda_r(t)\big)+ \diag\big(\phi(\lambda_1), \dots, \phi(\lambda_r)\big)X_0^\dagger  \Big)\big| j\big\rangle=0.
\end{align}
Therefore, all the first order terms in $[ \phi(\sigma + t\xi)]_{r\times r}$ are independent of $\nu, \xi_1$.  Thus, in computing 
\begin{equation}
    \frac{\dd}{\dd t} \big[ \phi(\sigma + t\xi)\big]_{r\times r} \bigg|_{t=0},
\end{equation}
we can simply put $\nu=0$ and $\xi_1=0$. This gives the desired identity.
\end{proof}

The following example shows that the above lemma does not hold for $\phi(t)=1/t$.

\begin{example}\label{example:derivative}
    Let $\phi(t)=1/t$, $t\in(0,+\infty)$, which fails to satisfy \eqref{eq:lemma-cond-phi}. Let
    \begin{align}
        M_t = (1-t)\begin{pmatrix}
            A & 0 \\
            0 & 0
        \end{pmatrix} + t\begin{pmatrix}
            I & I \\
            I & I
        \end{pmatrix}
        =
        \underbrace{\begin{pmatrix}
            A & 0 \\
            0 & 0
        \end{pmatrix}}_{=\sigma} + t\underbrace{\begin{pmatrix}
            I-A & I \\
            I & I
        \end{pmatrix}}_{=\xi},\qquad P:=\sigma^0=\begin{pmatrix}I & 0 \\0 & 0\end{pmatrix},
    \end{align}
    with $A=\sigma_0$ being a positive definite $r\times r$ matrix and $I$ denoting the $r\times r$ identity matrix. Note that 
    $P\phi(\sigma+t\xi)P$ and $\phi(\sigma+tP\xi P)$ both have non-zero entries only in the top left block, which we denote by $[X]_{r\times r}$ for any $2r\times 2r$ matrix $X$. Thus, our goal is to show that  
    $[M_t^{-1}]_{r\times r}$ (identified with $P\phi(\sigma+t\xi)P$) 
    and the inverse of $A_t=(1-t)A+tI$ (identified with $\phi(\sigma+tP\xi P)$) have different derivatives 
    at $t=0$. In the special case $r=1$, this follows immediately from 
    \begin{align*}
    M_t^{-1}=\begin{pmatrix}\frac{1}{(1-t)A} & -\frac{1}{(1-t)A} \\-\frac{1}{(1-t)A} & \frac{1}{(1-t)A}+\frac{1}{t}\end{pmatrix},
    \qquad t>0.
    \end{align*}
    This already shows that \eqref{lem:phi-derivative-support} does not hold for $\phi(t)=1/t$ in general. 
    The case $r>1$ can also be explicitly computed, using the formula for the inverse of $2\times 2$ block matrices (see, e.g.,~\cite[Proposition~2.8.7]{Bernstein2009}), which gives
    \begin{align}
        \big[M_t^{-1}\big]_{r\times r} = A_t^{-1} +t A_t^{-2} (I-tA_t^{-1})^{-1}.  
    \end{align}
    We note that $\frac{\dd}{\dd t}A_t = I-A$ 
    and $A_t$ commutes with $I-A$. Therefore, 
    \begin{align}
        \frac{\dd}{\dd t} A_t^{-1}\Big|_{t=0} = -A^{-2}(I-A),
    \end{align}
    and 
    \begin{align}
        \frac{\dd}{\dd t}\big[M_t^{-1}\big]_{r\times r} \Big|_{t=0} & = -A^{-2}(I-A) + A^{-2},
    \end{align}
    which are clearly not equal to each other.
\end{example}

\section{Necessary and sufficient conditions for the optimizer}
\label{app:necessary and sufficient}

In this appendix, we review the necessary and sufficient conditions for the optimizer(s) of the $\alpha$-$z$-R\'enyi relative entropies derived in~\cite{rubboli2024new}. 

Let $\rho,\tau$ be two states. We define
\begin{align}
&\Xi'_{\alpha,z}(\rho,\sigma) :=  \begin{dcases}
\chi_{\alpha,1-\alpha}(\rho,\sigma) &  \textup{if} \; z=1-\alpha \\
\sigma^{-1}  \chi_{\alpha,\alpha-1}(\rho, \sigma) \sigma^{-1}  &  \textup{if} \; z=\alpha-1 \\
\sinc\left(\pi\frac{{1-\alpha}}{z}\right) \int_0^\infty (\sigma + t)^{-1}\chi_{\alpha,z}(\rho,\sigma) (\sigma + t)^{-1} t^\frac{1-\alpha}{z} \d t & \textup{if} \; |(1-\alpha)/z| \neq 1
\end{dcases} \,,
\end{align}
where
\begin{equation}
    \chi_{\alpha,z}(\rho,\sigma):=  \rho^\frac{\alpha}{2z}( \rho^{\frac{\alpha}{2z}}\sigma^{\frac{1-\alpha}{z}} \rho^{\frac{\alpha}{2z}})^{z-1} \rho^\frac{\alpha}{2z} \,.
\end{equation}
Here, all inverses are defined as generalized inverses. Moreover, we define
\begin{align}
S_{\alpha,z}(\rho) =  \begin{cases} 
\left\{ \sigma \; \text{is a state} : \supp(\rho) = \supp\!\left(\rho^0 \sigma \rho^0\right) \right\}  &  \textup{if} \; (1-\alpha)/z = 1,\\
\left\{ \sigma \; \text{is a state} : \supp(\rho) \subseteq \supp(\sigma)  \right\} &  \textup{otherwise}. 
\end{cases}
\end{align} 

The observation about the supports in Appendix~\ref{app:derivative} together with the results in~\cite[Theorem 4]{rubboli2024new} yield the following theorem.\footnote{Note that in~\cite[Theorem 4]{rubboli2024new}, the condition is stated as an inequality for all $\sigma \in \cF$. However, this is equivalent to verifying the inequality at the supremum over all $\sigma \in \cF$. Moreover, observe that the inequality is saturated at $\tau = \sigma_0$. Finally, we observe that although the set $\cF$ is assumed to consist of states, the same proof remains valid when $\cF$ is taken to be a set of positive operators.}
\begin{theorem}
\label{main Theorem}
Let $(\alpha,z) \in \mathcal{D}$, $z \neq \alpha-1$, $\rho$ be a quantum state and $\cF$ be a convex and compact subset of positive operators that includes at least one state whose support contains the support of $\rho$. Then, $\sigma_0 \in \argmin_{\sigma \in \cF} D_{\alpha,z}(\rho \| \sigma)$ if and only if $\sigma_0 \in S_{\alpha,z}(\rho)$ and 
\begin{equation}
    \sup_{\tau \in \cF}\textup{Tr}[\tau\, \Xi'_{\alpha,z}(\rho,\sigma_0)] =  Q_{\alpha,z}(\rho \|\sigma_0) \,.
\end{equation}
Moreover, if $\rho$ is full-rank, the same holds for $z=\alpha-1$.
\end{theorem}

In the case the state $\rho$ commutes with the optimizer, the above theorem considerably simplifies as follows. 

\begin{corollary}
\label{commuting}
Let $(\alpha,z) \in \mathcal{D}$, $z \neq \alpha-1$,  $\rho$ be a quantum state, and $\cF$ be a convex and compact subset of positive operators that includes at least one state whose support contains the support of $\rho$. Then, a state $\sigma_0$ satisfying $[\rho,\sigma_0]= 0$ satisfies $\sigma_0 \in \argmin_{\sigma \in \mathcal{F}} D_{\alpha,z}(\rho \| \sigma)$ if and only if $\supp(\rho)\subseteq \supp(\sigma_0)$ and
\begin{equation}
    \sup_{\tau\in \cF}\textup{Tr}[\tau\rho^\alpha \sigma_0^{-\alpha}] =  \textup{Tr}[\rho^\alpha \sigma_0^{1-\alpha}] \,.
\end{equation}
Moreover, if $\rho$ is full-rank, the same holds for $z=\alpha-1$.
\end{corollary}


\section{Continuity of the minimized divergences}\label{app:continuity}

In the following, we establish the continuity of the function $\rho \mapsto \inf_{\sigma \in \cF} D_{\alpha,z}(\rho\|\sigma)$ whenever the set $\cF$ contains a full-rank state. This result is particularly useful because it circumvents technical difficulties associated with states lacking full support and enables the extension of results valid for full-rank states to those without full support.
Our proof is based on the ideas in the proof of~\cite[Theorem 1]{lami2024continuity} and~\cite[Theorem 5]{lami2023attainability}. 

\begin{prop}\label{Continuity}
Let $\cF$ be a closed convex set of quantum states. Then, for any $\alpha,z\in(0,+\infty)$, 
$f_{\alpha,z}(\rho):=\inf_{\sigma\in\cF}D_{\alpha,z}(\rho\|\sigma)$ is lower semi-continuous on the set of quantum states. Moreover, it is upper semi-continuous if either of the following holds: 
\begin{enumerate}[label=(\roman*)] 
\item $\alpha\in(0,1)$ 
\item $\alpha>1$, $\max\{\alpha/2,\alpha-1\}\leq z\leq \alpha$, and $\cF$ contains a full-rank state.
\end{enumerate}
\end{prop}
\begin{proof}
For any $\ep>0$, the function $g_{\alpha,z,\ep}(\rho,\sigma):= D_{\alpha,z}(\rho\|\sigma+\ep I)$ is continuous, whence 
$D_{\alpha,z}(\rho\|\sigma)=\lim_{\ep\searrow 0}D_{\alpha,z}(\rho\|\sigma+\ep I)=\sup_{\ep> 0}D_{\alpha,z}(\rho\|\sigma+\ep I)$
is lower semi-continuous. Therefore, $f$ is also lower semi-continuous 
(see, e.g., \cite[Lemma 2.3]{mosonyi2022geometric}). 

We note that $g_{\alpha,z,0}(\rho, \sigma)$ is
continuous when $\alpha\in(0,1)$, and therefore $f_{\alpha,z}$ is also continuous (see again \cite[Lemma 2.3]{mosonyi2022geometric}).

Hence, we only need to prove upper semi-continuity of $f$ when condition (ii) holds. In this case, we follow
the approach in the proof of~\cite[Theorem 1]{lami2024continuity}. Fix a full-rank state $\eta\in \cF$ and for any $\ep \in (0,1)$ let
\begin{equation}\label{eq:Y-eps-convexity}
Y_{\ep}(\rho):=\inf_{\sigma \in \cF} D_{\alpha,z}\big(\rho\big\|(1-\ep)\sigma+\ep\eta\big)
\le
(1-\ep) \inf_{\sigma \in\cF} D_{\alpha,z}\big(\rho\big\|\sigma\big)+ \ep D_{\alpha,z}\big(\rho\big\|\eta\big),
\end{equation}
where the inequality follows from the convexity of $D_{\alpha,z}(\cdot \| \cdot)$ in its second argument in the given range of 
$(\alpha,z)$. We note that $f_{\alpha,z}(\rho)\le Y_\ep(\rho)$ since by the convexity of $\cF$ states of the form $(1-\ep)\sigma + \ep \eta$ belong to $\cF$. Thus,
\begin{align}
 f_{\alpha,z}(\rho)\le \inf_{\ep\in(0,1)} Y_{\ep}(\rho)\le \liminf_{\ep\searrow 0} Y_{\ep}(\rho)\le f_{\alpha,z}(\rho),
    \end{align}
where the latter inequality follows from~\eqref{eq:Y-eps-convexity}.     
Hence, all the above inequalities hold with equality and
\begin{equation}
f_{\alpha,z}(\rho)=\inf_{\ep \in (0,1)}\inf_{\sigma\in\cF}\frac{1}{\alpha-1}\log Q_{\alpha,z}\big(\rho\big\|(1-\ep)\sigma+\ep\eta\big).
\end{equation}
Since $\rho \mapsto D_{\alpha,z}\big(\rho\big\|(1-\ep)\sigma+\ep\eta\big)$ is continuous for any fixed $\ep\in(0,1)$ and $\sigma\in\cF$, we find that $\rho\mapsto f_{\alpha,z}(\rho)$ is upper semi-continuous. 
\end{proof}


\section{Lower bounds for Stein exponent}\label{app:lower-Stein-exp}
In this appendix, we present a single-letter lower bound for the Stein exponent. This lower bound is in terms of additive quantities similar to those explored in Theorem~\ref{Additivity in the line}. Interestingly, our lower bound reduces to the KL divergence and is tight in the classical case. Such single-letter bounds may be obtained for the other error exponents by similar ideas.

Following the notation in Section~\ref{sec:additivity}, the R\'enyi divergence $D_{\alpha,1-\alpha}(\rho\|\sigma)$ is called the reversed sandwiched R\'enyi relative entropy. It does not satisfy the data-processing inequality for $\alpha \in [1/2,1)$, yet the condition for the minimum in Theorem~\ref{thm:derivative-condition} and the additivity result in Theorem~\ref{Additivity in the line} can be extended to this range since the corresponding optimization problem is convex. The point is that by the operator concavity of $x\mapsto x^{1-\alpha}$ for $\alpha\in(0,1)$, the function
\begin{align}
\sigma\mapsto Q_{\alpha,1-\alpha}(\rho\|\sigma)=\Trm\Big[\Big(\rho^\frac{\alpha}{2(1-\alpha)}\sigma\rho^\frac{\alpha}{2(1-\alpha)}\Big)^{1-\alpha}\Big]
\end{align} 
is concave.

The limit $\alpha \rightarrow 1^-$ of these relative entropies is given by~\cite[Section 3]{audenaert2015alpha}
\begin{align}\label{eq:limit-D-down}
\lim_{\alpha \rightarrow 1^-}D_{\alpha,1-\alpha}(\rho\|\sigma) = D^{\downarrow}\big(\rho\big\|\sigma):= &D(\rho\|\diag_{\rho}(\nu_1,\nu_2 /\nu_1,\dots,\nu_d /\nu_{d-1})\big).
\end{align}
Here, $d$ is the dimension of the underlying Hilbert space, and $\nu_k = \det (\sigma_{1:k,1:k})$ where $\sigma_{1:k,1:k}$ is the $k\times k$ matrix consisting of the first $k$ rows and $k$ columns of $\sigma$ in the eigenbasis of $\rho$. 

Another interesting property of the reversed sandwiched R\'enyi relative entropy is that 
 $D_{\alpha,1-\alpha}(\rho\|\sigma)$ $ \leq D_{\alpha,1}(\rho\|\sigma)$~\cite[Proposition 6]{lin2015investigating}. Then, the above equation and $\lim_{\alpha\to 1} D_{\alpha, 1}(\rho\|\sigma) = D(\rho\|\sigma)$ implies  $ D^{\downarrow}(\rho\| \sigma) \leq D(\rho\| \sigma)$.

\begin{lemma}\label{lem:D-downarrow-additive}
Suppose that Assumption~\ref{assumption:polar} and Assumption~\ref{assumption:support} are satisfied. Then, we have
 \begin{align}
 \lim_{\alpha \rightarrow 1^-}\min_{\sigma \in \mathcal{F}_1}D_{\alpha,1-\alpha}(\rho\|\sigma) = \min_{\sigma \in \mathcal{F}_1} D^{\downarrow}(\rho\|\sigma).    
 \end{align} 
Moreover, $\min_{\sigma^{\n}\in \cF_n} D^{\downarrow} \big(\rho^{\otimes n}\big\| \sigma^{\n}\big) = n \min_{\sigma\in \cF_1} D^\downarrow (\rho\| \sigma)$. 
\end{lemma}

\begin{proof}
First, we note that $\alpha \mapsto D_{\alpha,1-\alpha}(\rho \| \sigma)$ is monotonically increasing on $(0,1)$~\cite[Lemma 25]{rubboli2024mixed}. Thus, the limit $\alpha\to 1^-$ can be replaced with a supremum over $\alpha\in (0,1)$. On the other hand, the order of supremum and minimum can be exchanged due to the minimax theorem~\cite[Lemma II.1]{mosonyi2022some} and the lower semicontinuity of the $\alpha$-$z$-R\'enyi relative entropies for $z=1-\alpha$ (see, e.g.,~\cite[Lemma 18]{rubboli2024new}). Then the first equation follows from~\eqref{eq:limit-D-down}. 

For the second equation, we use the extension of Theorem~\ref{Additivity in the line} as discussed above to compute
\begin{align}
\min_{\sigma^{\n}\in \cF_n} D^{\downarrow} \big(\rho^{\otimes n}\big\| \sigma^{\n}\big) & = \lim_{\alpha \rightarrow 1^-}\min_{\sigma^{\n} \in \mathcal{F}_n}D_{\alpha,1-\alpha}\big(\rho^{\otimes n}\big\|\sigma^{\n}\big) \\
& = \lim_{\alpha \rightarrow 1^-}\min_{\sigma \in \mathcal{F}_1} nD_{\alpha,1-\alpha}(\rho\|\sigma)\\
& =n \min_{\sigma \in \cF_1} D^{\downarrow}(\rho\| \sigma). 
\end{align}
\end{proof}

Using Lemma~\ref{lem:D-downarrow-additive} and $ D^{\downarrow}(\rho\| \sigma) \leq D(\rho\| \sigma)$ established above we find that
\begin{align}
\rms(\rho\|\{\mathcal{F}_n:\, n\in \N\}) &= \lim_{n\rightarrow \infty}\frac{1}{n}\min_{\sigma^{\n}\in\mathcal{F}_n}D\big(\rho^{\otimes n}\big\|\sigma^{\n}\big) \\
&\geq  \lim_{n\rightarrow \infty}\frac{1}{n}\min_{\sigma^{\n}\in\mathcal{F}_n}D^{\downarrow}\big(\rho^{\otimes n}\big\|\sigma^{\n}\big) \\
&= \min_{\sigma \in \mathcal{F}_1} D^{\downarrow}(\rho\|\sigma).
\end{align}
This lower bound is tight in the classical case, and is tighter than the bound in terms of fidelity (corresponding to $\alpha=1/2$) due to the monotonicity of $\alpha \mapsto D_{\alpha,1-\alpha}(\rho \| \sigma)$ mentioned above.

\section{Alternative proof of the inequality of~\cite{Audenaert+07discriminating}}\label{app:proof-Audenaert-ineq}

In this appendix, we provide an alternative proof of the inequality derived by Audenaert et al.~in~\cite{Audenaert+07discriminating}. 

\begin{lemma}
    Let $A,B$ be positive semidefinite operators. Then, for any $\alpha \in (0,1)$, there exists a test $0 \leq T \leq I$ such that
    \begin{equation}
        \textup{Tr}[(I-T)A]-\textup{Tr}[TB] = \textup{Tr}[A^\alpha B^{1-\alpha}]
    \end{equation}
\end{lemma}

\begin{proof}
We start with
\begin{align}
&\Trm[A] - \Trm\big[A^\alpha B^{(1-\alpha)}\big]\\  
&\quad= \int_{0}^1 \frac{\dd}{\dd t} \tr \big[ A^\alpha \big(t A + (1-t)B\big)^{1-\alpha} \big]\dd t\\
&\quad = \frac{\sin((1-\alpha)\pi)}{\pi}\int_0^1 \int_0^\infty \lambda^{1-\alpha} \Trm\left[ A^\alpha \frac{1}{\lambda + tA + (1-t)B} (A-B) \frac{1}{\lambda + tA + (1-t)B}  \right] \dd t\dd \lambda.
\end{align}
Here, to compute the derivative in the third line we use the integral representation~\eqref{eq:x-alpha-integral}.
Therefore, letting
\begin{align}
T = \frac{\sin((1-\alpha)\pi)}{\pi}\int_0^1 \int_0^\infty \lambda^{1-\alpha}   \frac{1}{\lambda + tA + (1-t)B} A^\alpha \frac{1}{\lambda + tA + (1-t)B}   \dd t\dd \lambda,
\end{align}
we have $\tr[A] - \tr[A^\alpha B^{(1-\alpha)}] = \tr[T(A-B)] $ or equivalently $\tr[(I-T)A] + \tr[TB]=\tr[A^\alpha B^{1-\alpha}]$. It remains to verify that $0\leq T\leq I$. The positivity of $T$ is clear from its definition. Next, by the operator monotonicity of $x\mapsto x^\alpha$ we have  
\begin{align}
T & = \frac{\sin((1-\alpha)\pi)}{\pi}\int_0^1 \int_0^\infty t^{-\alpha} \lambda^{1-\alpha}   \frac{1}{\lambda + tA + (1-t)B} (tA)^\alpha \frac{1}{\lambda + tA + (1-t)B}   \dd t\dd \lambda \\
& \leq  \frac{\sin((1-\alpha)\pi)}{\pi}\int_0^1 \int_0^\infty t^{-\alpha} \lambda^{1-\alpha}   \frac{1}{\lambda + tA + (1-t)B} (tA + (1-t)B)^\alpha \frac{1}{\lambda + tA + (1-t)B}   \dd t\dd \lambda \\
& = \bigg( (1-\alpha)\int_0^1 t^{-\alpha}\dd t \bigg)I \\
& = I. \label{eq:T-leq-I}
\end{align}
Here, the third line is obtained using
\begin{align}\label{eq:int-rep-1-alpha-derivative}
(1-\alpha) x^{-\alpha} = \frac{\sin((1-\alpha)\pi)}{\pi}\int_0^\infty \lambda^{1-\alpha} \frac{1}{(\lambda+x)^2}\dd \lambda,
\end{align}
which can be verified by computing the derivative of both sides in~\eqref{eq:x-alpha-integral} for $\theta=1-\alpha$. 
\end{proof}

\section{Bounds on the strong converse exponent}\label{app:strong-converse}

This appendix is devoted to proving the upper and lower bounds on the strong converse exponent given in Proposition \ref{prop:sc bounds}. We start by characterizing the strong converse exponent in the commutative case. 

\begin{theorem}\label{thm:SC-commutative}
Consider the quantum hypothesis testing problem~\eqref{eq:problem} and suppose that Assumption~\ref{assumption:polar} and Assumption~\ref{assumption:support} are satisfied. Assume further that all states in $\cF_n$, for any $n\in \N$, mutually commute and also commute with $\rho^{\otimes n}$. Then, the strong converse exponent is
\begin{align}\label{eq:Hoeffding-composite sc}
\rmsc_r(\rho\| \{\cF_n\}_{n\in \N}) &= \sup_{1\leq \alpha\leq \infty}\sup_{\sigma\in \cF_1} \frac{\alpha-1}{\alpha}\big(r - \widetilde D_\alpha(\rho\|\sigma)\big),
\end{align}
where $\rmsc_r(\rho\| \{\cF_n\}_{n\in \N})$ is defined in~\eqref{eq:def-strong-coverse-exp} in terms of type I and type II errors given in~\eqref{eq:type-I-error} and~\eqref{eq:type-II-error}.
\end{theorem}

\begin{proof}
The proof of the converse is similar to that for the Chernoff and Hoeffding exponents, and we skip it. The proof of achievability is more or less similar to that of Theorem~\ref{thm:additivity-strong-converse}, so we only discuss the differences. 

Observe that by the commutativity assumption together with Theorem~\ref{thm:Stein} and Theorem~\ref{thm:Stein-additivity}, we have $\rms(\rho\|\{\cF_n\}_{n\in \N}) =  D(\rho\|\cF_1)=\inf_{\sigma\in \cF_1}D(\rho\|\sigma)$. Therefore, by the definition of the Stein exponent, if $r\leq D(\rho\|\cF_1)$ then $\rmsc_r(\rho\| \{\cF_n\}_{n\in \N})=0$. In this case, by the monotonicity of $\alpha\mapsto \widetilde D_\alpha(\rho\| \sigma)$, for any $\alpha\geq 1$ we have $r\leq \widetilde D_\alpha(\rho\|\cF_1)$, so the right-hand side of~\eqref{eq:Hoeffding-composite sc} is also equal to $0$ (for the choice of $\alpha=1$). If $D(\rho\|\cF_1)=\widetilde D_{\infty}(\rho\|\cF_1)$, then this also shows that \eqref{eq:Hoeffding-composite sc} holds for every 
$r\le \widetilde D_{\infty}(\rho\|\cF_1)$. If $D(\rho\|\cF_1)<r<\widetilde D_{\infty}(\rho\|\cF_1)$, then~\eqref{eq:Hoeffding-composite sc} follows by Theorem \ref{thm:additivity-strong-converse} due to 
$\widetilde D_{\infty}(\rho\|\cF_1)\le r_{\infty}(\rho\|\cF_1)$. 

Hence, we need to prove \eqref{eq:Hoeffding-composite sc} for 
$\widetilde D_{\infty}(\rho\|\cF_1)\le r$. For the rest, fix such an $r$.
By the above, for any $r'<  \widetilde D_{\infty}(\rho\|\cF_1)$,  there is a sequence $(T'_n)_{n\in \N}$  of tests such that 
\begin{align}
\liminf_{n\to +\infty} -\frac 1n\log \beta_n(\cF_n|T'_n)\geq r',
\end{align}
and
\begin{align}
\lim_{n\to +\infty} -\frac 1n \log(1-\alpha_n(\rho| T'_n)) = \sup_{1< \alpha<+\infty}\sup_{\sigma\in \cF_1} \frac{\alpha-1}{\alpha}\big(r'-\widetilde D_\alpha(\rho\|\sigma)\big).
\end{align}
Now, let $T_n=e^{-n(r-r')}T'_n$. We then have 
\begin{align}
\liminf_{n\to +\infty} -\frac 1n\log \beta_n(\cF_n|T_n)\geq (r-r')+r'=r,
\end{align}
and similarly
\begin{align}
\lim_{n\to +\infty} -\frac 1n \log(1-\alpha_n(\rho| T_n)) = r-r'+\sup_{1< \alpha<+\infty}\sup_{\sigma\in \cF_1} \frac{\alpha-1}{\alpha}\big(r'-\widetilde D_\alpha(\rho\|\sigma)\big).
\end{align}
We conclude that 
\begin{align}
\rmsc_r(\rho\|\{\cF_n\}_{n\in \N})\leq \inf_{r'<\widetilde D_{\infty}(\rho\|\cF_1)} r-r'+\sup_{1< \alpha<+\infty}\sup_{\sigma\in \cF_1} \frac{\alpha-1}{\alpha}\big(r'-\widetilde D_\alpha(\rho\|\sigma)\big).
\end{align}
We note that here by the monotonicity of $\alpha\mapsto \widetilde D_\alpha(\rho\|\cF_1)$, if $r'>\widetilde D_{\alpha'}(\rho\|\cF_1)$, then the supremum over $\alpha$ can be restricted to $\alpha>\alpha'$. Thus, as $r'\searrow \widetilde D_{\infty}(\rho\| \cF_1)$, the optimal $\alpha$ tends to $+\infty$. Therefore, the infimum over $r'<\widetilde D_{\infty}(\rho\|\cF_1)$ is obtained in the limit $r'\searrow \widetilde D_{\infty}(\rho\|\cF_1)$ and is equal to $r-\widetilde D_{\infty}(\rho\| \cF_1)$ as desired.

\end{proof}

We can now present our bounds on the strong converse exponent.

\begin{corollary}\label{cor:app-SC-corollary}
Consider the quantum hypothesis testing problem~\eqref{eq:problem} for some state $\rho$ and suppose that the alternative hypothesis $\{\cF_n\}_{n\in \N}$ satisfies Assumption~\ref{assumption:polar} and Assumption~\ref{assumption:support}. Then, the strong converse exponent is bounded as
\begin{align}
&\sup_{\alpha>1}\sup_{n\in\N}\sup_{\sigma^{(n)}\in\cF_n}\frac{\alpha-1}{\alpha}\Big(
r-\frac{1}{m}\widetilde{D}_{\alpha}\big(\rho^{\otimes n}\|\sigma^{\n}\big)\Big)\\
&=
\sup_{\alpha>1}\lim_{n\to+\infty}\sup_{\sigma^{\n}\in\cF_n}\frac{\alpha-1}{\alpha}\Big(
r-\frac{1}{n}\widetilde{D}_{\alpha}\big(\rho^{\otimes n}\|\sigma^{\n}\big)\Big)\label{eq:sc lower bound1}\\
& \leq \rmsc_r(\rho\|\{\cF_n\}_{n\in \N})\\
&\leq \inf_{n\in\N}\inf_{\cM_{n}}
\sup_{\alpha>1}\sup_{\sigma^{\n}\in\cF_n}\frac{\alpha-1}{\alpha}\Big(r-\frac{1}{n}D_{\alpha}\big(\cM_{n}(\rho^{\otimes n})\|
\cM_{n}(\sigma^{\n})\big)\Big), 
\label{eq:sc lower bound2}
\end{align}
where in the last line $\cM_{n}$ runs over all measurements on the $n$-fold tensor product space with finitely many POVM elements and $\rmsc_r(\rho\| \{\cF_n\}_{n\in \N})$ is defined in~\eqref{eq:def-strong-coverse-exp} in terms of type I and type II errors given in~\eqref{eq:type-I-error} and~\eqref{eq:type-II-error}.
\end{corollary}

\begin{proof}
Equality in~\eqref{eq:sc lower bound1} follows by the Fekete lemma since $n\mapsto \inf_{\sigma^{\n}\in \cF_n}\widetilde D_{\alpha}(\rho^{\otimes n}\| \sigma^{\n})$ is subadditive. Also, the first inequality is a consequence of~\eqref{eq:Strong-Converse}.  Finally, the last inequality is a consequence of Theorem~\ref{thm:SC-commutative} when we consider the hypothesis testing problem with $\rho'=\cM_n(\rho^{\otimes n})$ and $\cF'_m=\{   \cM_n^{\otimes m}(\sigma^{(mn)}):\, \sigma^{(mn)}\in \cF_{mn}    \}$.
\end{proof}

\begin{remark}
In the upper bound~\eqref{eq:sc lower bound2} we may restrict $\cM_n$ to symmetric measurements, i.e., measurements that output the same distribution if we apply a permutation of the subsystems of the input $\sigma^{\n}$. In this case, by the convexity of $ D_\alpha(\cdot\|\cdot)$ in its second argument, if we also impose Assumption~\ref{assumption:perm-inv} that $\cF_n$ is permutation invariant, we may restrict the supremum over $\cF_n$, to the subset of symmetric states in $\cF_n$ which we denote by $\cF_n^{\textup{symm}}$. Now, suppose that the former infimum and the latter supremum can be swapped. Then, we have 
\begin{align}
&\rmsc_r(\rho\|\{\cF_n\}_{n\in \N})\\
&\leq \inf_{n\in\N}\inf_{\cM_{n}}
\sup_{\alpha>1}\sup_{\sigma^{\n}\in\cF_n}\frac{\alpha-1}{\alpha}\Big(r-\frac{1}{n}D_{\alpha}\big(\cM_{n}(\rho^{\otimes n})\|
\cM_{n}(\sigma^{\n})\big)\Big)\\
&\leq \inf_{n\in\N}\inf_{\cM_{n}: \textup{symm}}
\sup_{\alpha>1}\sup_{\sigma^{\n}\in\cF_n}\frac{\alpha-1}{\alpha}\Big(r-\frac{1}{n}D_{\alpha}\big(\cM_{n}(\rho^{\otimes n})\|
\cM_{n}(\sigma^{\n})\big)\Big)\\
&= \inf_{n\in\N}\inf_{\cM_{n}: \textup{symm}}
\sup_{\alpha>1}\sup_{\sigma^{\n}\in\cF_n^{\textup{symm}}}\frac{\alpha-1}{\alpha}\Big(r-\frac{1}{n}D_{\alpha}\big(\cM_{n}(\rho^{\otimes n})\|
\cM_{n}(\sigma^{\n})\big)\Big)\\
&= \inf_{n\in\N} \sup_{\alpha>1}\sup_{\sigma^{\n}\in\cF_n^{\textup{symm}}}  \inf_{\cM_{n}: \textup{symm}}
\frac{\alpha-1}{\alpha}\Big(r-\frac{1}{n}D_{\alpha}\big(\cM_{n}(\rho^{\otimes n})\|
\cM_{n}(\sigma^{\n})\big)\Big),
\end{align}
where the last equality is our assumption. Now thinking of the pinching map as in~\cite[Lemma 2.4]{Berta+2021composite} we find that there is a polynomial $p(n)$ such that  
\begin{align}
\sup_{\cM_n} D_{\alpha}\big(\cM_{n}(\rho^{\otimes n})\|
\cM_{n}(\sigma^{\n})\big)\geq D_{\alpha}\big(\rho^{\otimes n}\|
\sigma^{\n}\big) - \log p(n),
\end{align}
for any permutation symmetric state $\sigma^{\n}$. Therefore, 
\begin{align}
&\rmsc_r(\rho\|\{\cF_n\}_{n\in \N})\\
&= \inf_{n\in\N} \sup_{\alpha>1}\sup_{\sigma^{\n}\in\cF_n^{\textup{symm}}}  
\frac{\alpha-1}{\alpha}\Big(r-\frac{1}{n}D_{\alpha}\big(\rho^{\otimes n}\|
\sigma^{\n}\big)\Big) + \frac {1}{n} \log p\n\\
&= \inf_{n\in\N} \sup_{\alpha>1}\sup_{\sigma^{\n}\in\cF_n}  
\frac{\alpha-1}{\alpha}\Big(r-\frac{1}{n}D_{\alpha}\big(\rho^{\otimes n}\|
\sigma^{\n}\big)\Big),
\end{align}
which coincides with the lower bound in Corollary~\ref{cor:app-SC-corollary}. 
\end{remark}

In the following lemma, we show that the supremum over the measurements can indeed be swapped with the infimum over $\cF_n$.\footnote{We do not, however, know that the optimization over $\alpha$ can also be swapped.}

\begin{lemma}\label{lemma:minimax for measured Renyi}
Let $\cF$ be a closed and convex subset of quantum states. Then, we have
\begin{align}
\sup_{\cM}\inf_{\sigma\in\cF}D_{\alpha}\big(\cM(\rho)\|\cM(\sigma)\big)
=
\inf_{\sigma\in\cF}\sup_{\cM}D_{\alpha}\big(\cM(\rho)\|\cM(\sigma)\big).
\end{align}
\end{lemma}

We develop some notations before giving the proof of this lemma. Let $\mathfrak{M}$ be the set of POVMs with finitely many POVM elements. That is, $\mathfrak{M}$ is the set of measurements $\cM=\{M_1, \dots, M_k\}$ such that $M_i$'s are positive semidefinite and $\sum_i M_i=I$. Here, each measurement $\cM$ is understood as a multiset since there may exist repeated measurement operators.

We denote the set of finitely supported probability distributions on $\mathfrak{M}$ by $\cP_f(\mathfrak{M})$, and for 
$p\in\cP_f(\mathfrak{M})$, we define $\cM^p\in \mathfrak{M}$ as follows. Let $\supp(p) =\{\cM^{(1)},\ldots, \cM^{(r)}\}$ be the support of $p$, i.e., the set of measurements $\cM \in \mathfrak{M}$ such that $p(\cM)>0$. Then, we let \begin{align}
\cM^p = \bigcup_{j=1}^r p(\cM^{(j)}) \cM^{(j)},
\end{align}
where for a measurement $\cM=\{M_1, \dots, M_k\}$ and $p>0$ we let $p\cM = \{pM_1, \dots, pM_k\}$.
Observe that for any $p\in \cP_f(\mathfrak{M})$, $\cM^p$ is again a POVM. Moreover, for any 
$\alpha>1$ we have
\begin{align}
&\sum_{\cM\in\mathfrak{M}}p(\cM)D_{\alpha}\big(\cM(\rho)\|\cM(\sigma)\big)\\
&=
\sum_{\cM\in\supp(p)}p(\cM)\frac{1}{\alpha-1}\log\sum_{M\in\cM}
(\Trm M\rho)^{\alpha}(\Trm M\sigma)^{1-\alpha}\\
&\le
\frac{1}{\alpha-1}\log\sum_{\cM\in\supp(p)}p(\cM)\sum_{M\in\cM}
(\Trm M\rho)^{\alpha}(\Trm M\sigma)^{1-\alpha}\\
&=
D_{\alpha}\big( \cM^{P}(\rho)\|\cM^{P}(\sigma)\big).
\label{eq:mixture inequality}
\end{align}

\begin{proof}
The proof goes the same way as that of 
\cite[Lemma 13]{Brandao+2020adversarial}. 
We have
\begin{align}
&\sup_{\cM\in\mathfrak{M}}\inf_{\sigma\in\cF}D_{\alpha}\big(\cM(\rho)\|\cM(\sigma)\big)\\
&\le
\inf_{\sigma\in\cF}\sup_{\cM\in\mathfrak{M}}D_{\alpha}\big(\cM(\rho)\|\cM(\sigma)\big)\\
&=
\inf_{\sigma\in\cF}\sup_{p\in\cP_f(\mathfrak{M})}\sum_{\cM\in\mathfrak{M}}p(\cM)D_{\alpha}\big(\cM(\rho)\|\cM(\sigma)\big)\\
&=
\sup_{p\in\cP_f(\mathfrak{M})}\inf_{\sigma\in\cF}\sum_{\cM\in\mathfrak{M}}p(\cM)D_{\alpha}\big(\cM(\rho)\|\cM(\sigma)\big)\\
&\le
\sup_{p\in\cP_f(\mathfrak{M})}\inf_{\sigma\in\cF}D_{\alpha}\big(\cM^P(\rho)\|\cM^P(\sigma)\big)\\
&\le
\sup_{\cM\in\mathfrak{M}}\inf_{\sigma\in\cF}D_{\alpha}\big(\cM(\rho)\|\cM(\sigma)\big),
\end{align}
whence all the inequalities are in fact equalities. 
The first inequality and the first equality above are trivial. 
The second inequality follows from \eqref{eq:mixture inequality}, and the last inequality is based on the fact that $\cM^p\in\mathfrak{M}$ for every $p\in \cP_f(\mathfrak{M})$. 
The second equality follows from Lemma~\ref{lem:minimax Farkas} 
due to the fact that the function
\begin{align}
(\sigma,p)\ni\cF\times\cP_f(\mathfrak{M})\mapsto\sum_{\cM\in\mathfrak{M}}p(\cM)D_{\alpha}\big(\cM(\rho)\|\cM(\sigma)\big)
\end{align}
is convex and lower semi-continuous in $\sigma$, and affine (hence concave) in $p$. 
\end{proof}

\end{document}